\documentclass[conference]{IEEEtran}
\IEEEoverridecommandlockouts
\usepackage[T1]{fontenc}

\usepackage{cite}
\usepackage{amsmath,amssymb,amsfonts}
\usepackage{algorithmic}
\usepackage{graphicx}
\usepackage{textcomp}
\usepackage{xcolor}
\usepackage{hyperref} 
\usepackage{braket}
\usepackage{subcaption}
\usepackage{cite}
\usepackage{amsmath,amssymb,amsfonts}
\usepackage{algorithm}
\usepackage{algorithmic}
\usepackage{textcomp}
\usepackage{mathtools,amsthm}
\def\BibTeX{{\rm B\kern-.05em{\sc i\kern-.025em b}\kern-.08em
    T\kern-.1667em\lower.7ex\hbox{E}\kern-.125emX}}
       \newtheorem{thm}{Theorem}
\def\BibTeX{{\rm B\kern-.05em{\sc i\kern-.025em b}\kern-.08em
    T\kern-.1667em\lower.7ex\hbox{E}\kern-.125emX}}
    
\def\BibTeX{{\rm B\kern-.05em{\sc i\kern-.025em b}\kern-.08em
    T\kern-.1667em\lower.7ex\hbox{E}\kern-.125emX}}
\begin{document}

\title{The Art of Avoiding Constraints: A Penalty-free Approach to 
Constrained Combinatorial Optimization with QAOA
\thanks{D.~Lykov and Y.~Alexeev's current affiliation is: NVIDIA Corporation, Santa Clara, CA, USA. This research was supported in part by a National Sciences and Engineering Research Council (NSERC) of Canada Collaborative Research and Training Experience (CREATE) grant on Quantum Computing, NSERC Alliance Consortium Grant entitled Quantum Software Consortium -- Exploring Distributed Quantum Solutions for Canada (QSC), and NSERC Alliance grant on Quantum Computing for Optimal Mobility. 
\\Corresponding author: pangara@uvic.ca}
}

\author{\IEEEauthorblockN{Prashanti Priya Angara, Ulrike Stege, Hausi M\"uller}
\IEEEauthorblockA{
\textit{University of Victoria}\\
Victoria, BC, Canada} 
\and
\IEEEauthorblockN{Danylo Lykov, Yuri Alexeev}
\IEEEauthorblockA{
\textit{Argonne National Laboratory (ANL)}\\
 Lemont, IL, USA}}

\maketitle

\begin{abstract}

The quantum approximate optimization algorithm (QAOA) is designed to determine optimum and near optimum solutions of quadratic (and higher order) unconstrained binary optimization (QUBO or HUBO) problems, which in turn accurately model \textit{unconstrained} combinatorial optimization problems. 
While the solution space of an unconstrained combinatorial optimization problem consists of all possible combinations of the objective function's decision variables, for a \textit{constrained} combinatorial optimization problem, the solution space consists only of solutions that satisfy the problem constraints. 
While solving such a QUBO problem optimally results in an optimal solution to the underlying constrained problem, setting the right penalty parameters is crucial. Moreover, finding suitable penalties that ensure near-optimal solutions is also challenging. 

In this article, we introduce our innovative profit-relaxation framework to solve constrained combinatorial optimization problems. Our effective hybrid approach transforms a selected constrained combinatorial optimization problem into an unconstrained \textit{profit} problem, such that the solution spaces of the two problems are interrelated with respect to their solution qualities. This allows us to determine optimal and near-optimal solutions for the constrained problem.

We use three NP-hard problems---\textit{\textsc{Minimum Vertex Cover}},\textit{\textsc{Maximum Independent Set}}, and \textit{\textsc{Maximum Clique}}---to demonstrate the feasibility and effectiveness of our approach. We conducted a detailed performance evaluation of our profit-relaxation framework on two different platforms: Xanadu PennyLane and Argonne QTensor Quantum Circuit Simulator. 

\end{abstract}
\begin{IEEEkeywords}
Quantum computing, QAOA, hybrid quantum-classical algorithms, solving constrained optimization problems, penalty-free QUBO, problem transformation, vertex cover, independent set, clique, profit cover, profit problem, quantum simulations, Xanadu PennyLane, Argonne QTensor Quantum Circuit Simulator
\end{IEEEkeywords}
\section{Introduction}
\label{sec:introduction}
A promising path forward in computing involves combining quantum and classical computing approaches. Quantum processing units (QPUs) will play a major role in accelerating algorithms in optimization, simulation, and machine learning applications beyond what classical computers can achieve on their own, similar to how graphical processing units (GPUs) have become integral to general-purpose computing~\cite{NVIDIAcudaq2024}. 

The effective integration of classical and quantum hardware and software architectures will play a crucial role in achieving quantum utility~\cite{qutility23} and quantum advantage~\cite{shaydulin2024evidence}. 
Variational quantum algorithms (VQAs) are the result of one of the most significant developments in hybrid quantum-classical techniques. These techniques can be used to solve practical real-world challenges~\cite{farhi2014quantum, peruzzo2014variational}such as NP-hard problems. The Quantum Approximate Optimization Algorithm (QAOA)~\cite{farhi2014quantum} is a variational quantum algorithm specifically designed to solve \textit{unconstrained} combinatorial optimization problems $P_U$'s. Solving $P_U$'s optimization problems involves finding a best solution from a large (but finite) set of possibilities (e.g., from all possible subsets of elements of the problem input). Note that the solution space of such $P_U$ has only feasible solutions, and can be described as QUBO (Quadratic Unconstrained Binary Optimization) or HUBO (Higher-order Unconstrained Binary Optimization) problems. 

In contrast, \textit{constrained} combinatorial optimization problems $P_C$'s refer to those combinatorial optimization problems that consist of constraints limiting the set of feasible solutions (or the solution space).

Some of the key methods for modeling constraints in a QUBO or HUBO framework include:
\begin{enumerate}
    \item introducing \textit{penalty parameters} to discourage solutions that violate the constraints during the optimization process~\cite{lucas2014ising}, or
    \item designing QAOA circuits that operate within a \textit{feasible subspace}~\cite{hadfield2019quantum} only.
\end{enumerate}

\noindent Both methods come with their specific challenges: method 1.~requires the determination of (input specific) penalty parameters and allows infeasible solutions (thereby making the solution space larger) while method 2.~requires more complex quantum circuitry to restrict the solution space to only feasible solutions.

In this paper, we present our findings exploring effective methods for solving \textit {constrained} combinatorial optimization problems with QAOA by formulating \textit{problem relaxations} of the constrained problems. Due to the unconstrained nature of the problem relaxations, neither penalty parameters nor complex quantum circuitry are required to enforce constraints. We define a problem relaxation as a variant of the constrained problem, which shares underlying structures and characteristics with the original problem while enforcing relationships between the solution landscapes. Our framework using problem relaxations allows to not only effectively support the effective finding of optimal solutions to the constrained problem, but also to capture the near-optimal solutions. There is value in identifying near-optimal solutions to optimization problems, rather than just the optimal solutions, in situations where trade-offs between accuracy, efficiency, and practicality matters~\cite{zaborniak2024discretequadraticmodelqubo, mohseni2023sampling}. Near-optimal solutions are valuable when developing scalable quantum computing optimization methods in general, and in particular when using methods such as QAOA due to its inexact and probabilistic nature. Leveraging the relationships between the solution landscapes of the constrained problem and the problem relaxation, we apply classical post-processing techniques to the solution obtained by performing QAOA on the problem relaxation. This approach enables us to solve the original constrained problem.  
Fig.~\ref{fig:cp-pu-flowchart} shows the workflow of solving a constrained optimization problem, $P_C$. Our framework involves relaxing the constraints of $P_C$ to obtain a problem relaxation $P_U$ of $P_C$; $P_U$ is an unconstrained optimization problem. It allows solving $P_C$ by first solving its $P_U$, and then uses an efficient classical post-processing procedure applied to obtained solutions for $P_U$, to obtain solutions to $P_C$ while maintaining the solution quality.

\begin{figure}
    \centering
    \includegraphics[width=1\linewidth]{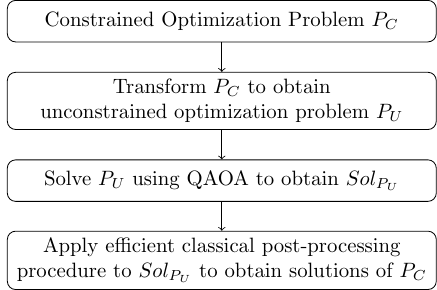}
    \caption{Workflow of solving constrained optimization problems using the our problem relaxation framework}
    \label{fig:cp-pu-flowchart}
\end{figure}

By evaluating QAOA's performance on problem relaxations, we  contribute to a broader understanding of how classical post-processing can be effectively harnessed in combination with hybrid quantum-classical techniques to obtain better solution strategies for both unconstrained and constrained combinatorial optimization problems. 

The main contributions of this paper include: (1) a problem relaxation framework applied to a set of related constrained combinatorial optimization problems from graph theory, all of which are NP-hard~\cite{garey1979computers}: \textit{\textsc{Minimum Vertex Cover}}, \textit{\textsc{Maximum Independent Set}}, and \textit{\textsc{Maximum Clique}}, (2) a performance evaluation of our problem relaxation framework using the full statevector simulators of Xanadu PennyLane,  (3) an assessment of scalability of our framework on large problems using the Argonne QTensor tensor network simulator,

and (4) an extension of the Argonne QTensor framework to include formulations of constrained combinatorial optimization problems.

The paper is organized as follows. Section~\ref{sec:constrained_problems} introduces our use cases that we demonstrate our framework on, i.e., the constrained combinatorial optimization problems {\sc Minimum Vertex Cover}, {\sc Maximum Independent Set} and {\sc Maximum Clique}.  Section~\ref{sec:background} reviews the relevant background and literature for this research. It offers a comprehensive analysis of QAOA techniques, outlines its fundamental principles, and explores recent developments for tackling constrained optimization problems. 
Section~\ref{sec:penalty-study} reviews existing approaches to formulating QUBOs for the aforementioned constrained combinatorial optimization problems. 
Section~\ref{sec:relaxations} introduces the problem relaxations and their QUBOs: {\sc Maximum Profit Cover} of {\sc Minimum Vertex Cover}, {\sc Maximum Profit Independence} of {\sc Maximum Independent Set}, and, {\sc Maximum Profit Clique} of {\sc Maximum Clique}. Section~\ref{sec:methodology} describes the experimental setup, including parameters and graph data sets employed for assessments on the Xanadu PennyLane and Argonne QTensor Simulator platforms. 
Section~\ref{sec:results} describes the evaluation metrics used and presents our results.
Section~\ref{sec:discussion} discusses the findings and results of our investigations. For smaller input graphs, we provide a full state vector simulation using PennyLane. For larger graphs, we showcase the scalability of our approach using the Argonne QTensor Simulator. Section~\ref{sec:future} concludes the paper, and additional background and results are provided in the Appendix (Section~\ref{sec:appendix}).
\section{Constrained Combinatorial Optimization Problems of Study}\label{sec:constrained_problems}

Although there is no expectation of solving NP-hard problems in polynomial time quantumly, intractable combinatorial optimization problems are excellent candidates to demonstrate practical quantum utility and advantage~\cite{brooks2019beyond, farhi2016quantum}. 
We now formally introduce the problems of our study, which are known to be NP-hard~\cite{garey1979computers}. All graphs considered in this article are undirected and simple.
\subsection*{\textit{\textsc{Minimum Vertex Cover}}}
We define the problem \textit{\textsc{Minimum Vertex Cover}} (\textit{\textsc{MinVC}}) for a graph $G = (V,E)$, where $\textit{VC}\subseteq V$, as follows.
\begin{align*}
&\text{Minimize:} \quad |\textit{VC}| \\
&\text{Subject to:} \quad \mbox{for all edges}~uv \in E, \quad u \in \textit{VC} \ \text{or} \ v \in \textit{VC}
\end{align*}

Any $\textit{VC} \subseteq V$ that satisfies this constraint is a vertex cover. Any $\textit{VC}$ that optimally satisfies the objective function is a \textit{ minimum vertex cover}.

\subsection*{\textit{\textsc{Maximum Independent Set}}}
We define the problem \textit{\textsc{Maximum Independent Set}} (\textit{\textsc{MaxIS}}) for a graph $G = (V,E)$, where $\textit{IS}\subseteq V$, as follows.
\begin{align*}
&\text{Maximize:} \quad |\textit{IS}| \\
&\text{Subject to:} \quad  \mbox{for all vertices}~u,v \in \textit{IS}, \quad uv \notin E
\end{align*}
We denote any subset of $V$ that satisfies the constraint an \textit{independent set}. Any largest independent set is also referred to as \textit{maximum independent set}.

\subsection*{\textit{\textsc{Maximum Clique}}}
We define the problem \textit{\textsc{Maximum Clique}} (\textit{\textsc{MaxCl}}) for a graph $G = (V,E)$, where $\textit{Cl}\subseteq V$, as follows.
\begin{align*}
&\text{Maximize:} \quad |\textit{Cl}| \\
&\text{Subject to:} \quad \mbox{for all vertices}~u,v \in \textit{Cl}, \quad uv \in E
\end{align*}
We denote any subset of $V$ that satisfies the constraint a \textit{clique}. Any largest clique is also referred to as \textit{maximum maximum clique}.
\begin{figure*}[!th]
    \centering
    \includegraphics[width=1\linewidth]{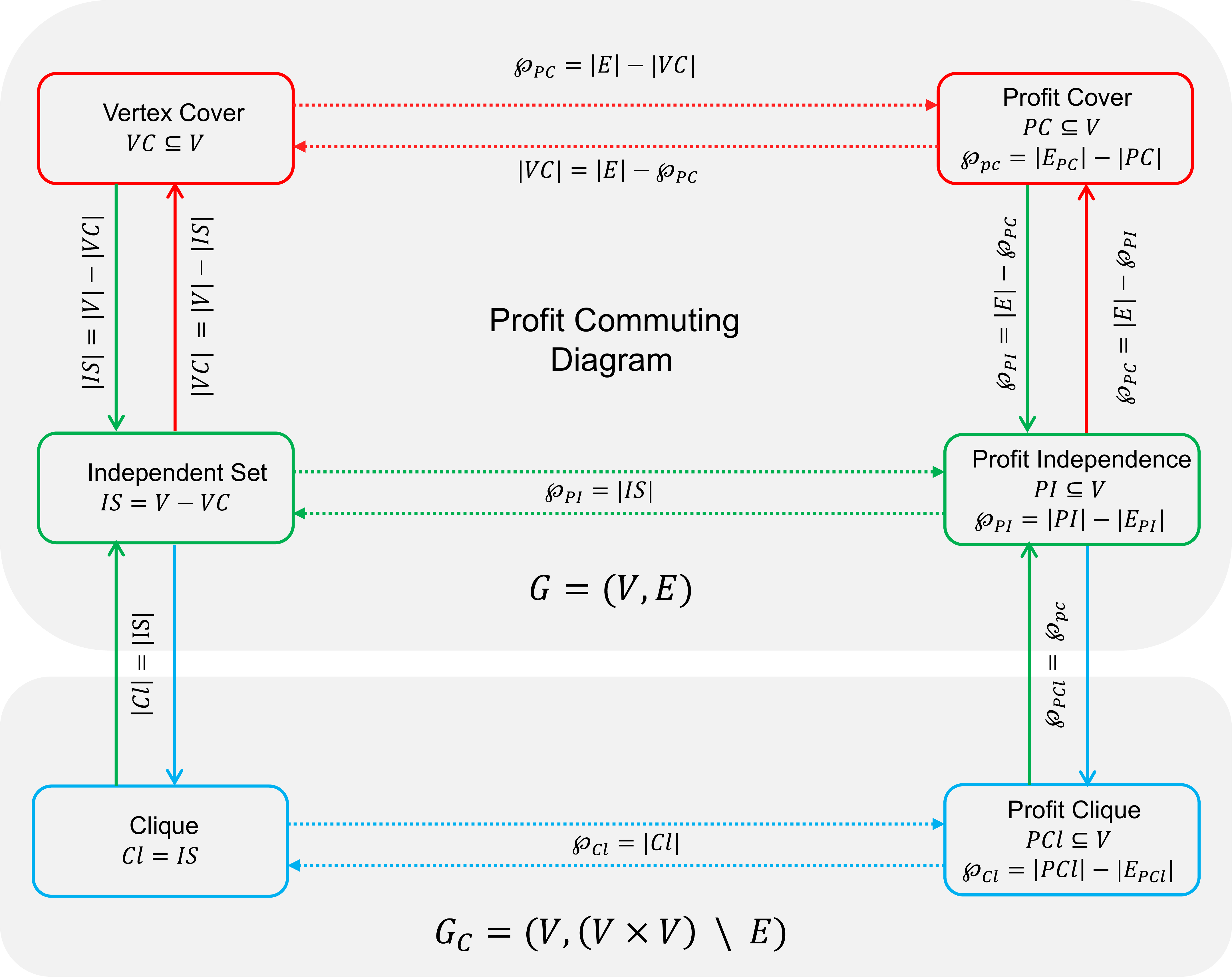}
    \caption{\textbf{Profit Commuting Diagram:} Relationship between constrained and profit (unconstrained) variants of vertex cover, independent set, and clique. Note that $\textit{VC}/\textit{PC}$ and $\textit{IS}/\textit{PI}$ correspond to the vertex/profit cover and independent set/profit independence of graph $G$ whereas $\textit{Cl}/\textit{PCl}$ refer to Clique/Profit Clique for the complement graph $G_c$. }
    \label{fig:relationship}
\end{figure*}
\subsection*{Relationships between \textit{\textsc{MinVC}}, \textit{\textsc{MaxIS}} and \textit{\textsc{MaxCl}}}
\label{sec:relationships}
All three problems are known to be NP-hard. Reviewing the relationship between the three problems shows that it permits an efficient translation of solutions for instances between the problems (cf. Fig.~\ref{fig:relationship}). Therefore, no matter which of \textit{\textsc{MinVC}}, \textit{\textsc{MaxIS}} and \textit{\textsc{MaxCl}} is the problem under consideration, it is of interest to explore algorithmic approaches for the three problems.

\textit{\textsc{MinVC}}, \textit{\textsc{MaxIS}} and \textit{\textsc{MaxCl}} are equivalent in the sense that: for any graph $G = (V,E)$ and its \textit{complement graph} $G_c = (V,E_c)$ where $E_c = (V\times V)\setminus E$, $\textit{VC} \subseteq V$ is a vertex cover for $G$ if and only if $V\setminus \textit{VC}$ is an independent set for $G$ if and only if $V\setminus \textit{VC}$ is a clique for $G_c$~\cite{garey1979computers}.\footnote{The "$\setminus$" symbol denotes the setminus operation.}
\section{Background}\label{sec:background}
This section outlines the background relevant for our paper. A discussion on the relevant classical results can be found in the Appendix (cf. Section~\ref{sec:classical-results}). 

\subsection{Variational Quantum Algorithms}
Significant advancements have been made in creating near-term hybrid algorithms that operate without the need for fault-tolerant machines~\cite{cerezo2021variational}, with VQAs being among the most notable examples~\cite{mcclean2016theory}. 
To understand our framework, we first describe how optimization problems are modeled within the QAOA framework.

In the realm of quantum mechanics, Hamiltonian operators are Hermitian matrices that describe the total energy of a quantum system. Its eigenvalues reveal the energy values that can be observed within the system. Specifically, the lowest eigenvalues correspond to its ground-state energies. In general, determining these ground-state energies for a given Hamiltonian is NP-hard~\cite{barahona1982computational}. 

The variational paradigm is a fundamental concept used to estimate the ground states of quantum systems~\cite{griffithsqmec1995}. A VQA is a computational method that leverages both classical and quantum resources to estimate the solution of an optimization problem. The goal is to find the best solution from a set of possible options using integrated quantum-classical techniques. VQAs use parameterized quantum circuits to estimate the optimum of a function and uses a classical optimizer to adjust the QAOA parameters iteratively. To achieve reliable results with a quantum computer, the depth and design of these parameterized quantum circuits must be tailored to the quantum computing hardware and the software architecture~\cite{kandala2017hardware, ravi2022cafqa}, and to specific problem to be solved~\cite{anselmetti2021local, hadfield2019quantum}.  The most famous VQA techniques are the variational quantum eigensolver (VQE)~\cite{peruzzo2014variational} and the quantum approximate optimization algorithm (QAOA)~\cite{farhi2014quantum}. These techniques are workhorses of quantum computing that integrate classical and quantum resources to solve various problems in quantum chemistry~\cite{o2016scalable, RevModPhys.92.015003}, machine learning~\cite{larose2019variational, havlivcek2019supervised}, and combinatorial optimization~\cite{harrigan2021quantum, lotshaw2022scaling}.
VQE and QAOA are hybrid quantum-classical techniques that can be used to estimate ground-state energies and corresponding ground states of Hamiltonians.
In the case of combinatorial optimization, substantial research has been devoted to exploring the potential of employing quantum adiabatic algorithms (QAA) on quantum annealers to  address challenges when solving NP-hard  optimization problems---with impressive developments and achievements by D-Wave Systems~\cite{johnson2011quantum, king2024computational}. The QAOA objective and structure is similar to the QAA technique---both seek to identify the lowest eigenvalue and its associated eigenvector. Moreover, the utilization of two Hamiltonians, a \textit{simple} Hamiltonian whose ground state is known, and a \textit{problem} or \textit{cost} Hamiltonian whose ground state we want to find, in QAA mirrors that of QAOA. However, QAA may be computationally expensive in the time it takes to evolve and is non-monotonic in its success probability. With optimal parameters, QAOA performance improves with the circuit depth (i.e., with the number of QAOA layers)~\cite{farhi2014quantum}. 

\subsection{Quantum Approximate Optimization Algorithm}
\noindent QAOA, introduced by Farhi et al.~in 2014, is a hybrid quantum-classical algorithm that mimics adiabatic quantum computation on near-term gate-based quantum computers aiming to solve combinatorial optimization problems. Farhi's QAOA, often referred to as \textit{vanilla} QAOA, has inspired numerous variants; selected variants are described in the Appendix (cf. Section~\ref{sec:related-work}).
The key building blocks of a vanilla QAOA setup include:
\begin{enumerate}
    \item An \textbf{initial state} $\vert \psi_0 \rangle$ in an equal superposition of all computational basis states, prepared by applying Hadamard gates to qubits initialized in state $\ket{0}^{\otimes n}$.
    \item \textbf{An ansatz} consisting of:
    \begin{itemize}
            
        \item \textbf{A cost or phase separator unitary} that encodes the cost function $C(\vec{x})$ of the optimization problem being solved by introducing relative phase shifts between the states based on their cost function values. For a minimization problem, the objective is to identify an assignment of variables that yields the smallest possible value of the cost function $C(\vec{x})$. This unitary is generated by the \textit{cost} Hamiltonian ($\hat{H}_C$) and for combinatorial optimization problems, it typically consists of sums of Pauli-$Z$ and $ZZ$ terms. 
        \item A \textbf{mixing unitary} that anti-commutes with the cost unitary and enables transitions between the computational basis states. This unitary is generated by the \textit{mixer} Hamiltonian ($\hat{H}_M$), commonly chosen as a sum of Pauli-$X$ operators.
         \item The number of \textbf{layers $p$}, i.e., the number of times the cost and mixer unitaries are applied alternatively. With increasing alternating layers, QAOA progressively mimics QAA through Trotterization~\cite{farhi2014quantum}.
        \item  \textbf{Parameter vectors}, $\vec{\beta}$ and $\vec{\gamma}$, comprising variational parameters that control the cost and mixer unitaries, respectively, across the $p$ alternating layers. 
        \end{itemize}

    \item A \textbf{classical optimizer} responsible for finding the optimal values of the parameter vectors $\vec{\beta}$ and $\vec{\gamma}$ such that the expectation value of the cost function is minimized (for a minimization problem):
    $$
       \min_{\vec{\beta}, \vec{\gamma}} ~\langle \psi(\vec{\beta}, \vec{\gamma}) | \hat{H}_C | \psi(\vec{\beta}, \vec{\gamma}) \rangle
        $$
    where $|\psi(\vec{\beta}, \vec{\gamma}) \rangle$ is the state prepared by the alternating ansatz with $p$ layers of cost and mixer unitaries.
\end{enumerate}

\noindent The quantum part of QAOA's hybrid routine evaluates the objective function. It alternates between unitaries corresponding to cost Hamiltonian $\hat{H}_C$ and mixer Hamiltonian $\hat{H}_M$, $p$ times.
$$\vert \psi(\vec{\beta}, \vec{\gamma}) \rangle = \underbrace{U({\beta_p}) U({\gamma_p}) \cdots U({\beta_1}) U({\gamma_1})}_{p \; \text{times}} 
\vert \psi_0 \rangle$$

\noindent where $U(\boldsymbol{\beta}) = e^{-i \boldsymbol{\beta} \hat{H}_M}$ and $U(\boldsymbol{\gamma}) = e^{-i \boldsymbol{\gamma} \hat{H}_C}$ characterized by the parameters $(\beta, \gamma)$. The goal of QAOA is to determine the optimal parameters $(\boldsymbol{\beta}_{\text{opt}}, \boldsymbol{\gamma}_{\text{opt}})$ such that the quantum state $|\psi(\boldsymbol{\beta}_{\text{opt}}, \boldsymbol{\gamma}_{\text{opt}})\rangle$ encodes the solution to the problem. 
Classical optimizers that aim to determine optimal parameters include  L-BFGS, Adam, COBYLA, and Gradient Descent~\cite{pellowoptimizers2021}.  Choosing an appropriate classical optimizer is influenced by several factors, including the number of parameters, the problem, and the cost function landscape~\cite{MALAN2013148, bonet2023performance, sung2020using}.

\subsection{Constrained and Unconstrained Optimization Problems}

Quantum utility is the ability of a quantum computer to perform reliable computations at a scale beyond brute force classical computing methods that provide exact solutions to computational problems~\cite{ibm_quantum_utility}. NP-hard combinatorial optimization problems are a class of problems that have the potential to demonstrate quantum utility~\cite{qutility23}. 


A combinatorial optimization problem can be unconstrained or constrained. Unconstrained problems consider all combinations of binary variables as feasible solutions. The \textit{\textsc{MaxCut}} problem is an example of unconstrained optimization, where the objective is to partition a graph's vertices into two disjoint subsets maximizing the sum of weights of edges between them. Every variable assignment corresponds to a cut, making all assignments feasible solutions. The QUBO framework can be used to represent various combinatorial optimization problems in an unconstrained manner. Many \textit{unconstrained} problems can be represented as QUBO problems using quadratic and linear terms involving binary variables.
The Ising model is a mathematical model of ferromagnetism~\cite{peierls1936ising}, whose basic element, the \textit{spin}, can take one of two values, up or down, typically represented as +1 or -1. Spins are arranged on a lattice, which can be in one, two, or more dimensions. The Ising model provides a natural framework for formulating QUBO problems on a quantum computer using the transformation $x_i = (1 - s_i)/2$, where the spins $s_i \in \{1, -1\}$. This maps the QUBO problem to an Ising Hamiltonian minimization problem, which is used as the cost Hamiltonian in QAOA~\cite{lucas2014ising}. 

Many problems of practical interest are constrained optimization problems (i.e., not all assignments of binary variables are feasible). 
There are different approaches to embed such a problem into a QUBO. The most common approach is to use penalty terms to penalize infeasible solutions.
Penalty function methods constitute a simple approach to tackling constrained optimization problems by transforming them into unconstrained optimization problems, allowing the use of algorithms explicitly designed for unconstrained problems. With penalties, infeasible solutions ideally come at a (sufficiently high) cost, discouraging their exploration in the solution landscape. However, there are limitations to this approach, such as the need to choose appropriate penalty parameters and potentially convergence issues~\cite{smith1997penalty}. 

These limitations extend to QAOA when addressing constrained optimization problems. There is a relatively high probability of obtaining infeasible solutions that violate the constraints. Beyond incorporating penalty terms into the objective function, several other considerations arise for QAOA. These include determining the optimal number of alternating cost and mixer layers, selecting an appropriate classical optimizer, defining the optimal number of steps and the optimum step size for the optimizer, and addressing the issue of barren plateaus~\cite{safro2022, huembeli2021characterizing}.

These challenges collectively contribute to the difficulty of solving constrained combinatorial optimization problems using QAOA. The subsequent section outlines selected variants that tackle constrained optimization using QAOA.

\label{sec:penalty-study}
The QUBO problem provides a framework for representing various combinatorial optimization problems in an unconstrained manner. Many \textit{unconstrained} problems can be represented as QUBO problems using quadratic and linear terms involving binary variables. For constrained problems, penalty parameters are introduced to accommodate the constraints. Assuming a minimization routine, they add a cost to solutions that violate the constraints, steering the algorithm towards feasible solutions. Solving constrained problems with QAOA requires fine-tuning of penalty parameters in addition to optimizing the variational parameters, to be able to maneuver the cost function landscape. The choice and strength of penalty parameters can significantly impact QAOA performance.

The cost function for \textit{\textsc{MinVC}} in the form of a QUBO is defined as per the specifications outlined in~\cite{lucas2014ising}. Given $\textit{VC}\subseteq V$, and $v\in V$, let $x_v$ be a binary variable whose value is $1$ if $v$ is included in the vertex cover $\textit{VC}$ ($v\in \textit{VC}$), and $0$ otherwise. The constraint that every edge $uv\in E$ has at least one of its vertices in the subset $\textit{VC}$ is encoded by 
$$C_E^{VC}(\vec{x}) = A \sum_{uv\in E} (1 - x_u)(1-x_v).$$

Given edge $uv\in E$, if both $x_u$ and $x_v$ are assigned to be $0$, then $uv$ is uncovered and therefore $uv$ violates the constraint that $\textit{VC}$ is a vertex cover (and therefore adds a penalty of $A$ to the cost function). A trivial vertex cover arises when all $x_v \in V$ are set to $1$. The objective of minimizing the number of vertices in the subset leads to 
$$C_V^{VC}(\vec{x})= B \sum_{v} x_v.$$

The total cost for the \textit{\textsc{MinVC}} problem is
\begin{equation}
\label{eq:vcqubo}
C_{\text{VC}}(\vec{x}) = C_E^{VC}(\vec{x}) + C_V^{VC}(\vec{x}).
\end{equation}
Note that $A$ and $B$ are penalty parameters, with $0 < B < A$~\cite{lucas2014ising}.

Similarly, the cost function for \textit{\textsc{MaxIS}} is defined as follows. Given $S\subseteq V$, and $v\in V$, let $x_v$ be a binary variable whose value is $1$ if $v$ is included the independent set $S$ ($v\in S$), and $0$ otherwise. 

The objective of maximizing the number of vertices in the subset leads to 
$$C_V^{IS}(\vec{x})= B \sum_{v} x_v.$$

The constraint that there does not exist an edge $(u, v) \in E'$, with both  $u, v \in S$
$$C_E^{IS}(\vec{x}) = -A \sum_{uv\in E} x_ux_v.$$

Given edge $uv\in E$, if both $x_u$ and $x_v$ are assigned to be $1$, then the edge constraint is violated (and therefore adds a penalty of $-A$ to the cost function). A trivial independent set arises when any one single node in $x_v$ is set to $1$.

The total cost for the \textit{\textsc{MaxIS}} problem is
\begin{equation}
\label{eq:isqubo}
C_{\text{IS}}(\vec{x}) = C_E^{IS}(\vec{x}) + C_V^{IS}(\vec{x}).
\end{equation}
Note that this is formulated as a maximization problem. To convert it into a minimization problem, one can change the sign of the objective function, which allows us to maximize the original function indirectly.

Due to the relationship between independent set and clique, the best QUBO formulation for \textit{\textsc{MaxCl}} is derived due to its relationship with \textit{\textsc{MaxIS}}:\footnote{recall:  $E_c$ denotes the set of edges in the complement graph $G_c$ of $G$}
\begin{align*}
    C_{E_c}^{Cl}(\vec{x}) &= -A \sum_{uv\in (V \times V) \setminus E} x_ux_v \\
C_V^{Cl}(\vec{x}) &= B \sum_{v} x_v
\end{align*}
The total cost for  \textit{\textsc{MaxCl}}  is
\begin{equation}
\label{eq:mclqubo}
C_{\text{MCl}}(\vec{x}) = C_{E_c}^{Cl}(\vec{x}) + C_V^{Cl}(\vec{x}).
\end{equation}

For all three problems, penalty parameters $A$ and $B$, $0 < B < A$,  control the enforcement of constraints. We can consider this as one single penalty value $\lambda = A/B$. Having smaller values for $\lambda$ allows for a better exploration of the solution space and a smoother landscape by maintaining a balance between the cost and mixer Hamiltonians. However, it may not effectively enforce constraints, leading to infeasible solutions and slower convergence. Larger values of $\lambda$ will enforce constraints more effectively and lead to faster convergence, however, this may lead to an imbalance of the scales of the cost and mixer Hamiltonian which leads to an ineffective exploration of the solution space~\cite{roch2023effect}. The optimal penalty strength often requires careful tuning and varies depending on the specific problem instance.

\subsubsection*{Near-optimal solutions}
\begin{figure}[!th]
    \centering
    \begin{subfigure}[b]{0.30\textwidth} 
        \centering
        \includegraphics[width=\textwidth]{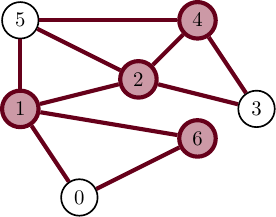}
        \caption{A feasible and optimal vertex cover}
        \label{fig:penalty-example}
    \end{subfigure}
    \hfill
    \begin{subfigure}[b]{0.30\textwidth} 
        \centering
        \includegraphics[width=\textwidth]{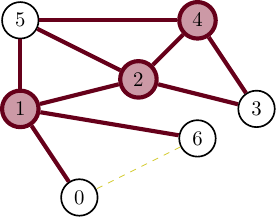}
        \caption{An infeasible vertex cover}
        \label{fig:penalty-infeasible}
    \end{subfigure}
    \caption{a) A vertex cover with a calculated cost of 8 (using $\lambda = 1.5$). b) The second-best solution, which is infeasible. For this particular graph, a higher penalty is needed to ensure vertex cover constraints are met.}
    \label{fig:penalty-demo}
\end{figure}

The penalty strength can affect the number of near-optimal solutions (i.e., solutions that are not optimal but may be of interest).
Consider Eq.~\ref{eq:vcqubo} for penalty parameters to be $A=3$ and $B=2$ (i.e.,~$\lambda=1.5$). The cost of the feasible and optimal vertex cover (in Fig.~\ref{fig:penalty-example}) is $8$. However, for this choice of penalty parameters, the cost of the second best solution is an infeasible solution---depicted in Fig.~\ref{fig:penalty-infeasible}---with a cost of 9. Notably, there exists no feasible vertex cover with a cost between 8 and 9. This is because adding any additional vertex to the optimal cover would increase the cost by at least 2, resulting in a minimum cost of 10 for the next best feasible solution. This shows that there exists an additional tier of challenges when it comes to satisfactorily represent near-optimal solutions using penalties. 
We offer an innovative profit-based approach to transform constrained optimization problems to QUBO without using penalty parameters. 

\section{Our profit-based approach to solving constrained optimization problems with QAOA}\label{sec:relaxations}

For each problem, we  state the constrained and profit versions, as well as their respective cost operators in the form of a QUBO, and their relationships. 
This section introduces our approach to penalty-free problem variants for solving the 
\textit{\textsc{MinVC}}, \textit{\textsc{MaxIS}}, and \textit{\textsc{MaxCl}} problems using QAOA. Our relaxation approach is based on the \textit{\textsc{Maximum Profit Cover}} problem, introduced in~\cite{stege2002}.  


Fig.~\ref{fig:relationship} below illustrates the connections between the constrained and profit (unconstrained) variants of Vertex Cover, Independent Set, and Clique for a given graph $G$. This diagram visualizes how these problems relate to one another and to their profit equivalents. In the subsequent sections, we provide formal and detailed descriptions of these relationships.
\subsection{\textit{\textsc{Maximum Profit Cover}}}
\label{subsec:pc}
Consider the \textit{\textsc{MinVC}} problem. Suppose we relax the requirement that the subset to be determined must be a vertex cover: If we consider the \textit{extent} to which this subset \textit{covers edges} and its \textit{closeness} to being a vertex cover (i.e., the more edge coverage the better), then we obtain the graph problem {\sc Maximum Profit Cover}~\cite{stege2002} or \textit{\textsc{MaxPC}}. For a subset $\textit{PC}\subseteq V$ in $G$, the number of edges covered by vertices from $\textit{PC}$, i.e., the number of edges that have at least one endpoint in $\textit{PC}$, is the \textit{gain}, and the number of vertices spent to cover these edges, i.e., the size of $\textit{PC}$, is considered the \textit{loss}. The profit of $\textit{PC} \subseteq V$ for a graph $G=(V,E)$ is then  defined as  
$\textit{profit} = \text{gain} - \text{loss}$.

We define the problem \textit{\textsc{MaxPC}} for a graph $G=(V, E)$ where $\textit{PC} \subseteq V$.

Maximize: profit $\mathfrak p_{\text{PC}}$, where 
\begin{align*}
{\mathfrak p}_{\text{PC}} = |E_{\text{PC}}(G, \textit{PC})| - |\textit{PC}|
\end{align*}

\noindent Here, $E_{\text{PC}}(G, \textit{PC})$ represents the edges in $G$ covered by $\textit{PC}$:
\begin{align*}
E_{\text{PC}}(G, \textit{PC}) = \{uv \in E: u \in \textit{PC} \vee v \in \textit{PC}\}
\end{align*}

\noindent We call a subset $\textit{PC}$ with maximum profit  also a \textit{maximum profit cover}. Like \textit{\textsc{MinVC}}, \textit{\textsc{MaxPC}}  is  NP-hard~\cite{stege2002}.

\subsubsection*{Cost function for \textit{\textsc{Maximum Profit Cover}}}

The binary variables for the \textit{\textsc{MaxPC}} problem are similar to \textit{\textsc{MinVC}}. Each $x_v$ is a binary variable with value  $1$ if $v$ is included in the profit cover $\textit{PC}$, and $0$ otherwise.
The edge and vertex cost functions for \textit{\textsc{MaxPC}} are given as follows:

    \textbf{Edge cost:}\\
    $$C_E^{PC}(\vec{x}) = \sum_{uv \in E} (x_u + x_v - x_ux_v)$$ 
    
    \textbf{Vertex cost:}\\
    $$C_V^{PC}(\vec{x}) =  \sum_{v} x_v$$

The total cost to be maximized for \textit{\textsc{MaxPC}} is
    \begin{equation}
\label{eq:pcqubo}
 C_{\text{PC}}(\vec{x}) = C_E^{PC}(\vec{x}) - C_V^{PC}(\vec{x})
 \end{equation}

\noindent While the definitions for $C_{\text{VC}}(\vec{x})$ and $C_{\text{PC}}(\vec{x})$ may appear similar, $C_{\text{PC}}(\vec{x})$ has no penalty parameters that need to be set, since every binary variable assignment corresponds to a feasible solution. 

\subsubsection{Finding minimum vertex covers via maximum profit covers}
We next point out an important relationship between the decision versions of both problems, which guarantees that we can derive a vertex cover $\textit{VC}$ for any subset of profit $p$ while guaranteeing the profit of the obtained vertex cover $\textit{VC}$.

\begin{thm}~\cite{stege2002}\label{thm:equiv}
For any graph $G =(V,E)$, $G$ has a vertex cover $\textit{VC}\subseteq V$ of size $k$ if and only if $G$ has a subset $\textit{PC}\subseteq V$ with profit ${\mathfrak p}_{\text{PC}}=|E| - k$. 
\end{thm}

\begin{proof} (Sketch) On the one hand, determining the profit ${\mathfrak p}_{\text{pc}}$ of a vertex cover $\textit{VC} \subseteq V$ for $G$ results in ${\mathfrak p}_{\text{PC}}=|E_{\text{ADJ}}(G,\textit{VC})| - |\textit{VC}|$, where $E_{ADJ}(G, \textit{VC})$ are the edges covered by $\textit{VC}$. Since $\textit{VC}$ is a vertex cover, $E_{\text{ADJ}}(G, \textit{VC}) =  E$ and therefore ${\mathfrak p}_{\text{PC}}= |E| - k$.

On the other hand, if subset $\textit{PC} \subseteq V$ has a profit ${\mathfrak p}_{\text{PC}}$ for $G$, then, since $|E_{\text{ADJ}}(G, \textit{PC})| - |\textit{PC}|$, $G$ as a vertex cover of size $k = |E| - {\mathfrak p}_{\text{PC}}$ for the case that $E_{\text{ADJ}}(G, \textit{PC}) = E$. What if $E_{\text{ADJ}}(G, \textit{PC}) \not= E?$ In this case, $E_{\text{ADJ}}(G, \textit{PC}) \subseteq E$, and not all edges in $G$ are covered by the vertices in $\textit{PC}$. In this case we can obtain a vertex cover of size at most $|E| - {\mathfrak p}_{\text{PC}}$ by covering the remaining edges in $E\setminus E_{\text{ADJ}}(G, \textit{PC})$ as follows. 

(*) Pick an edge $uv \in E\setminus E_{\text{ADJ}}(G, \textit{PC})$. To cover  edge $uv$, add one of the vertices, say $u$, to $\textit{PC}$. 

Note that the profit of the update set $\textit{PC}$ remains ${\mathfrak p}_{\text{PC}}$ since one more edge, $uv$, is covered with one additional vertex, $u$. 

We repeat (*) as long as $E\setminus E_{\text{ADJ}}(G, \textit{PC})\not=\emptyset$.
\end{proof}

\noindent Theorem~\ref{thm:equiv} implies that, given a maximum profit cover $\textit{PC}$ for a graph $G$, we can obtain a minimum vertex cover for $G$ using the procedure given in the proof.  
To convert profit cover results to vertex cover, we apply Algorithm~\ref{alg:add_vertices} based on Theorem~\ref{thm:equiv}. 
\begin{algorithm}
\caption{Converting Profit Cover to Vertex Cover}
\label{alg:add_vertices}
\begin{algorithmic}

\REQUIRE Graph $G$, Solution $\textit{PC}$
\STATE Let $E_{\text{unc}}$ be the set of uncovered edges in $G$
\FOR{each edge $uv \in E_{\text{unc}}$}
    \IF{$u \notin \textit{PC}$ and $v \notin \textit{PC}$}
        \STATE Add $u$ to $\textit{PC}$
    \ENDIF
\ENDFOR
\end{algorithmic}
\end{algorithm}

\begin{figure}[!htbp]
    \centering
    \begin{subfigure}{0.2\textwidth}
        \centering 
      
        \includegraphics[width=\textwidth]{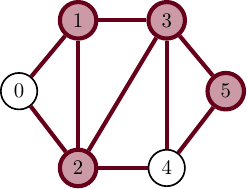}
        \caption{$\textit{PC} = \{1, 2, 3, 5\}$, \\ $\text{profit} = 5$}
        \label{subfig:pc-a}
    \end{subfigure}
    \hspace{0.05\textwidth}
    \begin{subfigure}{0.2\textwidth}
        \centering
        \includegraphics[width=\textwidth]{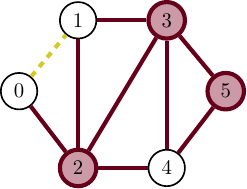}
        \caption{$\textit{PC} = \{2, 3, 5\}$, \\$\text{profit} = 5$}
        \label{subfig:pc-b}
    \end{subfigure}
    \hspace{0.05\textwidth}
    \begin{subfigure}{0.2\textwidth}
        \centering
          \includegraphics[width=\textwidth]{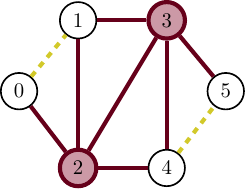}
        \caption{$\textit{PC} = \{2, 3\}$,\\$\text{profit} = 5$}
        \label{subfig:pc-c}
    \end{subfigure}
    \hspace{0.05\textwidth}
    \begin{subfigure}{0.2\textwidth}
        \centering
          \includegraphics[width=\textwidth]{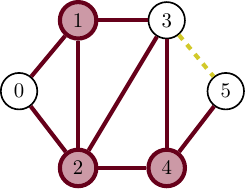}
        \caption{$\textit{PC} = \{1, 2, 4\}$,\\$\text{profit} = 5$}
        \label{subfig:pc-d}
    \end{subfigure}
    \caption{An illustration of different profit covers with optimal profit.  The graph shown in Fig.~\ref{subfig:pc-a} is also a minimum vertex cover. Graphs in Fig.~\ref{subfig:pc-b}, ~\ref{subfig:pc-c}, and ~\ref{subfig:pc-d} can be converted into minimum vertex covers using classical post-processing (cf.  Sec.~\ref{subsec:classical-post}). }
    \label{fig:vc-pc}
\end{figure}

\begin{figure}[!htbp]
    \centering
    \begin{subfigure}{0.2\textwidth}
        \centering
        \includegraphics[width=\textwidth]{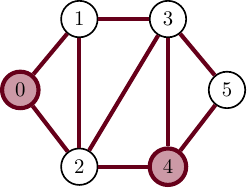}
        \caption{$\textit{PI} = \{0, 4\}$,\\$\text{profit} = 2$}
        \label{subfig:pi-a}
    \end{subfigure}
    \hspace{0.05\textwidth}
    \begin{subfigure}{0.2\textwidth}
        \centering
        \includegraphics[width=\textwidth]{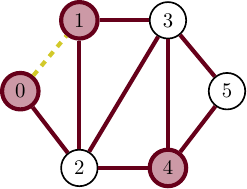}
        \caption{$\textit{PI} = \{0, 1, 4\}$,\\$\text{profit} = 2$}
        \label{subfig:pi-b}
    \end{subfigure}
    \hspace{0.05\textwidth}
    \begin{subfigure}{0.2\textwidth}
        \centering
          \includegraphics[width=\textwidth]{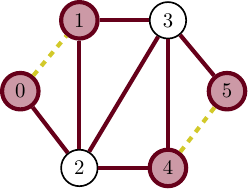}
        \caption{$\textit{PI} = \{0, 1, 4, 5\}$,\\$\text{profit} = 2$}
        \label{subfig:pi-c}
    \end{subfigure}
    \hspace{0.05\textwidth}
    \begin{subfigure}{0.2\textwidth}
        \centering
          \includegraphics[width=\textwidth]{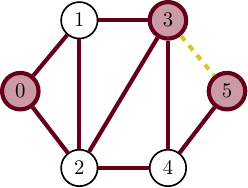}
        \caption{$\textit{PI} = \{0, 3, 5\}$,\\$\text{profit} = 2$}
        \label{subfig:pi-d}
    \end{subfigure}
    \caption{An illustration of different profit independent sets with optimal profit. The graph shown in Fig.~\ref{subfig:pi-a} is also a maximum independent set. Graphs in Fig.~\ref{subfig:pi-b}, ~\ref{subfig:pi-c}, and ~\ref{subfig:pi-d} can be converted into a maximum independent set using classical post-processing (cf.  Sec.~\ref{subsec:classical-post}).}
    \label{fig:is-pi}
\end{figure}

\subsubsection{An Example}
Fig.~\ref{fig:vc-pc} shows a graph with different choices of subsets of vertices as profit covers (in red or shaded).  The selected vertices in Fig.~\ref{subfig:pc-a} form a minimum vertex cover (of size 4) as well as a maximum profit cover (with $\mathfrak p_{\text{PC}}= 5$). Fig.~\ref{subfig:pc-b},~\ref{subfig:pc-c}, and~\ref{subfig:pc-d} are not feasible vertex covers but are feasible profit covers with maximum profit $\mathfrak p_{\text{PC}}= 5$.   
The maximum profit covers in Fig.~\ref{subfig:pc-b},~\ref{subfig:pc-c}, and~\ref{subfig:pc-d} can be converted into a minimum vertex cover without changing the profit by adding vertices associated with uncovered edges using classical post processing based on Algorithm~\ref{alg:add_vertices}.

\subsection{\textit{\textsc{Maximum Profit Independence}}}
\label{subsec:pi}
We now consider  \textit{\textsc{MaxIS}} and relax the requirement that the subset must be an independent set. That is, we consider the closeness of a subset to being an independent set  (i.e., the more independent vertices the better) and arrive at the problem {\sc Maximum Profit Independence}~\cite{van2008tractable}, or \textit{\textsc{MaxPI}}. For a subset $\textit{PI}\subseteq V$ in $G$, the number of vertices in $\textit{PI}$ is the \textit{gain}, and the number of vertices violating the independent vertex constraint is considered the \textit{loss}. The profit of $\textit{PI} \subseteq V$ for a graph $G=(V,E)$ is then  defined as  $\textit{profit} = \text{gain} - \text{loss}$.

\noindent We define the problem \textit{\textsc{MaxPC}} (\textit{\textsc{MaxPI}}) for a graph $G=(V, E)$ where $\textit{PI} \subseteq V$, as follows.

Maximize: $\mathfrak p$, where 
\begin{align*}
{\mathfrak p}_{\text{PI}} =  |\textit{PI}| - |E_{\text{PI}}(G, \textit{PI})|
\end{align*}

\noindent Here, $E_{\text{PI}}(G, \textit{PI})$ represents are the edges with both endpoints in $\textit{PI}$:
\begin{align*}
E_{\text{PI}}(G, \textit{PI}) = \{uv \in E: u, v \in \textit{PI}\}
\end{align*}

\noindent We call a $\textit{PI}$ with maximum profit also a \textit{maximum profit independence}.  \textit{\textsc{MaxPI}} is NP-hard~\cite{stege2002}.

\subsubsection*{Cost function for \textit{\textsc{Maximum Profit Independence}}}

The binary variables for the \textit{\textsc{MaxPI}} problem are similar to \textit{\textsc{MaxIS}}: each $x_v$ is a binary variable with value $1$ if $v$ is included in the Profit Independence set $\textit{PI}$, and $0$ otherwise.
The edge and vertex cost functions for \textit{\textsc{MaxPI}} are given as follows:

    \textbf{Edge cost:}\\
    $$C_E^{PI}(\vec{x}) = \sum_{uv \in E} (x_ux_v)$$ 
    
    \textbf{Vertex cost:}\\
    $$C_V^{PI}(\vec{x}) =  \sum_{v} x_v$$

The total cost that is maximized for \textit{\textsc{MaxPI}} is
    \begin{equation}
\label{eq:piqubo}
 C_{\text{PI}}(\vec{x}) =  C_V^{PI}(\vec{x}) -  C_E^{PI}(\vec{x})
 \end{equation}

\noindent While $C_{\text{IS}}(\vec{x})$ and $C_{\text{PI}}(\vec{x})$ may appear similar, $C_{\text{PI}}(\vec{x})$ has no penalty parameters that need to be set, since every binary variable assignment corresponds to a feasible solution. 
\subsubsection{Finding maximum independent sets via maximum profit independent sets}
Similar to the connection between the decision versions of \textit{\textsc{MinVC}} and \textit{\textsc{MaxPC}}, maximum independent sets can also be identified through profit independent sets. Fig.~\ref{subfig:pi-a} shows a maximum independent set and a maximum profit cover (with profit $\mathfrak p_{\text{PI}}=2$). Fig.~\ref{subfig:pi-b},~\ref{subfig:pi-c}, and~\ref{subfig:pi-d} show feasible (and maximum) profit independent sets that are infeasible independent sets. The following theorem shows the relationship between profit independence and independent sets which can then be subsequently used to convert profit independent sets to independent sets. 

\begin{thm}~\cite{stege2002}\label{thm:equiv-pi}
For any graph $G =(V,E)$, $G$ has an independent set $\textit{IS}\subseteq V$ of size $k$ if a only if $G$ has a subset $\textit{PI}\subseteq V$ with profit ${\mathfrak p}_{\text{PI}}= k$. 
\end{thm} 

\noindent Theorem~\ref{thm:equiv-pi} implies that, given a maximum profit independence $\textit{PI}$ for a graph $G$, we can obtain a minimum vertex cover for $G$. 
To convert profit independence results to independent sets, we apply Algorithm~\ref{alg:remove_vertices} based on Theorem~\ref{thm:equiv-pi}. 
\begin{algorithm}
\caption{Converting Profit Independent Set to Independent Set}
\label{alg:remove_vertices}
\begin{algorithmic}

\REQUIRE Graph $G$, Solution $\textit{PI}$
\STATE Let $E_{\text{conf}}$ be the set of edges with both endpoints in $\textit{PI}$
\FOR{each edge $uv \in E_{\text{conf}}$}
    \IF{$u \in \textit{PI}$ and $v \in \textit{PI}$}
        \STATE Remove $u$ from $\textit{PI}$
    \ENDIF
\ENDFOR
\end{algorithmic}
\end{algorithm}

\subsection{\textit{\textsc{Maximum Profit Clique}}}
\label{subsec:pcli}
The problem \textit{\textsc{Maximum Profit Clique}}~\cite{scott2004classical} involves the relaxation of the definition of the concept clique to include less-than-complete sub-graphs. Cliques and independent sets are related: an independent set ($\textit{IS}$) in a graph's complement ($G_c$) is equivalent to a clique in the original graph ($G$). This applies to Profit Independence and Profit Clique as well (see Fig.~\ref{fig:relationship}). 

We define the problem \textit{\textsc{Maximum Profit Clique}} \textit{\textsc{MaxPCl}} for a given graph $G=(V, E)$ where $\textit{PCl} \subseteq V$, as follows.

Maximize: $\mathfrak p_{\text{Cl}}$, where 
\begin{align*}
{\mathfrak p}_{\text{Cl}} &=  |\textit{PCl}| - |E_{\text{PCl}}(G, \textit{PCl})|\\
&= |\textit{PI}| - |E_{\text{PI}}(G_c, \textit{PI})|
\end{align*}

\noindent Here, $E_{\text{PCl}}(G, \textit{PCl})$ represents are the edges in the complement graph $G_c$ with both endpoints in $\textit{PCl}$:
\begin{align*}
E_{\text{PCl}}(G, \textit{PCl}) = \{uv \in (V \times V) \setminus E: u, v \in \textit{PCl}\}
\end{align*}

\noindent We call a subset $\textit{PCl}$ with maximum profit also a \textit{maximum profit clique}. \textit{\textsc{MaxPCl}} is  NP-hard~\cite{scott2004classical}.

\subsubsection*{Cost function for \textit{\textsc{Maximum Profit Clique}}}

The binary variables for the \textit{\textsc{MaxPCl}} problem are similar to \textit{\textsc{MaxPI}}. $x_v$ is a binary variable whose value is $1$ if $v$ is included in the profit clique set $\textit{PCl}$, and $0$ otherwise.
The cost function for  \textit{\textsc{MaxPCl}} is given as follows:

    \textbf{Edge cost:}\\
    $$C_{E_c}^{PCl}(\vec{x}) = \sum_{uv \in E_c} (x_ux_v)$$ 
    
    \textbf{Vertex cost:}\\
    $$C_V^{PCl}(\vec{x}) =  \sum_{v} x_v$$

The total cost that needs to be maximized for \textit{\textsc{MaxPCl}} is
    \begin{equation}
\label{eq:pclqubo}
 C_{\text{PCl}}(\vec{x}) =  C_V^{PCl}(\vec{x}) -  C_{E_c}^{PCl}(\vec{x})
 \end{equation}

\subsection{Relationships among all six problems}
We summarize the relationships among all six problems, see also Fig.~\ref{fig:relationship}. Let $G = (V,E)$. Then $G$ has a vertex cover of size $k$ if and only if $G$ has an independent set of size $|V| - k$ if and only if $G$'s complement $G_c$ has a clique of size $|V| - k$  
if and only if $G$ has a profit cover of profit $|E| - k$ if and only if $G$ has a profit independent set of profit $|V| - k$ if and only if $G_c$ has a profit clique of profit $|V| - k$.  

\section{EXPERIMENTAL SETUP}\label{sec:methodology}

Our primary goal is to evaluate and contrast the quality of solutions obtained for profit formulations and penalty-term constrained formulations using QAOA. 

\subsection{Cost Hamiltonians}

We formulate all cost Hamiltonians as minimization problems, where vertices in the subset evaluated are mapped to $-1$ and the variables of vertices not chosen are mapped to $1$. Based on the QUBO in Eqn.~\ref{eq:vcqubo}, we apply the transformation $x_i \rightarrow \frac{1 -Z_i}{2}$ to derive the  cost Hamiltonian for \textit{\textsc{MinVC}}:
\begin{equation}
\label{eq:vch}
    \hat{H}_{\text{VC}} = \frac{A}{4}\sum_{uv\in E}  (Z_uZ_v + Z_u + Z_v) - \frac{B}{2}\sum_{v\in V}Z_v + \Delta_{\text{VC}}.
\end{equation}
Here,  $Z_u$ and $Z_v$ are Pauli-$Z$ operators acting on qubits $u$ and $v$. 
Furthermore, $\Delta_{\text{VC}} = \frac{1}{4}A|E| + \frac{1}{2}B|V|$, and penalties $A$ and $B$ in $\hat{H}_{\text{VC}}$ are set to $3$ and $2$, respectively, unless stated otherwise. Note that  $\frac{\hat{H}_{\text{VC}}}{B}$ corresponds to the size of the solution state for (only) the cases that correspond to a feasible vertex cover solution.

Our choice of penalties is based on PennyLane's implementation~\cite{pennylane_qaoa_cost} of \textit{\textsc{MinVC}}. It is difficult to set desirable penalties~\cite{boros2008max}, as they depend on the particular input graphs. Small penalties increases the  likelihood of infeasible solutions being returned and values that are too large may lead to slower convergence.

The corresponding Hamiltonian $\hat{H}_{\textit{PC}}$  for \textit{\textsc{MaxPC}}, is based on its QUBO as in Eqn.~\ref{eq:pcqubo}:
\begin{equation}
\label{eq:pch}
\hat{H}_{\text{PC}} = \frac{1}{4}\sum_{uv\in E} (Z_uZ_v + Z_u + Z_v) - \frac{1}{2}\sum_{v\in V}Z_v + \Delta_{\text{PC}}
\end{equation}
where $\Delta_{\text{PC}} = \frac{|V|}{2} - \frac{3|E|}{4}$. $\hat{H}_{\text{PC}}$ can be obtained by setting $A=B=1$ in Eq.~\ref{eq:vch}, but it is important to highlight that this does not imply that all constrained problems can be transformed into profit problems by setting penalties to $1$. For instance, setting $A=1$ in the Dominating Set QUBO~\cite{dinneen2017formulating} does not transform it to the QUBO for Profit Domination~\cite{van2008tractable}.
Any solution obtained using the maximum profit Hamiltonian is either  a vertex cover or can be converted into a vertex cover by a classical polynomial-time post-processing step, and without loss of profit (cf. Theorem~\ref{thm:equiv}).

The following cost Hamiltonians are used for \textit{\textsc{MaxIS}} and \textit{\textsc{MaxPI}}:

\begin{equation}
\hat{H}_{\text{IS}} = \frac{A}{4}\sum_{uv\in E}  (Z_uZ_v - Z_u - Z_v)  + \frac{B}{2}\sum_{v\in V}Z_v + \Delta_{\text{IS}}    
\end{equation}

where $\Delta_{\text{IS}}  = \frac{A}{4}|E| - \frac{B}{2}|V|$.

\begin{equation}
\hat{H}_{\text{PI}} = \frac{1}{4}\sum_{uv\in E}  (Z_uZ_v - Z_u - Z_v)  + \frac{1}{2}\sum_{v\in V}Z_v + \Delta_{\text{PI}}    
\end{equation}

where $\Delta_{\text{PI}}  = \frac{1}{4}|E| - \frac{1}{2}|V|$.

\noindent Similar Cost Hamiltonians are used for  \textit{\textsc{MaxCl}} and \textit{\textsc{MaxPCl}}.

\begin{equation}
\hat{H}_{\text{Cl}} = \frac{A}{4}\sum_{uv\in (V \times V) \setminus E}  (Z_uZ_v - Z_u - Z_v)  + \frac{B}{2}\sum_{v\in V}Z_v + \Delta_{\text{Cl}}    
\end{equation}

where $\Delta_{\text{Cl}}  = \frac{A}{4}|(V \times V) \setminus E| - \frac{B}{2}|V|$.
\begin{equation}
\hat{H}_{\text{PCl}} = \frac{1}{4}\sum_{uv \in (V \times V) \setminus E}  (Z_uZ_v - Z_u - Z_v)  + \frac{1}{2}\sum_{v\in V}Z_v + \Delta_{\text{PCl}}    
\end{equation}

where $\Delta_{\text{PCl}}  = \frac{1}{4}|(V \times V) \setminus E| - \frac{1}{2}|V|$.

\noindent Note that all Cost Hamiltonians are formulated as minimization problems.




\subsection{Mixer Hamiltonian}
We use a basic Pauli-X mixer Hamiltonian as follows:
$$
\hat{H}_M = \sum_i X_i
$$

\subsection{Evaluation on the Xanadu PennyLane and Argonne QTensor Simulator platforms}

To evaluate the performance of \textit{\textsc{MaxPC}}, we use two quantum computing software frameworks, PennyLane and QTensor. 

\begin{enumerate}
    \item \textbf{Xanadu PennyLane}: PennyLane is a cross-platform Python library for differentiable quantum programming, making it a powerful tool for hybrid quantum-classical techniques that involve variational circuits~\cite{bergholm2018pennylane}. PennyLane enables the integration of various libraries such as TensorFlow, PyTorch, and Autograd due to the quantum node abstraction, thereby fitting seamlessly into the existing automatic differentiation methods. In our work, we utilized PennyLane not only for its robust support for variational algorithms but also for its ease of use and the availability of plugins for different quantum simulators and hardware. 

\item \textbf{Argonne QTensor Quantum Circuit Simulator}: Argonne QTensor Simulator~\cite{lykov2021performance, lykov2022} is a highly efficient quantum circuit simulator library which is founded upon the tensor network contraction technique~\cite{markov2008simulating},  offering a significant performance speedup compared to existing simulators (e.g., those that run a full amplitude-vector evolution). Quantum circuit simulators are essential in understanding how quantum computers work and in experimenting with  quantum algorithms including benchmarking and verification. Variational algorithms that involve iteratively updating the parameters based on the results obtained from quantum circuit executions benefit from fast simulators to obtain parameters. Furthermore, quantum circuit simulators can be used to study the behaviour of such hybrid algorithms under varying parameters. A core feature of Argonne QTensor is its highly parallelizable evaluation of observables based on the \textit{lightcone} or the \textit{reverse causal cone} technique~\cite{farhi2004quantum, streif2020training}.  This means that if some observable acts only on a small subset of the qubits, then most of the gates in the quantum circuit commute through and cancel out when evaluating the expectation value. The QAOA cost operator is a sum of $m$ independent terms, each of which can be computed separately (i.e.,  can be computed in parallel). The lightcone (or subgraph) of the computation depends on the number of layers $p$. For example, when $p=1$, the simulator only needs to evaluate the circuit on the independent cost term and its immediate neighbourhood. This method works well for sparse or large fixed-degree graphs. However, for dense graphs, since every evaluation requires consideration of numerous nodes, the speedup provided by Argonne QTensor may be negligible.
\end{enumerate}
For smaller problem instances involving graphs with fewer than 14 nodes, we use PennyLane to perform full state vector simulation. For larger problem instances (up to 70 nodes), we use tensor network expectation value simulation with QTensor due to its memory efficiency and scalability. PennyLane code\footnote{\href{https://github.com/RigiResearch/pf-opt}{{https://github.com/RigiResearch/pf-opt}}} and QTensor code used in this study will be made available on GitHub.\footnote{\href{https://github.com/danlkv/QTensor}{https://github.com/danlkv/QTensor}} To validate our proposed approach, we employ the following methodology: 
\begin{enumerate}
    \item A set of random connected graphs of varying size and density are generated. 
    \item The given problem is solved using QAOA on the set of random graphs using our profit relaxation approach, as well as the penalty-term formulation of \textit{\textsc{MinVC}}. 
    \item Results from our approach and results from the penalty-term QAOA are compared against classical exact solutions for varying graph sizes and densities. 
\end{enumerate}

  \begin{figure}[!bth]
\centering
  \begin{tabular}{@{}cccc@{}}
    \includegraphics[width=1\columnwidth]{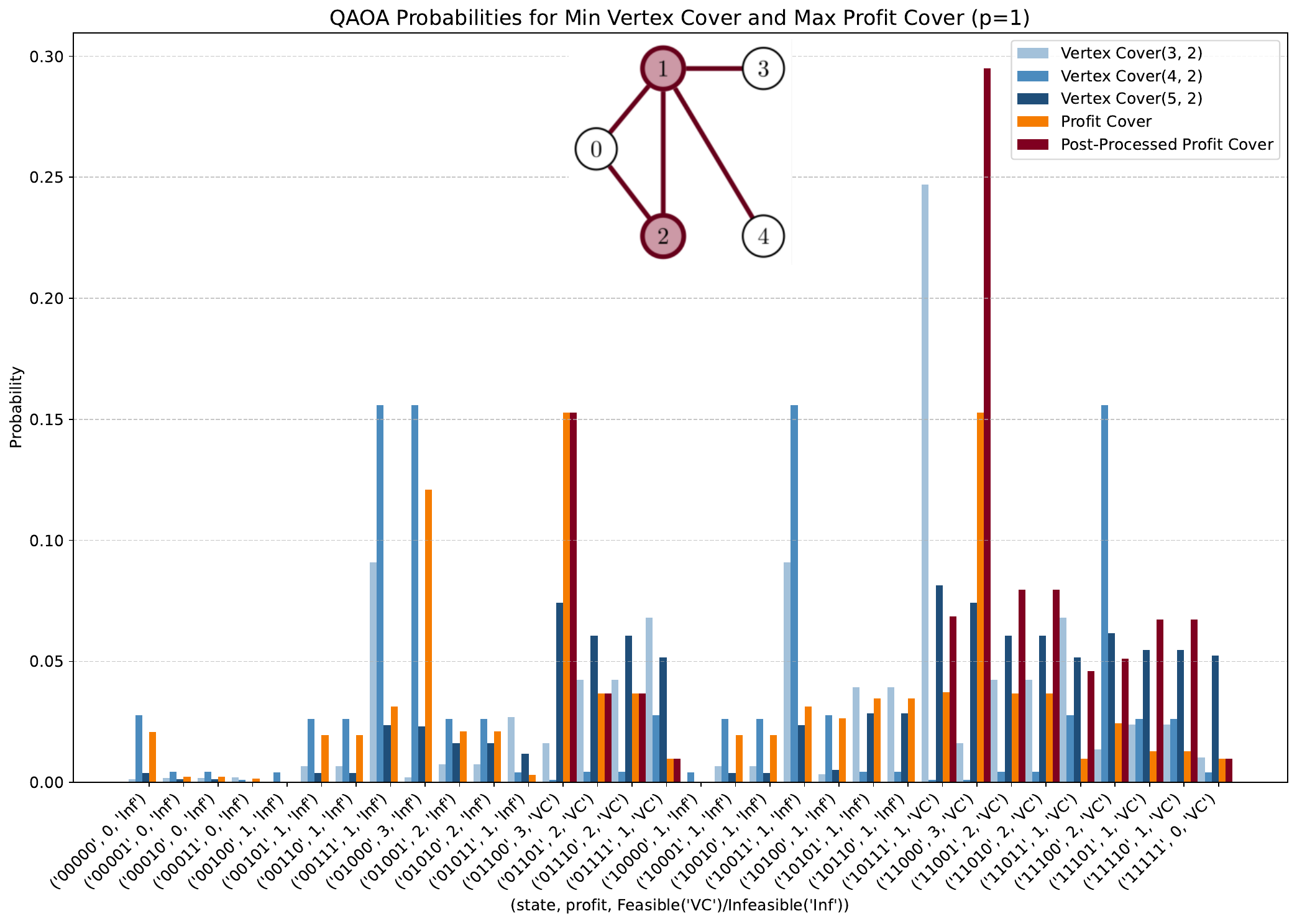}\\
    \includegraphics[width=1\columnwidth]{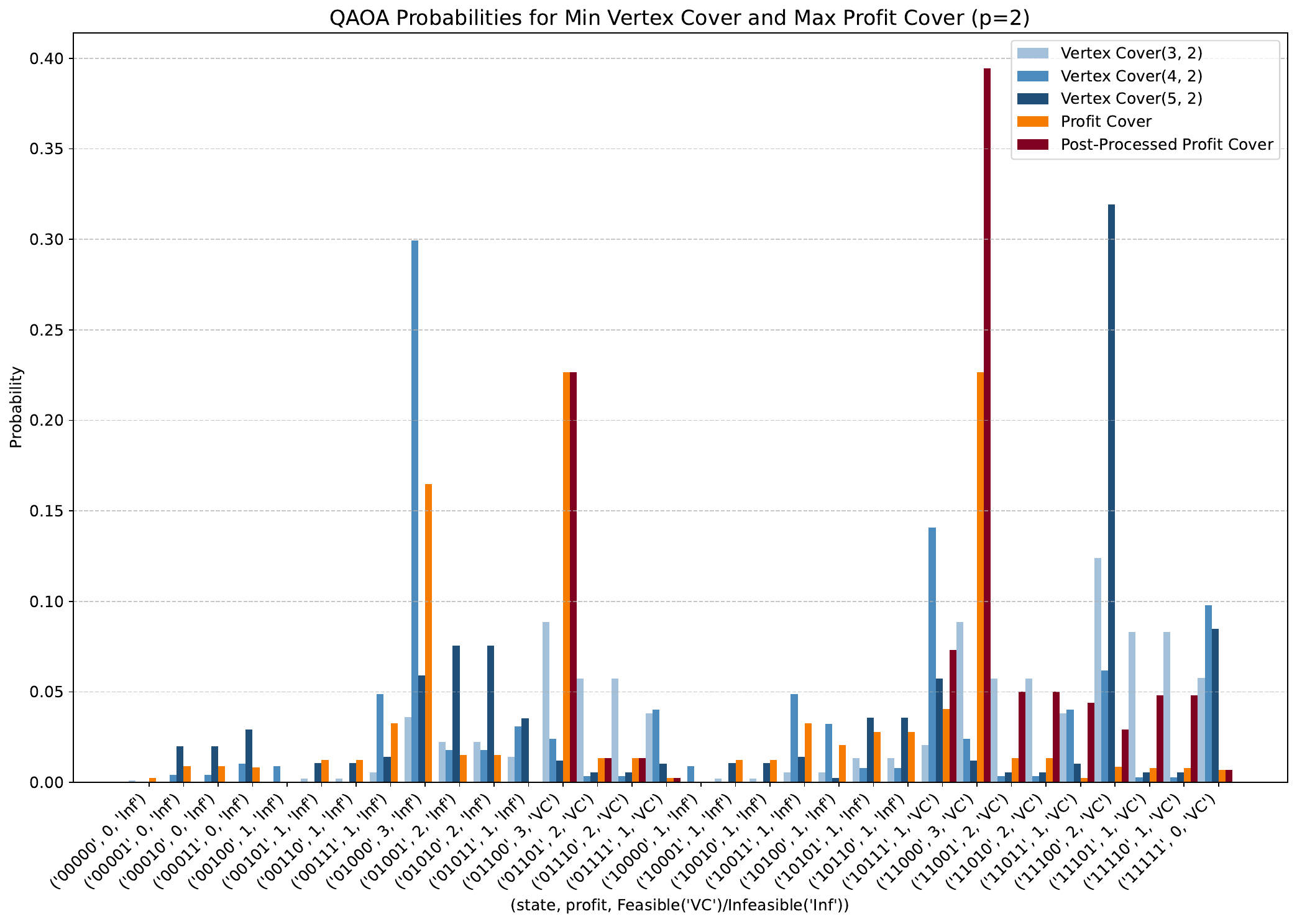}\\
    \includegraphics[width=1\columnwidth]{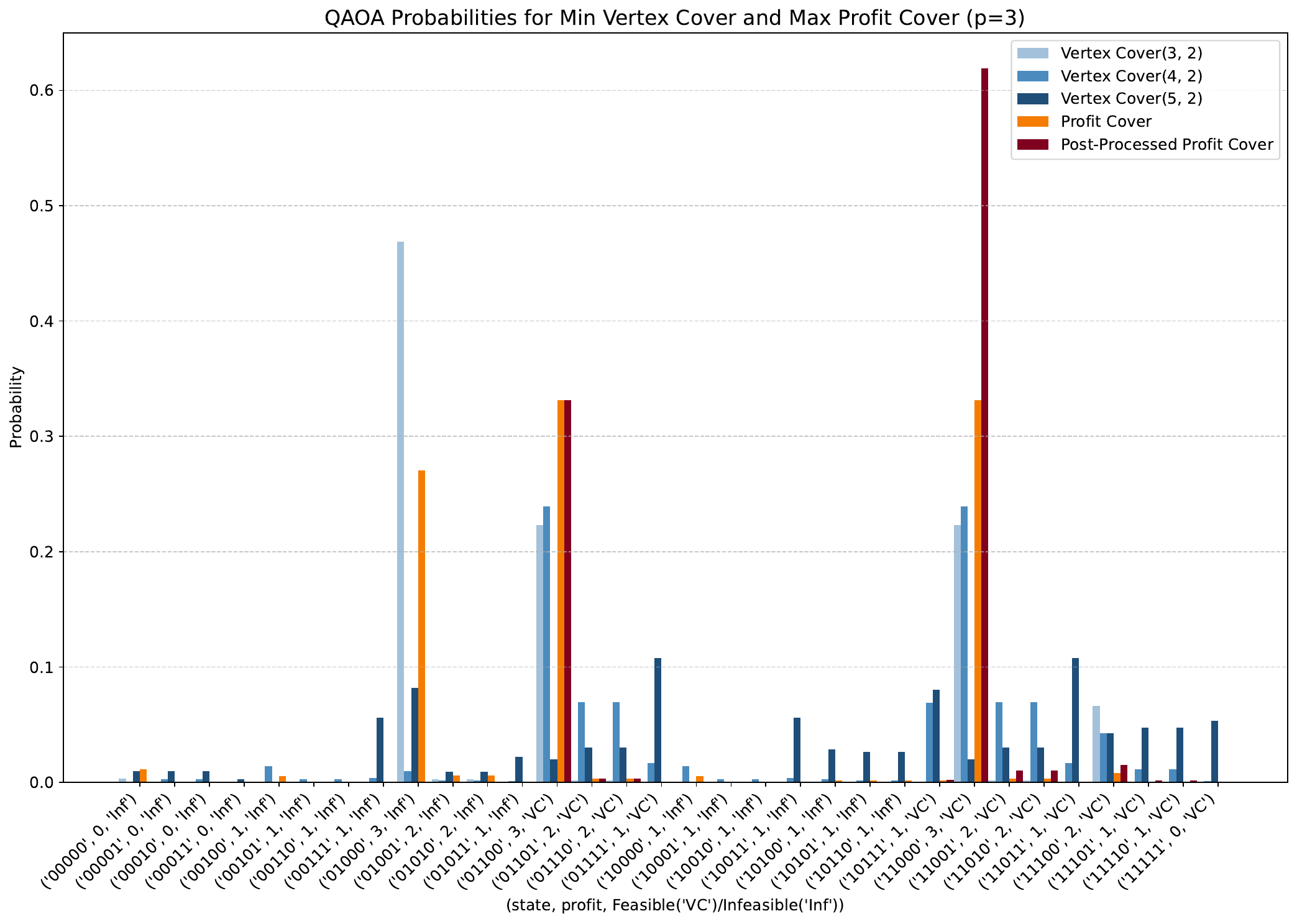}

  \end{tabular}
  \caption{Probabilities for a sample graph (shown in the inset of topmost figure or graph) containing 5 vertices with $p\in \{1,2,3\}$ layers. \textit{\textsc{MinVC}} is run for three different penalty parameters and is shown in blue. For example, Vertex Cover (3, 2) refers to penalties $A=3, B=2$. Results of \textit{\textsc{MaxPC}} are shown in orange and the post-processed results of \textit{\textsc{MaxPC}} are shown in burgundy.}
    \label{fig:vc-both}
\end{figure}

\subsubsection{Choice and density of graphs} For PennyLane, we used random connected Erd{\H o}s-R{\'e}nyi~\cite{erdHos1960evolution} graphs of varying densities (with edge probabilities of $0.1, 0.3, 0.5, 0.8$ and 3-regular graphs).  For QTensor, we chose to work with sparse graphs (edge probability=0.1), and 3-regular graphs due to the large size of the graphs being simulated.

\subsubsection{Depth of circuits} For graphs with fewer that $n=14$ nodes, we performed experiments with up to eight layers of QAOA. For larger graphs on QTensor we performed simulations with up to five layers. 

\subsubsection{Classical optimizer} We have used both Gradient Descent and a Gradient Descent based optimizer RMSProp (Root Mean Square Propagation) on PennyLane and QTensor. Both are similar optimizers based on gradients, with a fixed learning rate on Gradient Descent, and an adaptive learning rate on RMSProp. The results presented here are obtained using RMSProp.

\subsection{Contributions to QTensor Simulator}
QTensor enables efficient classical simulation of quantum circuits using tensor network contractions and is well studied for the \textit{\textsc{MaxCut}} problem with QAOA~\cite{lykov2021performance}. The design principles of QTensor allow for ease of extensibility of the framework with minimal changes to the underlying QTensor architecture. The primary components of QTensor include \textit{Composers} and \textit{Simulators}. Composers allow for the creation of quantum circuits in a way that is backend-agnostic. These circuits can then be run using Simulators of popular libraries such as Qiskit and Cirq, as well as the QTensor tensor network simulator. We extend this framework to include the vertex cover and profit cover problems and their expectation value calculation on QTensor. Because our Hamiltonians contain four distinct terms, calculating an expectation value requires  four calls to the simulator. To alleviate this, we consolidate the node contributions, for example:

\begin{align}
\hat{H}_{\text{PC}} 
&=\frac{1}{4}\sum_{(i, j) \in E} \Big[ Z_iZ_j\Big]  +  \frac{1}{4}\sum_{i \in V}(\text{deg}(i) - 2)\cdot Z_i
\end{align}

\begin{align}
\hat{H}_{\text{PI}} 
&=\frac{1}{4}\sum_{(i, j) \in E} \Big[ Z_iZ_j\Big]  -  \frac{1}{4}\sum_{i \in V}(\text{deg}(i)-2)\cdot Z_i
\end{align}

\subsection{Classical Post-processing}
\label{subsec:classical-post}
To convert the results for the unconstrained profit problem (i.e., solutions for $P_U$) to results of the constrained problem ($P_C$), we apply polynomial-time classical post-processing that runs in $\mathcal{O}(|E|)$ time. We apply Algorithm~\ref{alg:add_vertices}  above based on Theorem~\ref{thm:equiv} on solutions to \textit{\textsc{MaxPC}} to obtain the corresponding solutions to \textit{\textsc{MinVC}}. Similarly, Algorithm~\ref{alg:remove_vertices} above based on Theorem~\ref{thm:equiv-pi} is used to obtain solutions to \textit{\textsc{MaxIS}}. To find solutions for \textit{\textsc{MaxCl}}, we use the same Algorithm~\ref{alg:remove_vertices} on the graph $G_C = (V, (V \times V) \setminus E)$ (the complement graph).

These polynomial-time algorithms guarantee that the profit of the post-processed result is at least as high as the profit of the input to the classical post-processing and ensure that the post-processed result is a feasible solution to the constrained optimization problem.
\subsection{Computing exact reference solutions classically}
To compare the results of our approach to optimum solutions for \textit{\textsc{MinVC}} inputs, we used the  implementation of an exact classical algorithm, as described below. The solutions obtained from this algorithm serve as benchmarks to evaluate the solution quality of our approach. 

To obtain exact solutions for the \textit{\textsc{MinVC}} problem, we employ a linear programming based problem kernel approach and use an implementation by Abu-Khzam of this reduction rule~\cite{Fomin2019,abu2005fast}, combined  with a basic bounded search-tree fixed-parameter algorithmic approach~\cite{DFS99,downeyFellows2013}. Our (classical) algorithm is listed below as Algorithm~\ref{alg:FPT}, and additional background on fixed-parameter tractability and classical algorithms is available in the Appendix (Section~\ref{sec:background}).
Before computing the bounded search tree, we  determine an upper bound of the size of the vertex cover using a greedy heuristic method, by repeatedly selecting vertices with the highest degree until a vertex cover is complete. Reference solutions for \textit{\textsc{MaxIS}} and \textit{\textsc{MaxCl}} are also obtained by running Algorithm~\ref{alg:FPT} with additional steps guided by the relationships among all three problems (cf. Section~\ref{sec:relationships}).


\begin{algorithm}[!htbp]
\caption{\textsc{MinVC($G$)}: \\Exact classical algorithm}
\begin{algorithmic}\label{alg:FPT}

\STATE $V_1, G' \gets \text{LPKernelization}(G)$ 
\STATE $V_2 \gets \text{GreedyVertexCover}(G')$ 
\STATE $k \gets |V_2|$
\STATE $V_{\text{best}} \gets V_1 \cup V_2$
\STATE $k_\text{best} = k -1$
\WHILE{$k_\text{best} \geq 0$}
    \STATE $V_k \gets \text{BoundedSearchTree}(G', k_\text{best})$
    \IF{$V_k = \emptyset$}
        \STATE \textbf{return} $V_{\text{best}}$
    \ELSE
        \STATE $k_\text{best} \gets k_\text{best} - 1$
        \STATE $V_{\text{best}} = V_1 \cup V_k$
    \ENDIF
\ENDWHILE
\STATE \textbf{return} $V_{\text{best}}$

\end{algorithmic}
\end{algorithm}

To find maximum independent sets and maximum cliques, we use Algorithm~\ref{alg:FPT} with some post-processing:
\begin{enumerate}
    \item \textit{\textsc{MaxIS}}: $\textit{IS} = V \setminus\text{MinVC($G$)}$
    \item \textit{\textsc{MaxCl}}: $MCl = V \setminus\text{MinVC($G_C$)}$
\end{enumerate}

\subsection{Evaluation metrics}
To evaluate our hybrid approach to find small vertex covers via large profit covers, we compare the results for \textit{\textsc{MaxPC}} and \textit{\textsc{MinVC}} on both PennyLane (less than 15 nodes) and QTensor (20-50 nodes). For smaller graphs, we examine the summed probability of optimal solutions for varying edge probabilities over eight layers.
For the larger graphs, we compare the expectation values obtained for both \textit{\textsc{MaxPC}} and \textit{\textsc{MinVC}} by considering the cost value evaluation over 100 iterations for $p\in\{1, 2, 3\}$ layers. Note that this is called \textit{loss} in Fig.~\ref{fig:cost-qtensor} but should not be confused with loss as part of the profit definition defined in Sec.~\ref{subsec:pc}.  We also show the approximation ratio obtained for \textit{\textsc{MaxPC}} in Fig.~\ref{fig:approx-ratios} and Fig.~\ref{fig:approx-ratios-2}. 
Both problems solved on PennyLane as well as QTensor use the same initial parameters and classical optimizer with 100 or more iterations.
\section{Results}
\label{sec:results}
This section presents our results which are based on the following metrics to evaluate our performance. For smaller problem instances, we use a combination of probabilities, summed optimal probabilities, summed near-optimal probabilities, and mean approximation ratio to evaluate performance. For larger graphs, a full statevector simulation is not possible due to an exponential number of states to evaluate. In this case, we perform an expectation value analysis using the Argonne QTensor Tensor network simulator.

\subsection{Probabilities}
In our analysis, we examine the individual probabilities of obtaining feasible solutions, comparing the profit  and constrained formulations using the standard (i.e., vanilla) QAOA implementation. Fig.~\ref{fig:vc-both} shows a probability distribution for \textit{\textsc{MaxPC}} and \textit{\textsc{MinVC}}  with $p\in \{1, 2, 3\}$ layers for the five node graph shown in the inset of Fig.~\ref{fig:vc-both}, run on PennyLane's standard analytical simulator (\texttt{default.qubit}). The classical optimizer used is the Root Mean Squared Propagation (RMSProp). The quantum-classical loop is run 200 times to obtain the results. The bar in the burgundy color in Fig.~\ref{fig:vc-both} shows the result of post-processing profit cover results to vertex cover using Alg.~\ref{alg:add_vertices}. As the number of layers increases, we can see a high probability (greater than 0.9) of obtaining the optimal result. 

\subsection{Summed Probabilities}
Summed probabilities (or the maximum success probability) refers to the summation of probabilities associated with obtaining optimal solutions in a given problem space. 

\begin{equation}
\label{eq:spopt}
    SP_{\text{OPT}} = \sum_{k \in sol_{\text{opt}}} P (X=k)
\end{equation}

We adopt the summed probability metric from ~\cite{Saleem2020}. They introduced it to capture optimal and near-optimal solutions to a constrained optimization problem better than, say, an approximation ratio.
In Equation~\ref{eq:spopt}, $sol_{\text{opt}}$ is the set of all optimal solutions, and $P(X=i)$ is the probability of obtaining solution $i$.  
The vertex cover, independent set, and clique problems have multiple optimal solutions in both their constrained and profit formulations. Notably, the profit variants of these problems tend to yield a greater number of optimal solutions compared to their constrained counterparts.  The increased density of optimal solutions for profit problems, coupled with the probabilistic nature of quantum computing outputs enhances the likelihood of the QAOA converging towards an optimal outcome.
Fig.~\ref{fig:vc-is-cl-summed-probs} shows the summed probabilities over eight layers for constrained and profit versions of vertex cover,  independent set and clique for varying edge probabilities averaged over ten graphs. 

\begin{figure*}[!bth]
\centering
  \begin{tabular}{@{}cccc@{}}
    \includegraphics[width=0.33\textwidth]{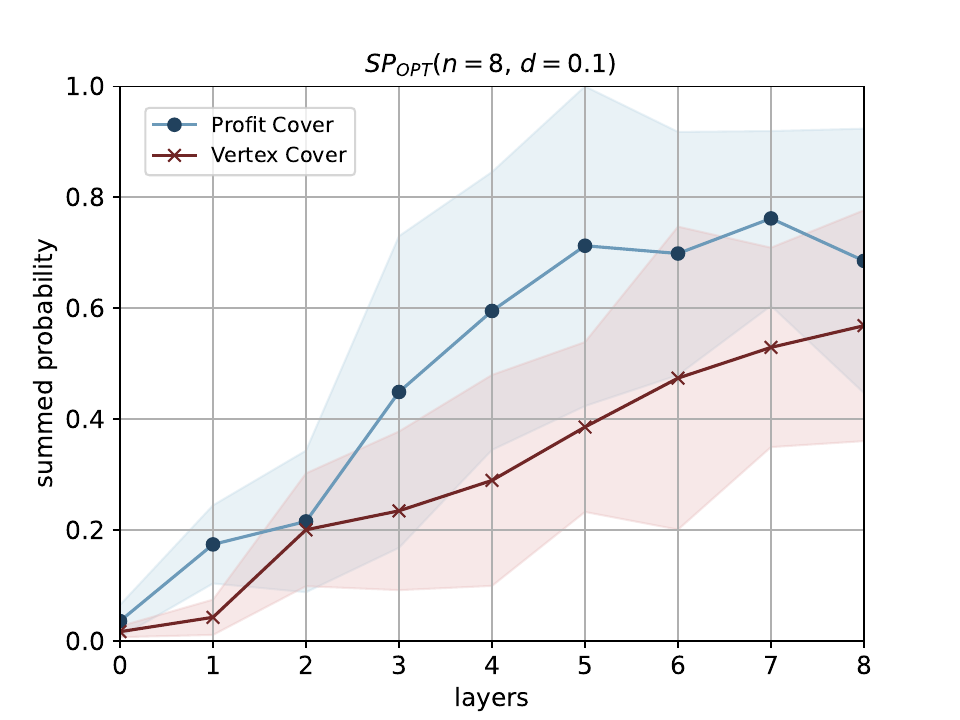} &
    \includegraphics[width=.33\textwidth]{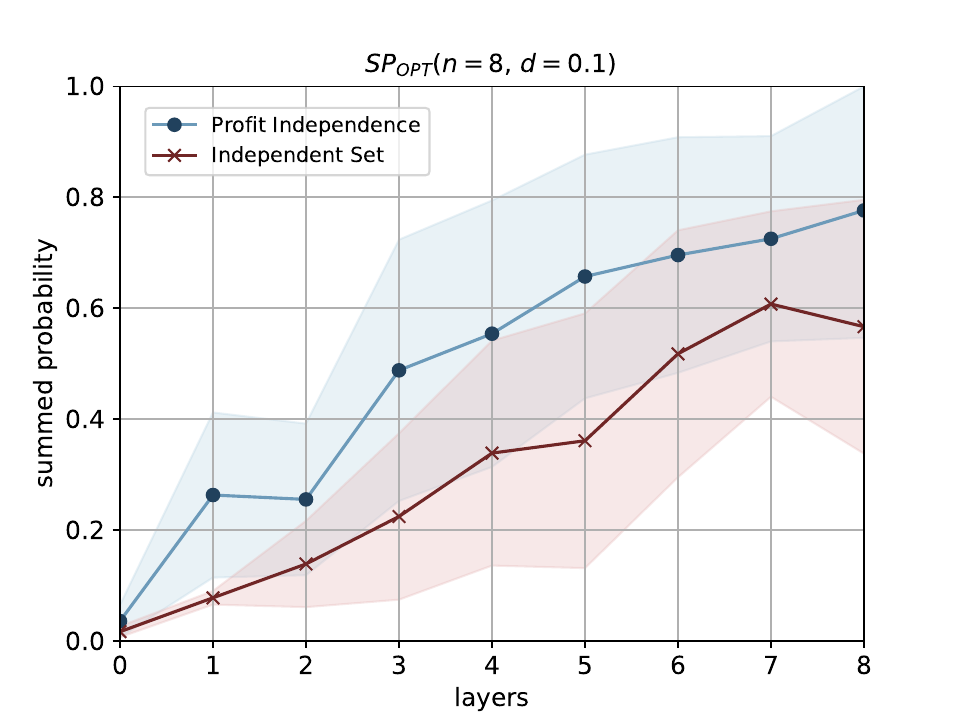} &
    \includegraphics[width=.33\textwidth]{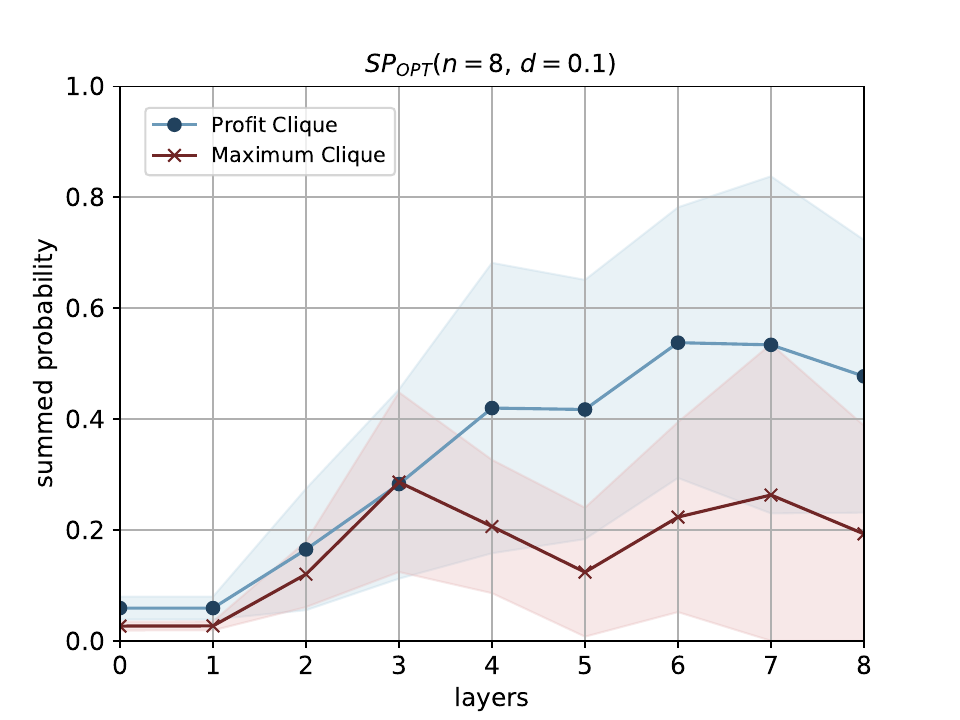} 
    \\
    \includegraphics[width=0.33\textwidth]{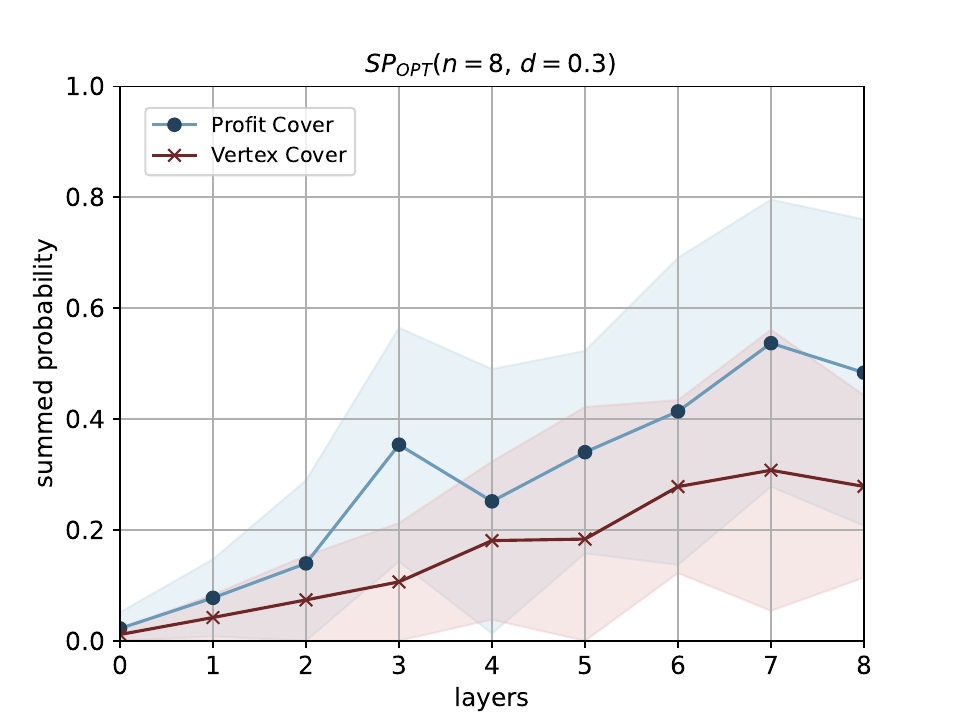} &
    \includegraphics[width=.33\textwidth]{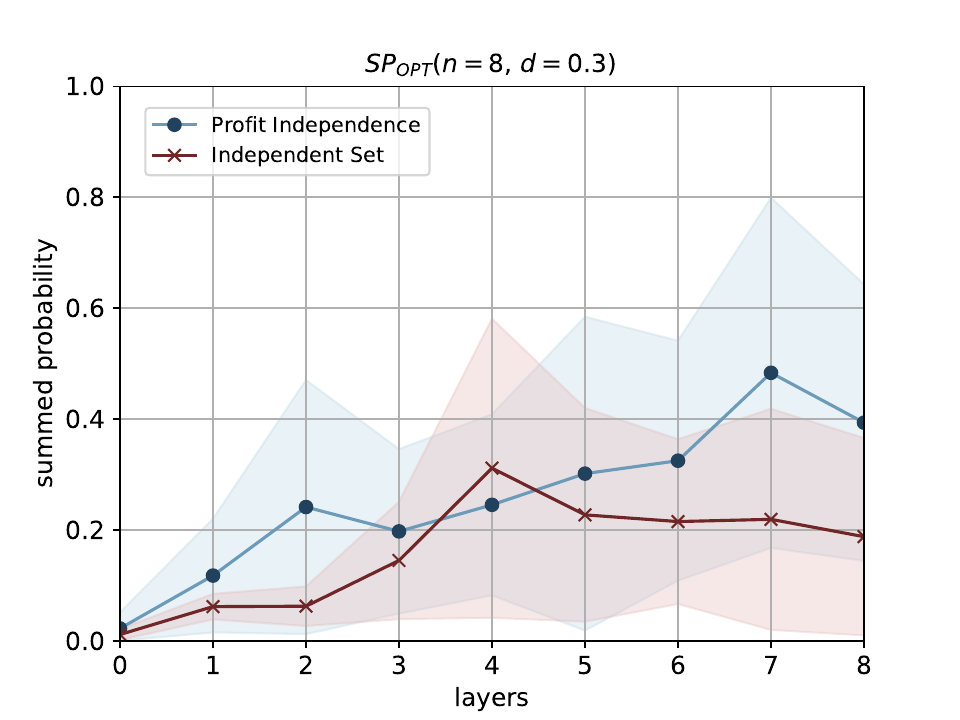} &
    \includegraphics[width=.33\textwidth]{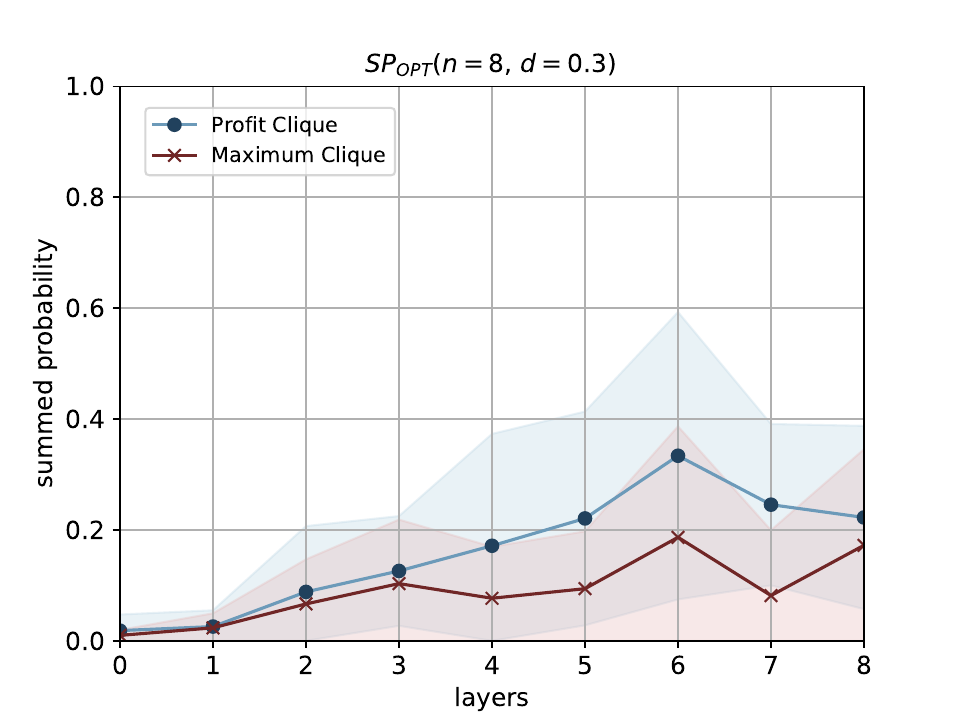} 
    \\
    \includegraphics[width=0.33\textwidth]{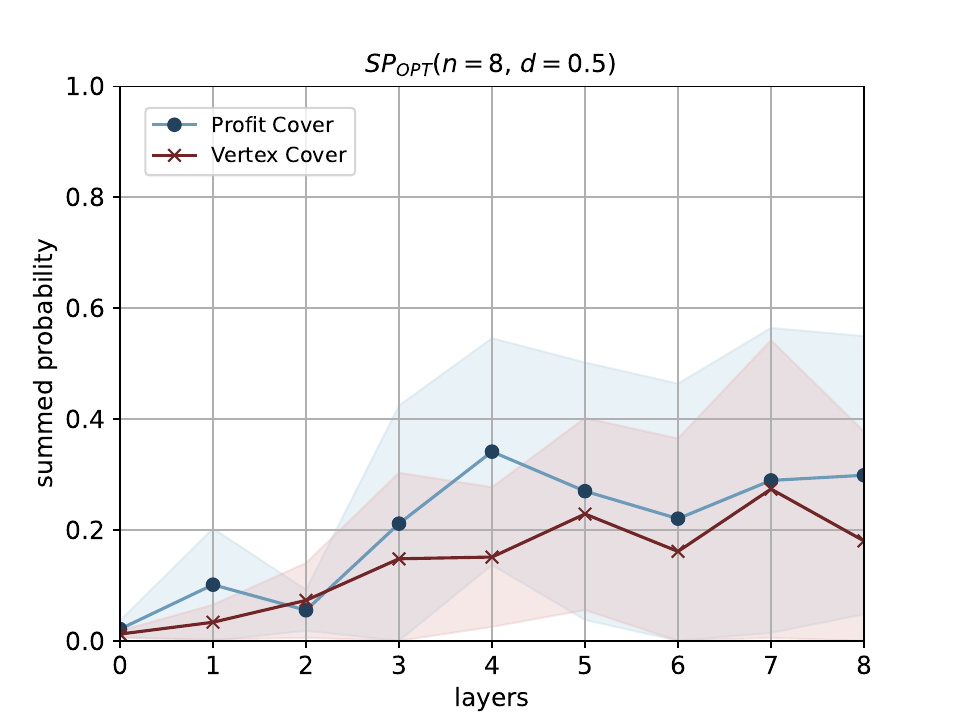} &
    \includegraphics[width=.33\textwidth]{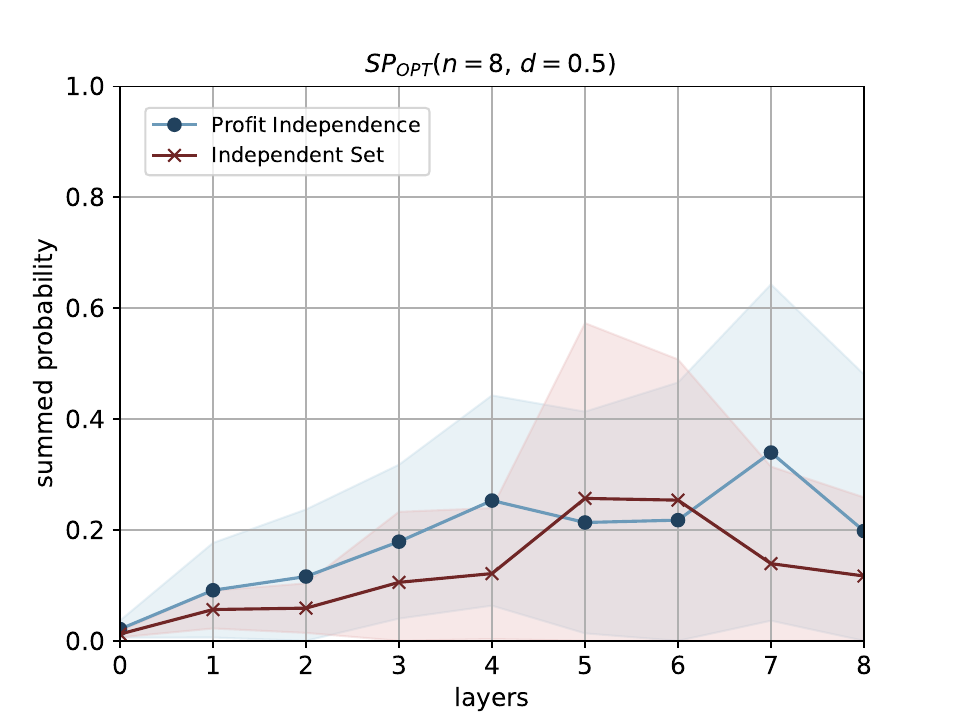} &
    \includegraphics[width=.33\textwidth]{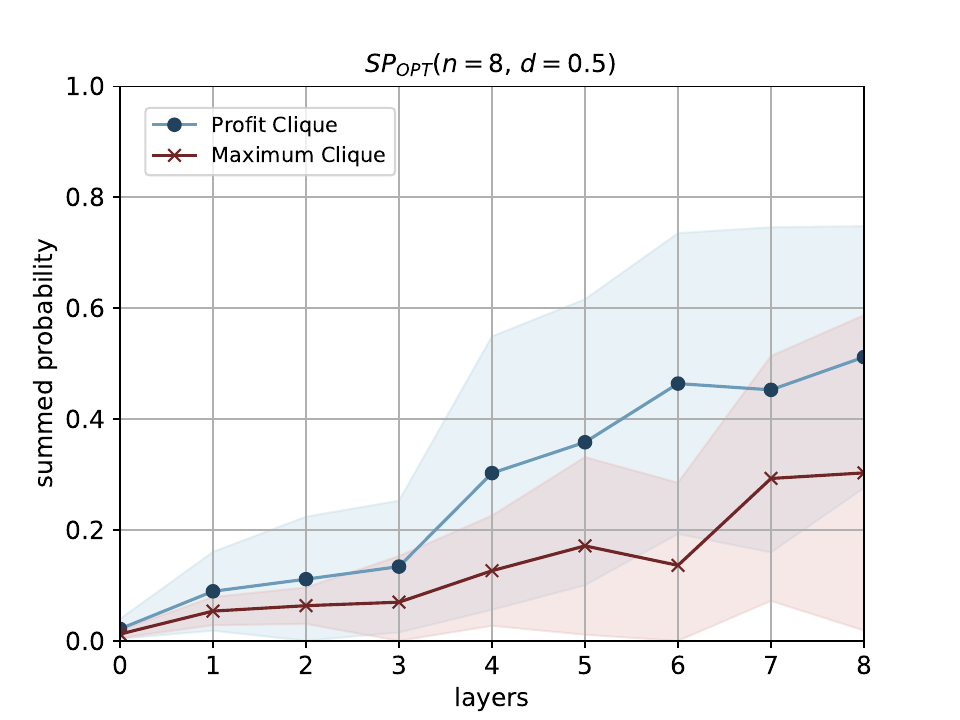} 
    \\
    \includegraphics[width=0.33\textwidth]{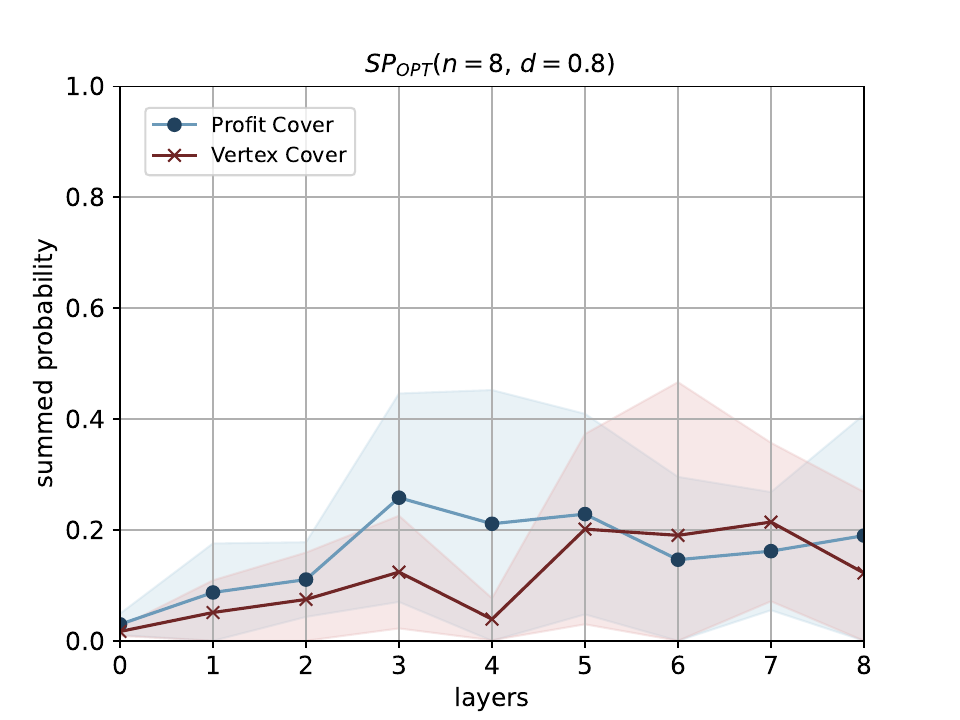} &
    \includegraphics[width=.33\textwidth]{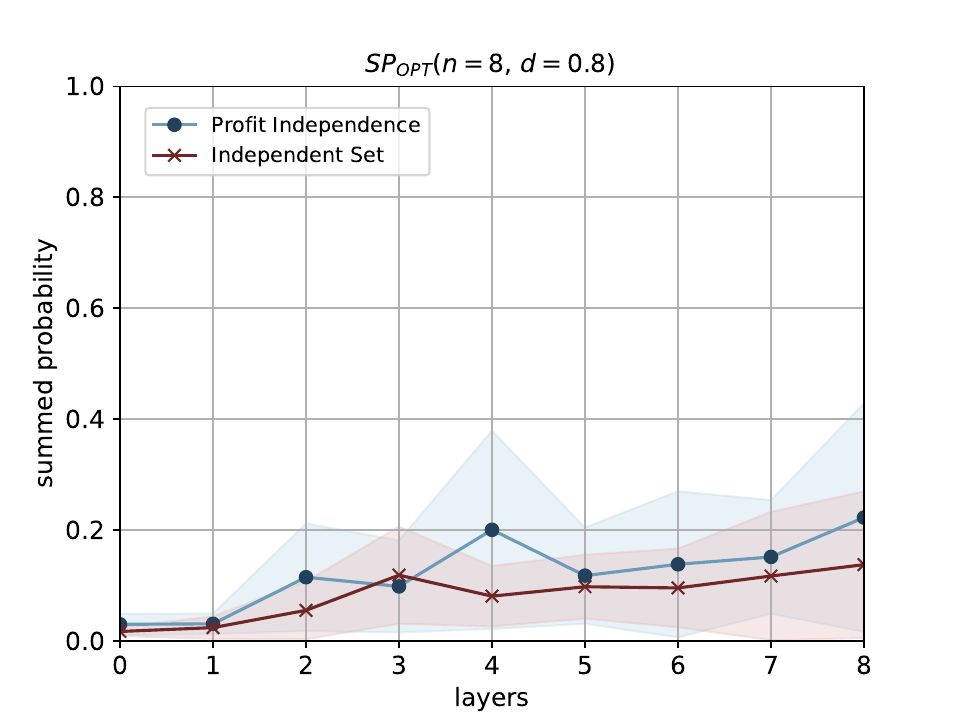} &
    \includegraphics[width=.33\textwidth]{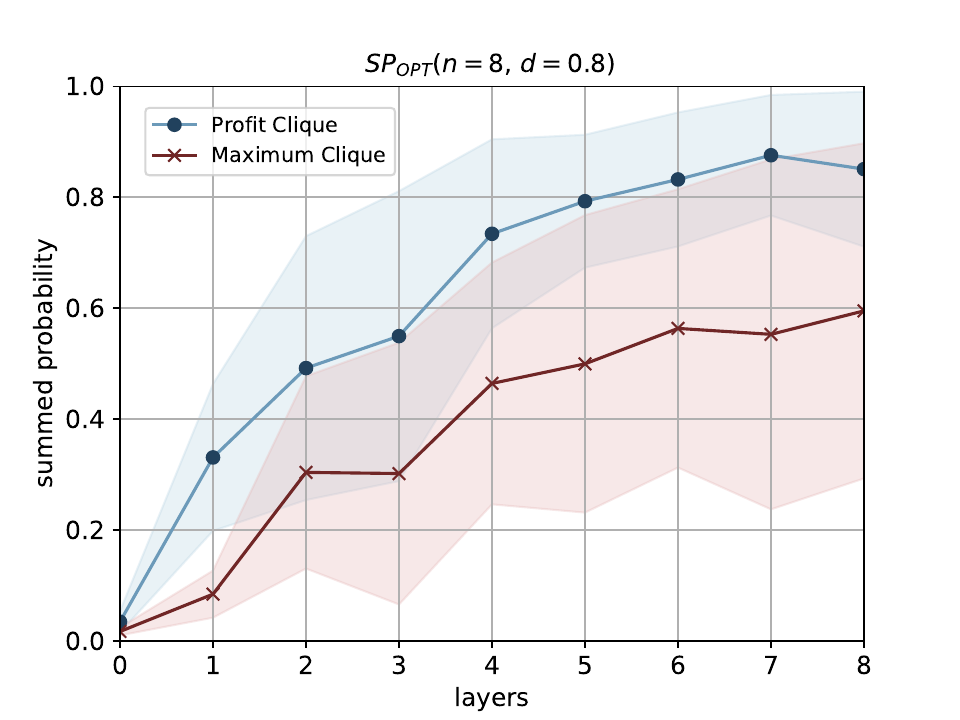} 
  \end{tabular}
  \caption{\textbf{Average summed probabilities of optimal solutions over ten graphs obtained over eight layers.} The graphs or figures in the leftmost column compare \textit{\textsc{MinVC}} and \textit{\textsc{MaxPC}}, the graphs or figures in the middle column show \textit{\textsc{MaxIS}} and  \textit{\textsc{MaxPI}}, the graphs or figures in the rightmost column compare \textit{\textsc{MaxCl}} and \textit{\textsc{MaxPCl}}. The rows indicate different edge densities of $d\in \{0.1, 0.3, 0.5, 0.8\}$. Each graph represents the comparison between the profit version (blue) and the constrained version (red). Data points correspond to the average value for each layer, with shaded regions indicating one standard deviation away from the mean.}
    \label{fig:vc-is-cl-summed-probs}
\end{figure*}

\subsection{Near-optimal analysis}
\label{subsec:near-opt}
Our analysis extends beyond focusing on the optimal solutions. In many practical scenarios, a range of near-optimal or sub-optimal solutions is desirable. For example, in financial applications, such as portfolio optimization, a solution with slightly higher risk might be preferred because it also comes with the potential for greater returns.  Therefore, we examine the probability distribution across this wider spectrum of solutions. 

Equations~\ref{eq:spopt1} and \ref{eq:spopt2} represent summed probabilities of optimal solutions,  second-,  third- best near-optimal solutions. 
\begin{equation}
\label{eq:spopt1}
    SP_{\text{OPT}-1} = SP_{\text{OPT}} + \sum_{k \in sol_{\text{opt}-1}} P (X=k)
\end{equation}

\begin{equation}
\label{eq:spopt2}
    SP_{\text{OPT}-2} = SP_{\text{OPT}} + SP_{\text{OPT}-1} + \sum_{k \in sol_{\text{opt}-2}} P (X=k)
\end{equation}

Here $sol_{\text{opt}-1}$ and $sol_{\text{opt}-2}$ refer to the sets of second and the third best solutions, respectively. 
Fig.~\ref{fig:vc-opt-summed-probs} shows summed probabilities of optimal and near-optimal solutions for \textit{\textsc{MinVC}} and \textit{\textsc{MaxPC}} for varying edge densities averaged over ten graphs. The leftmost column shows the summed probability of obtaining optimal solutions, the middle column shows the summed probability of obtaining optimal solutions and second best solutions. The rightmost column shows the summed probability of obtaining optimal solutions, second best solutions, and third best solutions.
In the case of profit problems, near-optimal solutions are those with the second and the third best profit. For example, in the case of \textit{\textsc{MaxPC}} the graph in Fig.~\ref{subfig:pc-a} has an optimal profit of 5. A second-best solution for this graph would be $PC = \{1, 3, 5\}$ with a profit of 4. A third best solution for this graph would be $PC=\{3, 5\}$ with a profit of 3. 

For constrained problems, such as the \textit{\textsc{MinVC}}, we choose the second-best solutions from the set of feasible solutions with a vertex cover size of $|VC|_{\text{OPT}} + 1$. Consider the same graph from Fig.~\ref{subfig:pc-a}, this has an optimal vertex cover of size 4. A second-best solution for this graph would be $\textit{VC} = \{1, 2, 3, 4, 5\}$ with a size of 5. A third best solution for this graph would be $PC=\{0, 1, 2, 3, 4, 5\}$ with a size of 6.  Note there are also examples of the second and the third best profit covers. Near-optimal solutions for vertex cover can only be obtained by adding vertices to the minimum vertex cover, however, near-optimal solutions to profit cover can be obtained by adding or removing vertices to the maximum profit cover. 

Given that constrained problems have a penalty-term formulation, it is worth noting that near-optimal solutions that have the second lowest or the third lowest cost (as calculated using Eq.~\ref{eq:vch}) may not correspond to feasible solutions depending on the penalties chosen (cf. Section \ref{sec:penalty-study}). Therefore, near-optimal solutions are chosen by summing the probabilities for all feasible solutions with size $|VC|_{\text{OPT}} + 1$ and $|VC|_{\text{OPT}} + 2$. 

\begin{figure*}[!bth]
\centering
  \begin{tabular}{@{}cccc@{}}

    \includegraphics[width=0.33\textwidth]{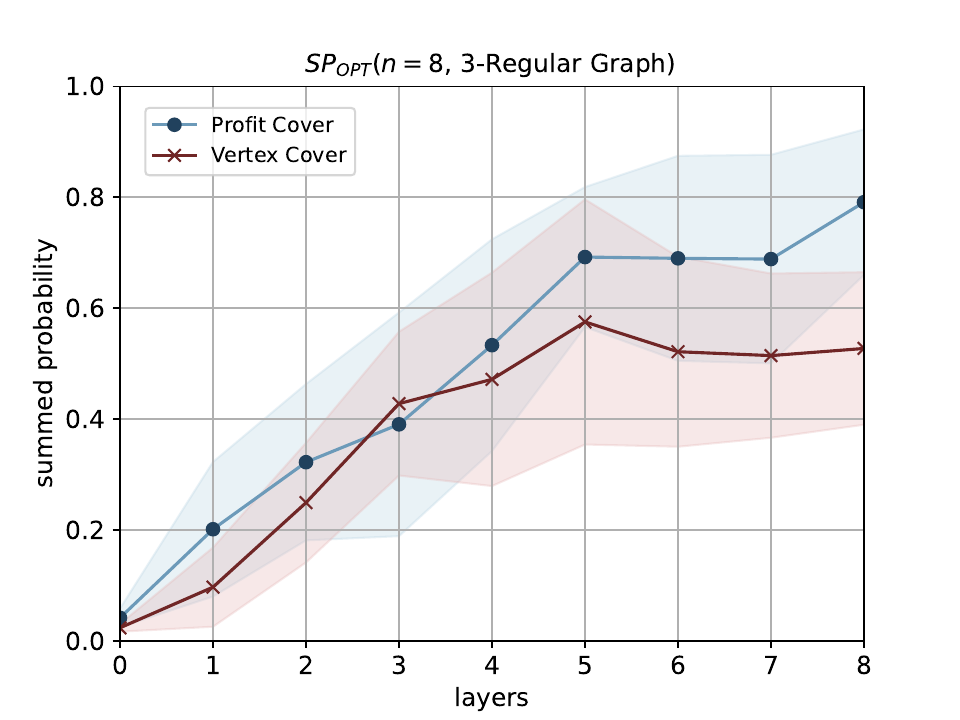} &
    \includegraphics[width=.33\textwidth]{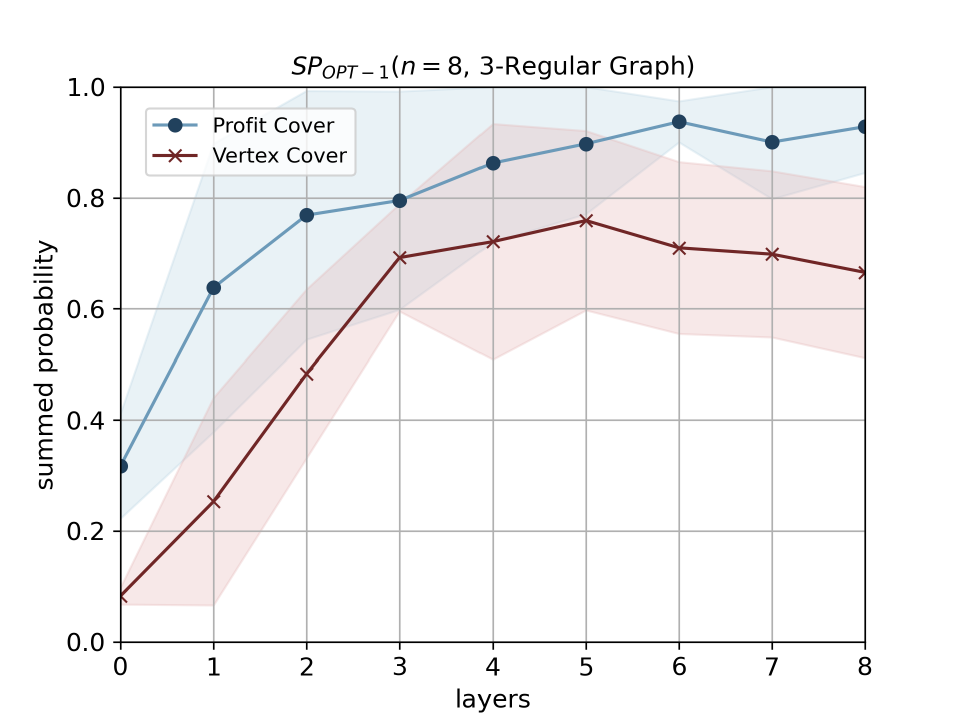} &
    \includegraphics[width=.33\textwidth]{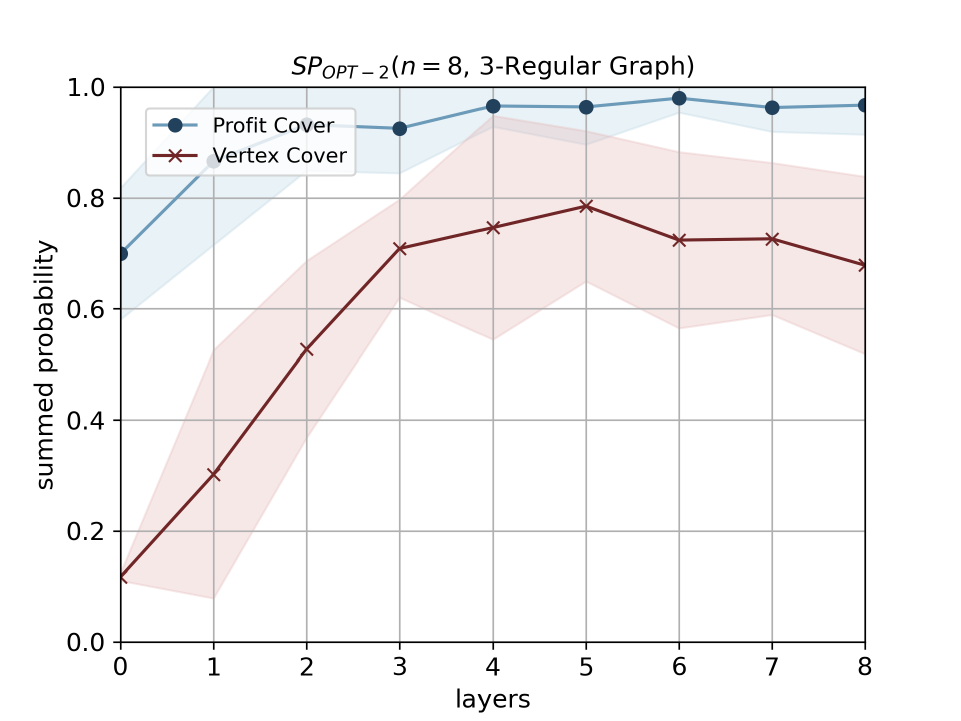} 
    \\
  \includegraphics[width=0.33\textwidth]{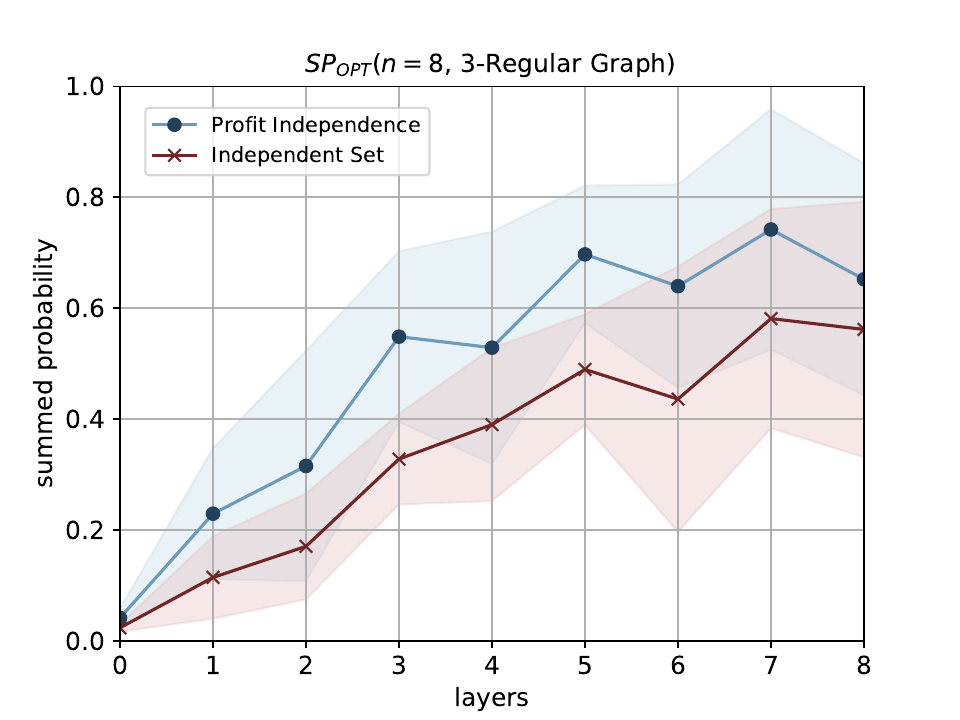} &
    \includegraphics[width=.33\textwidth]{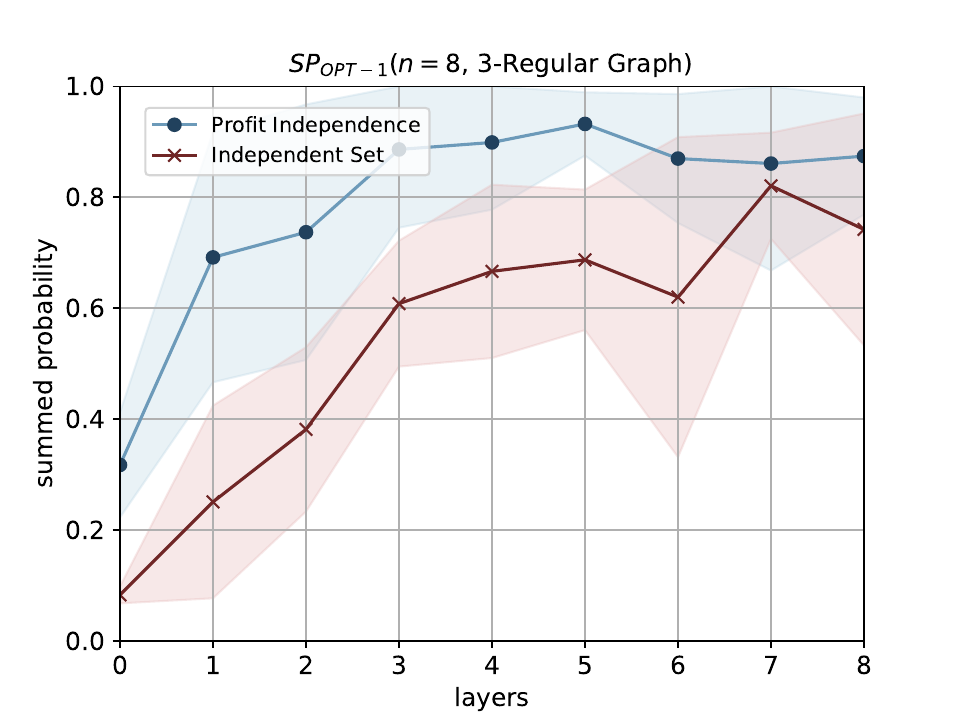} &
    \includegraphics[width=.33\textwidth]{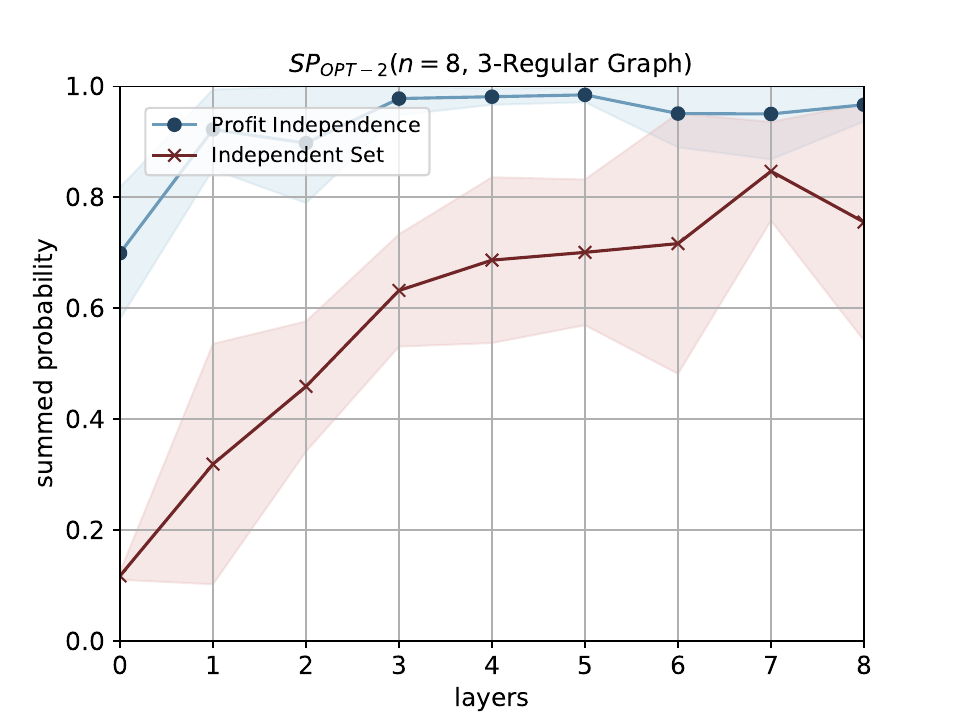} 
    \\
    \includegraphics[width=0.33\textwidth]{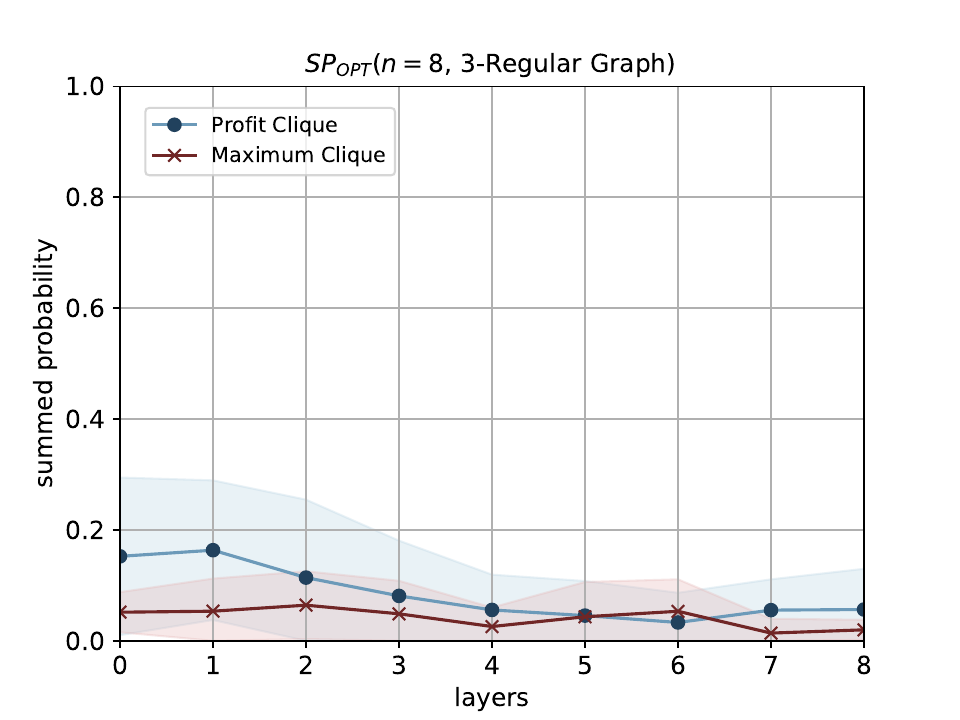} &
    \includegraphics[width=.33\textwidth]{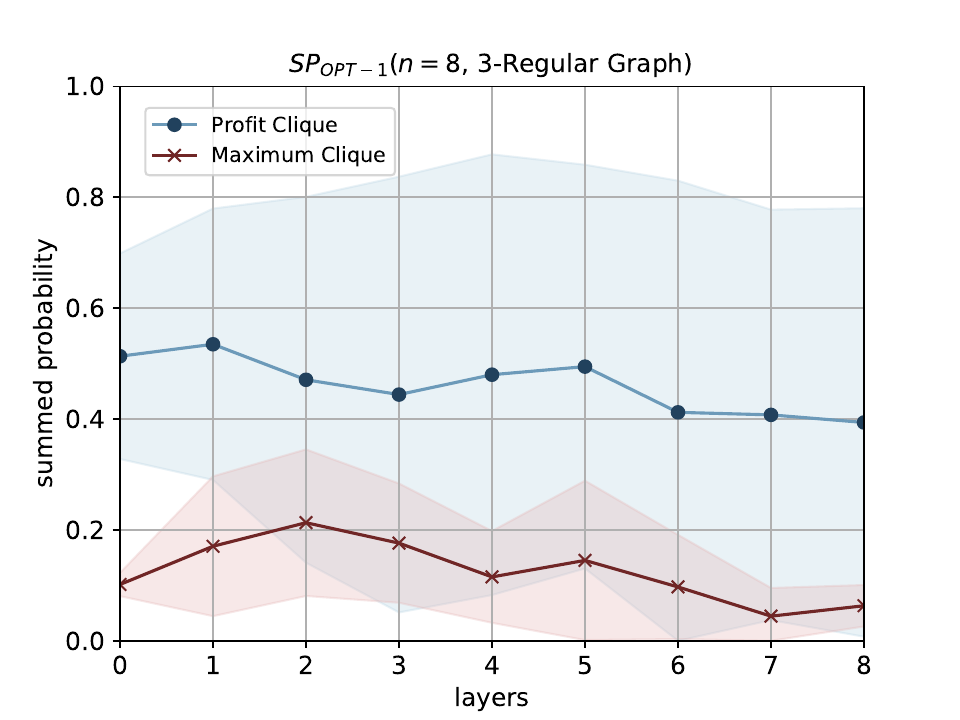} &
    \includegraphics[width=.33\textwidth]{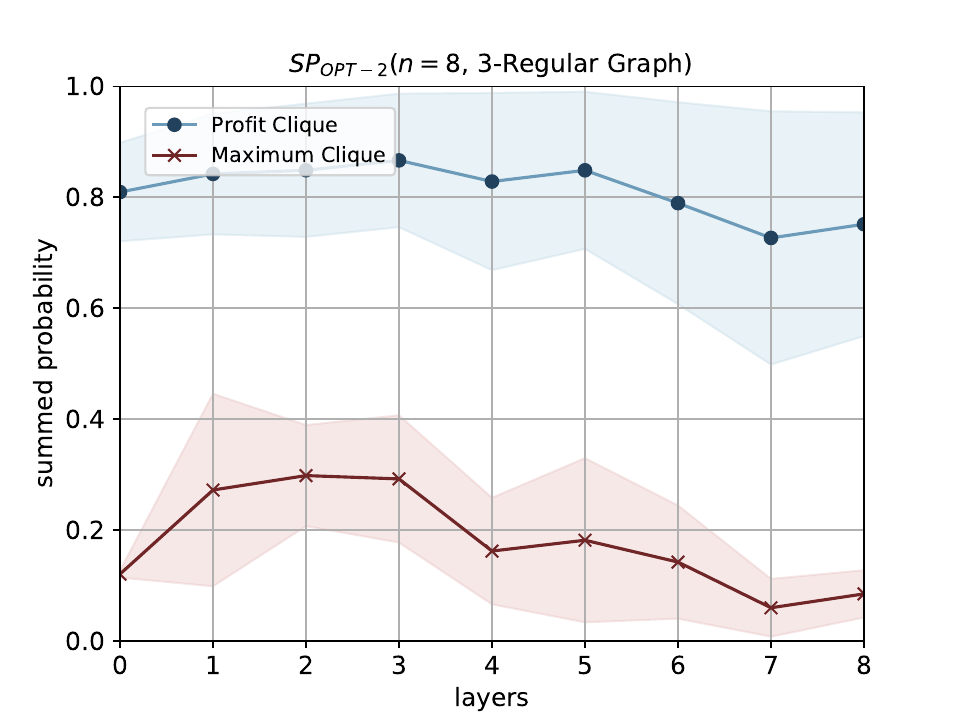} 
  \end{tabular}
  \caption{\textbf{Average summed probabilities of optimal solutions over ten 3-regular graphs obtained over eight layers.} The graphs or figures in the topmost row compare \textit{\textsc{MinVC}} and \textit{\textsc{MaxPC}}, the graphs or figures in the middle row compare \textit{\textsc{MaxIS}} and  \textit{\textsc{MaxPC}}, the  graphs or figures in the bottom row compare \textit{\textsc{MaxCl}} and \textit{\textsc{MaxPCl}}. The rows indicate the extent of near-optimality. Each graph represents the comparison between the profit version (blue) and the constrained version (red). Data points correspond to the average value for each layer, with shaded regions indicating one standard deviation away from the mean.}
    \label{fig:vc-is-cl-3-reg-summed-probs}
\end{figure*}
Fig.~\ref{fig:vc-is-cl-3-reg-summed-probs} shows the near-optimality analysis for 3-regular graphs for all six problems; \textit{\textsc{MinVC}} and \textit{\textsc{MaxPC}} in the first row, \textit{\textsc{MaxIS}} and \textit{\textsc{MaxPI}} in the middle row, and, \textit{\textsc{MaxCl}} and \textit{\textsc{PCl}} in the bottom row. 
\begin{figure*}[!htbp]
    \centering
    \includegraphics[width=0.95\linewidth]{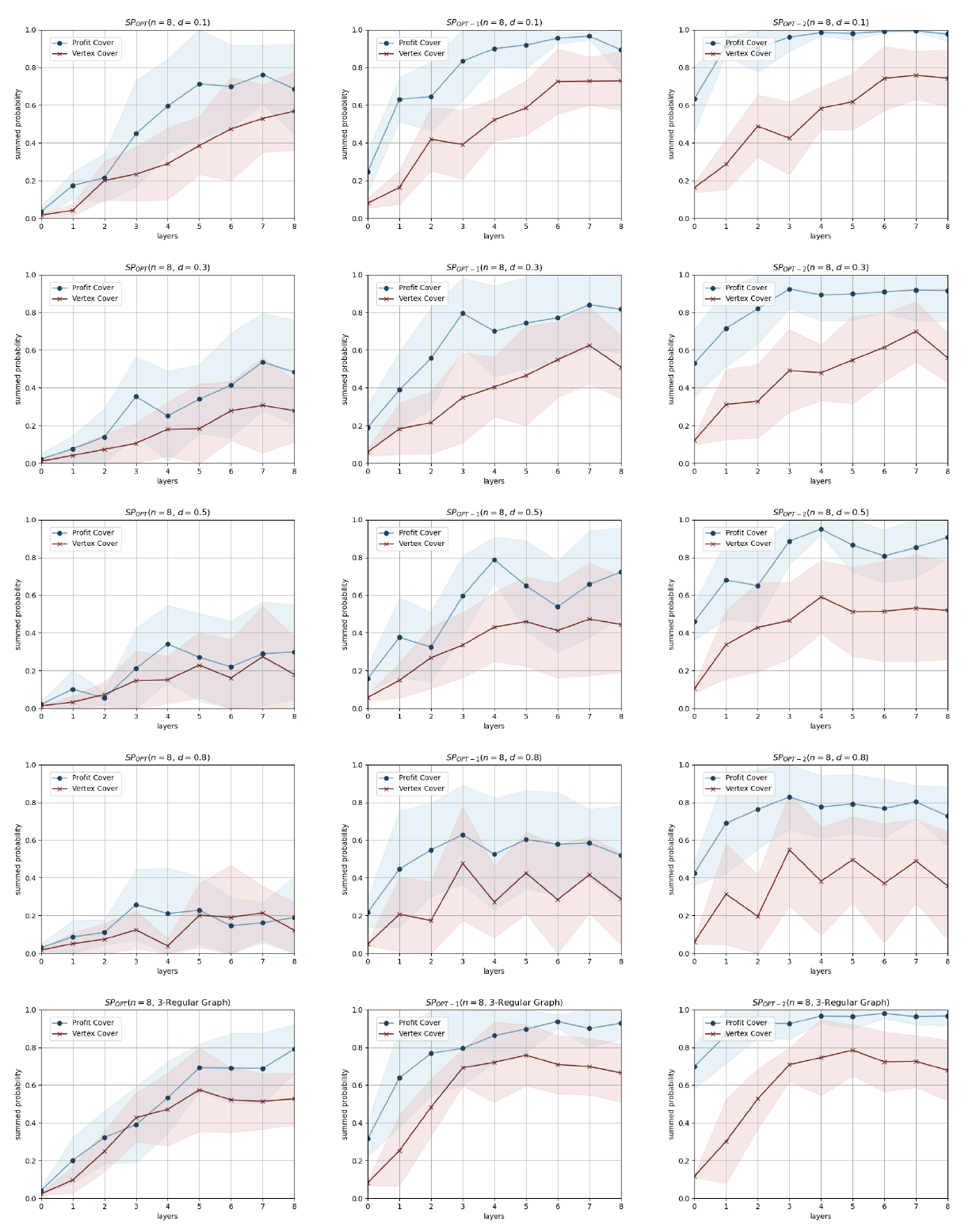}
    \caption{Probabilities of optimal and near-optimal solutions obtained over eight layers for \textit{\textsc{MinVC} and \textit{\textsc{MaxPC}}}. The figures or graphs in the first column show the summed probability of all the optimal solutions averaged over ten graphs with $n = 8$. The figures or graphs in the  second column depict the summed probability of obtaining the optimal solution and the second best solution. The figures or graphs in the  third column indicates the summed probability of the optimal, second best and the third best solutions. Each row indicates a different edge density with the last row showing the results for 3-regular graphs.}
    \label{fig:vc-opt-summed-probs}
\end{figure*}


\begin{figure*}[!th]
\centering
  \begin{tabular}{@{}cccc@{}}
    \includegraphics[width=0.33\textwidth]{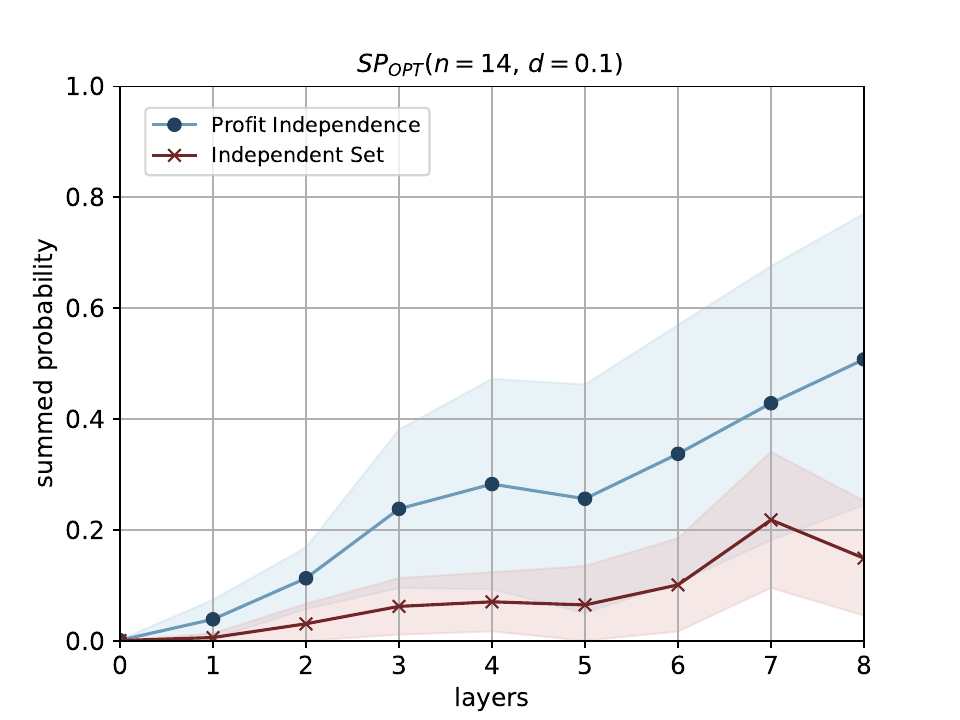} &
    \includegraphics[width=.33\textwidth]{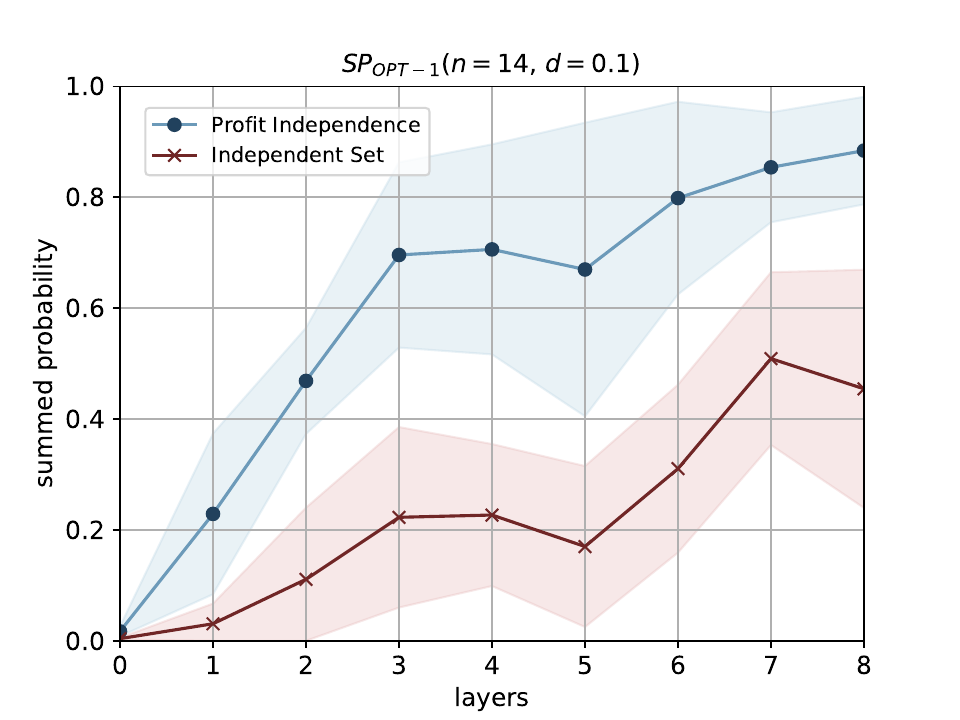} &
    \includegraphics[width=.33\textwidth]{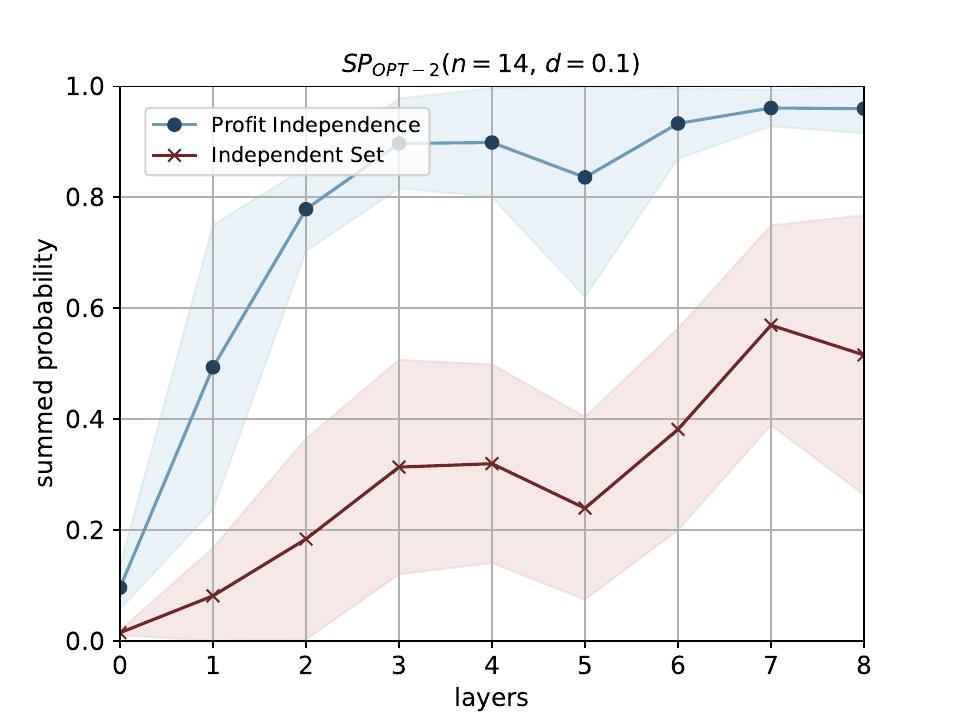} 
    \\
    \includegraphics[width=0.33\textwidth]{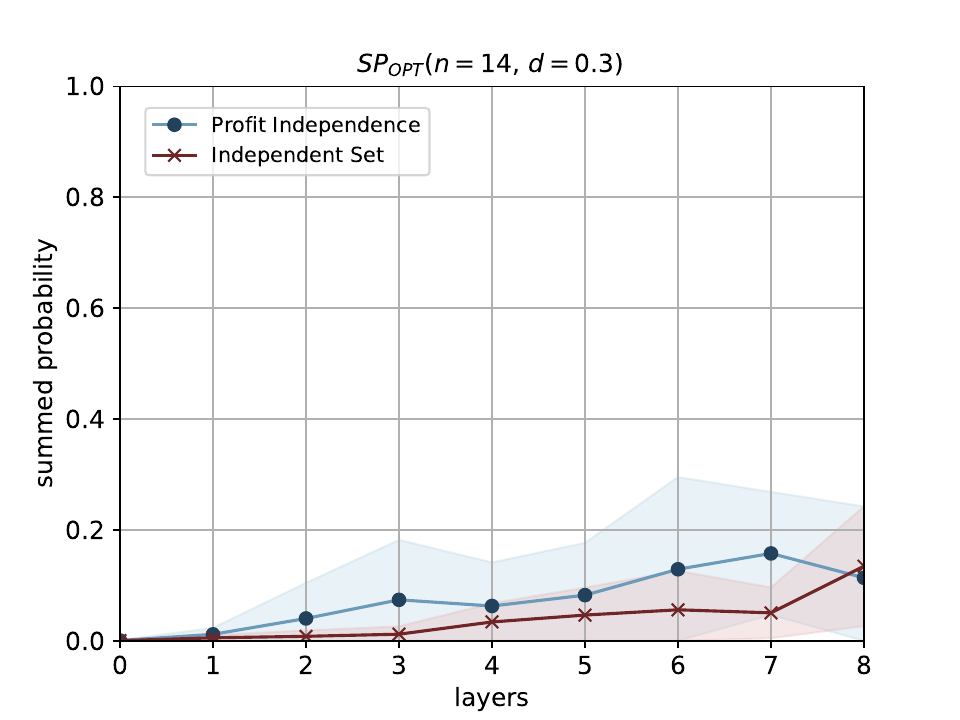} &
    \includegraphics[width=.33\textwidth]{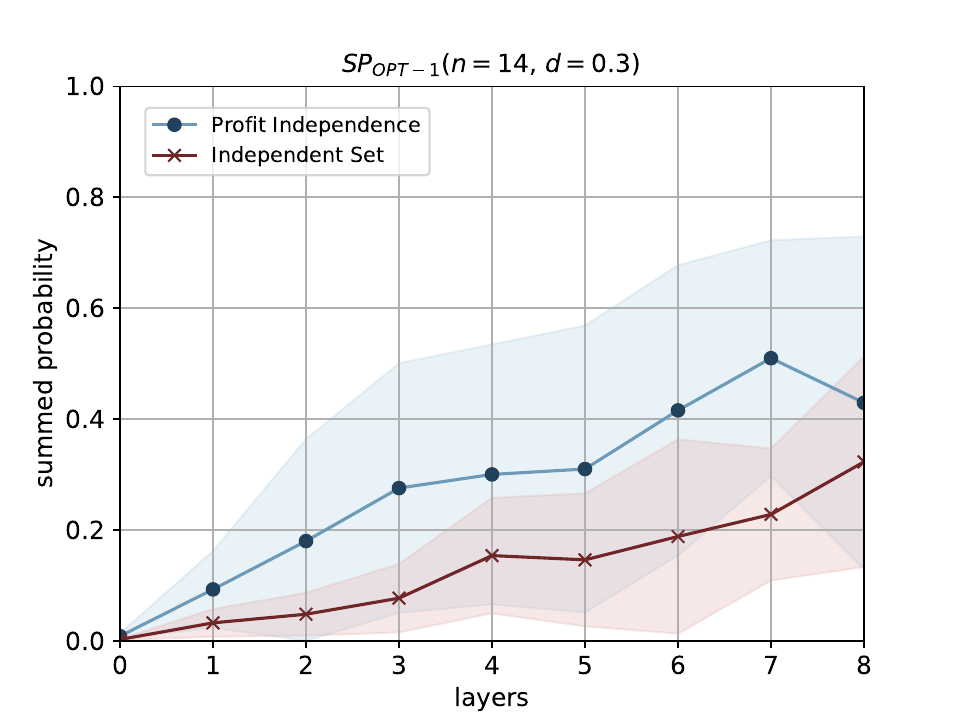} &
    \includegraphics[width=.33\textwidth]{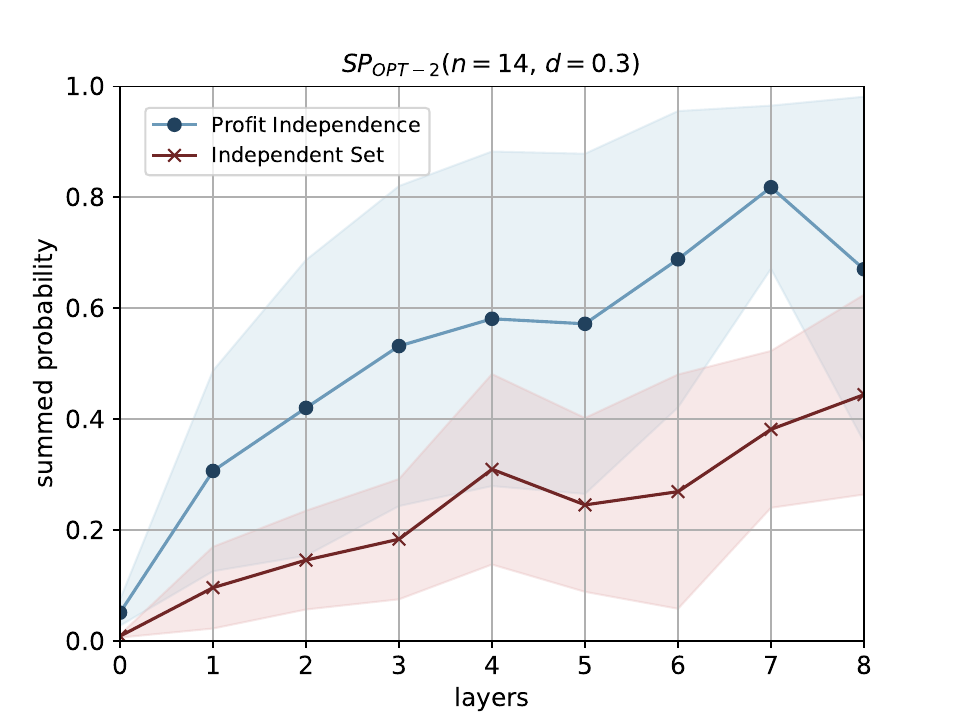} 
    \\
    \includegraphics[width=0.33\textwidth]{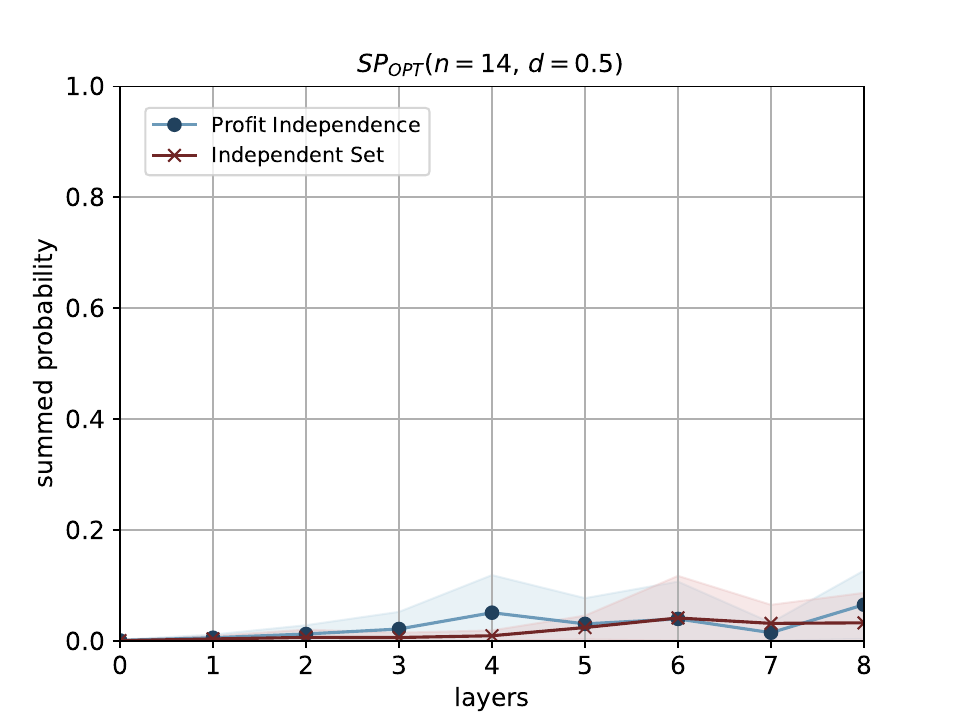} &
    \includegraphics[width=.33\textwidth]{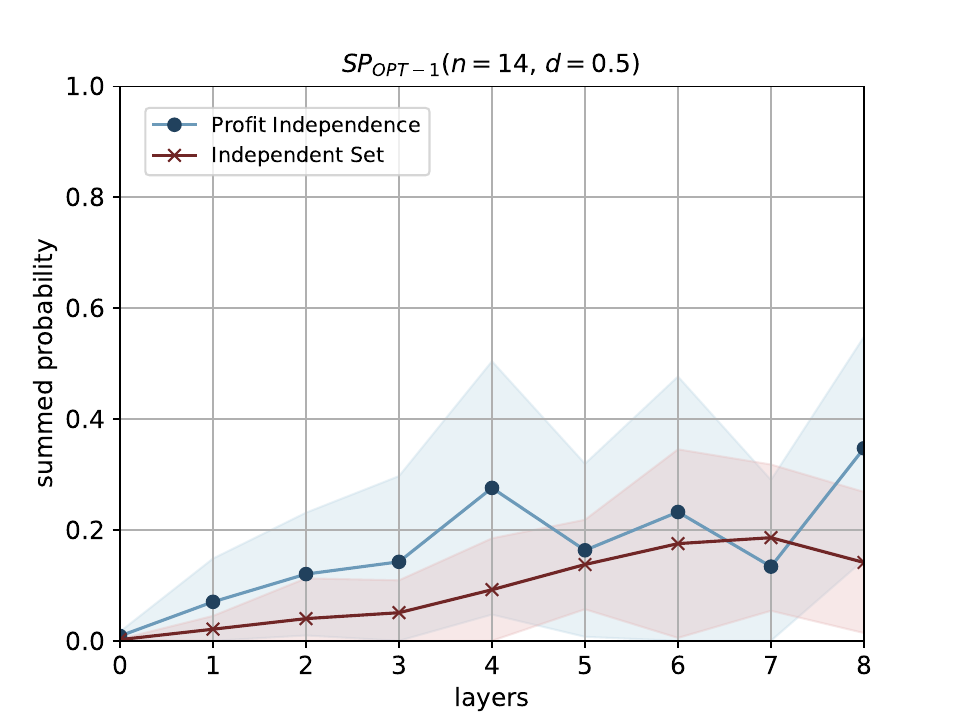} &
    \includegraphics[width=.33\textwidth]{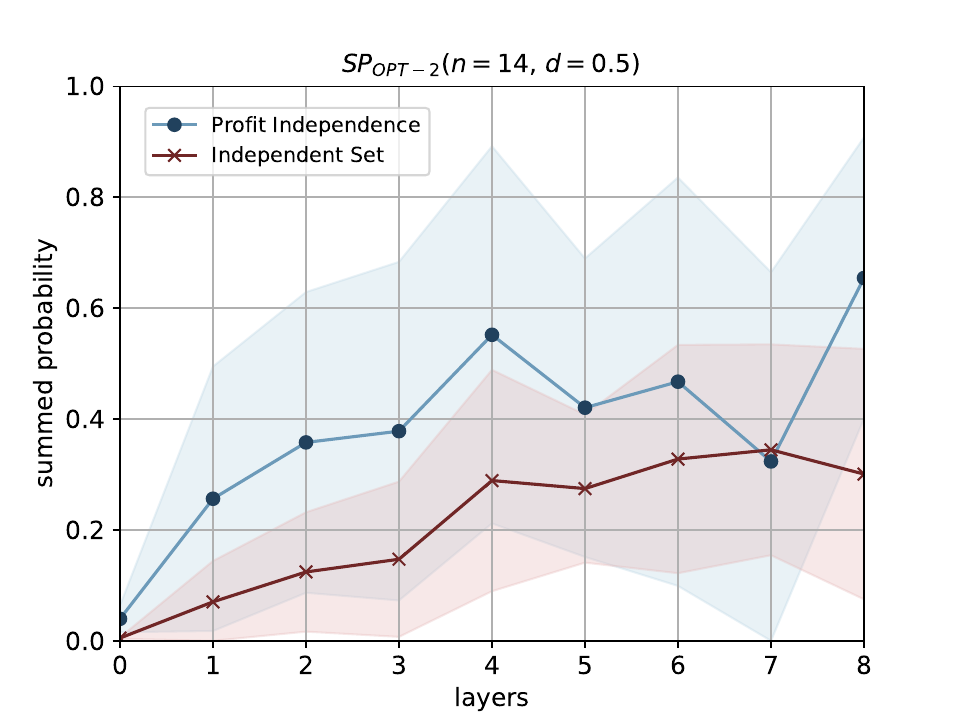} 
    \\
    \includegraphics[width=0.33\textwidth]{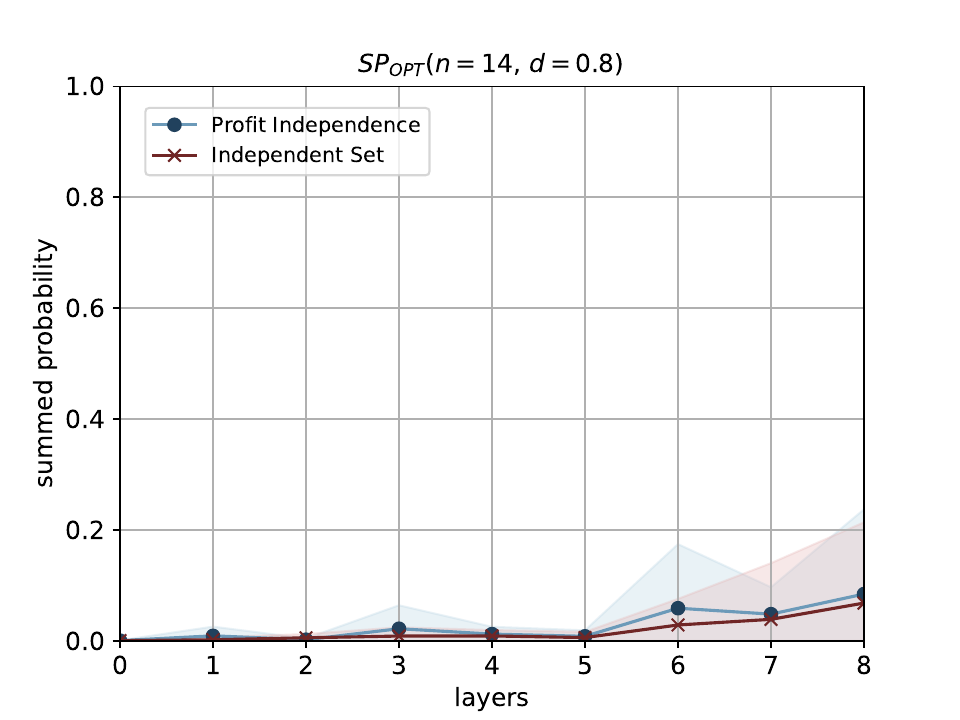} &
    \includegraphics[width=.33\textwidth]{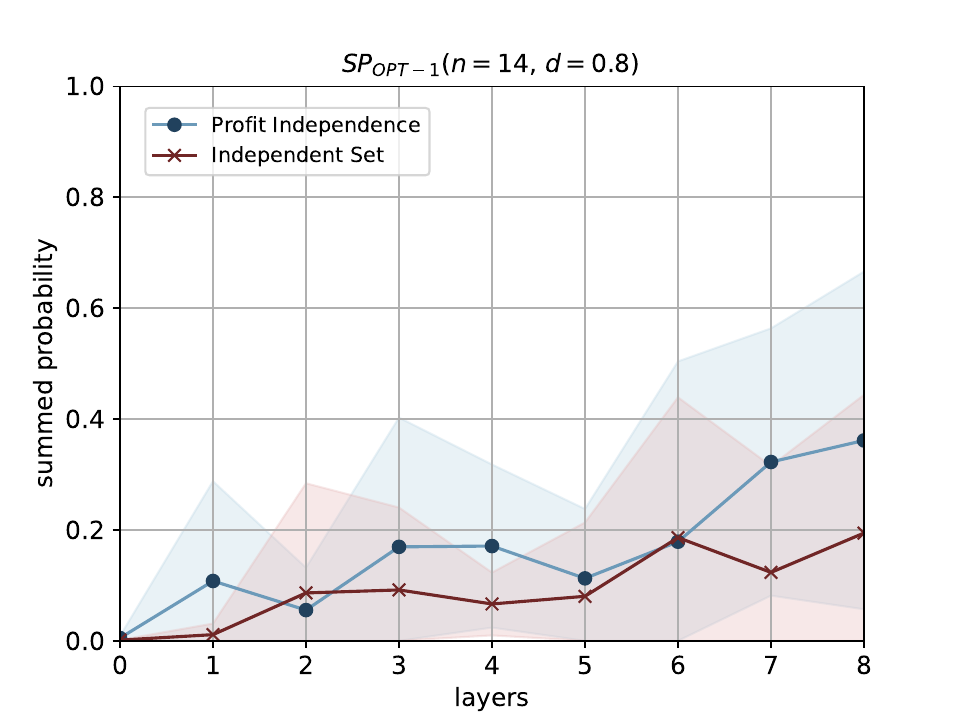} &
    \includegraphics[width=.33\textwidth]{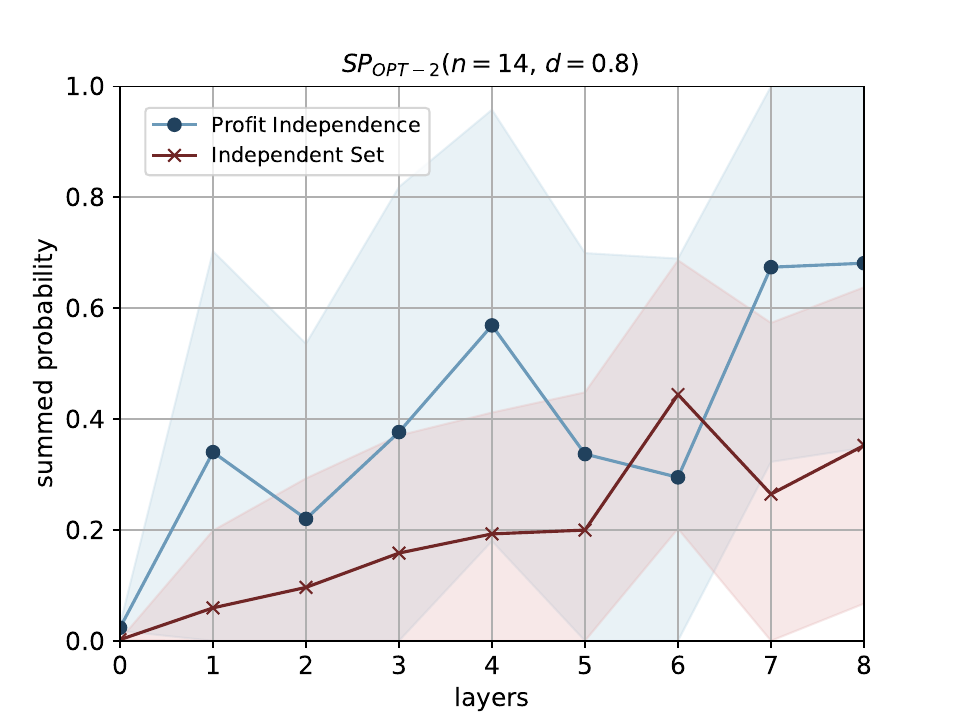}
  \end{tabular}
  \caption{\textbf{Summed optimal and near-optimal probabilities obtained over eight layers averaged over ten graphs for \textit{\textsc{MaxPI}} and \textit{\textsc{MaxIS}} for $n=14$.}The graphs or figures presented here can be interpreted in a manner similar to Figure~\ref{fig:vc-opt-summed-probs}. Each row represents a different edge density of the graphs. The figures or graphs in the leftmost column display the summed optimal probabilities, while those in the middle and right columns show the near-optimal probabilities.}
  \label{fig:is-pi-14}
\end{figure*}
\subsection{Approximation Ratio}

Fig.~\ref{fig:approx-ratios} shows the approximation ratios for solving \textit{\textsc{MaxPC}} for $n=7$ and $n=8$ over eight layers and ten graphs calculated using PennyLane. The approximation ratio is calculated using:
\begin{equation}
\label{eq:ar}
   r = \frac{|\text{Expectation value obtained} +\Delta_{\text{PC}}|}{|\text{Optimal Profit}|} 
\end{equation}

We show the results starting with layer $0$, which has no cost or mixer unitaries, and shows the expectation value of an equal superposition state added to the offset $\Delta_{\text{PC}}$ (note the addition of $\Delta_{\text{PC}}$ in Equation~\ref{eq:ar}). Adding the offset provides the exact expectation value of the profit obtained.

Figure~\ref{fig:approx-ratios-2} shows the average \textit{\textsc{MaxPC}} performance of 10 3-regular graphs each for $n=\{8, 10, 12, 14\}$ on PennyLane and for $n=\{20, 30, 40, 70\}$ calculated on QTensor. 

\subsection{Interpretation on large problems with QTensor}
Fig.~\ref{fig:cost-qtensor} shows the cost function evaluations for a random 3-regular graph for $n \in \{30, 40, 50, 70\}$. While expectation values indicate the closeness to the optimal solution, with profit cover, we have the advantage that the expectation value corresponds to the calculated profit, with the minimum expectation value equal to the maximum profit. For example, consider the  results for layer $3$  for {\sc MaxPC}. The cost value is close to $-15$, suggesting a high likelihood of a profit cover with a profit of $15$, and therefore, a high likelihood of a vertex cover of size at most $|E|-15$. 
\begin{figure*}[!th]
\centering
  \begin{tabular}{@{}cccc@{}}
    \includegraphics[width=0.5\textwidth]{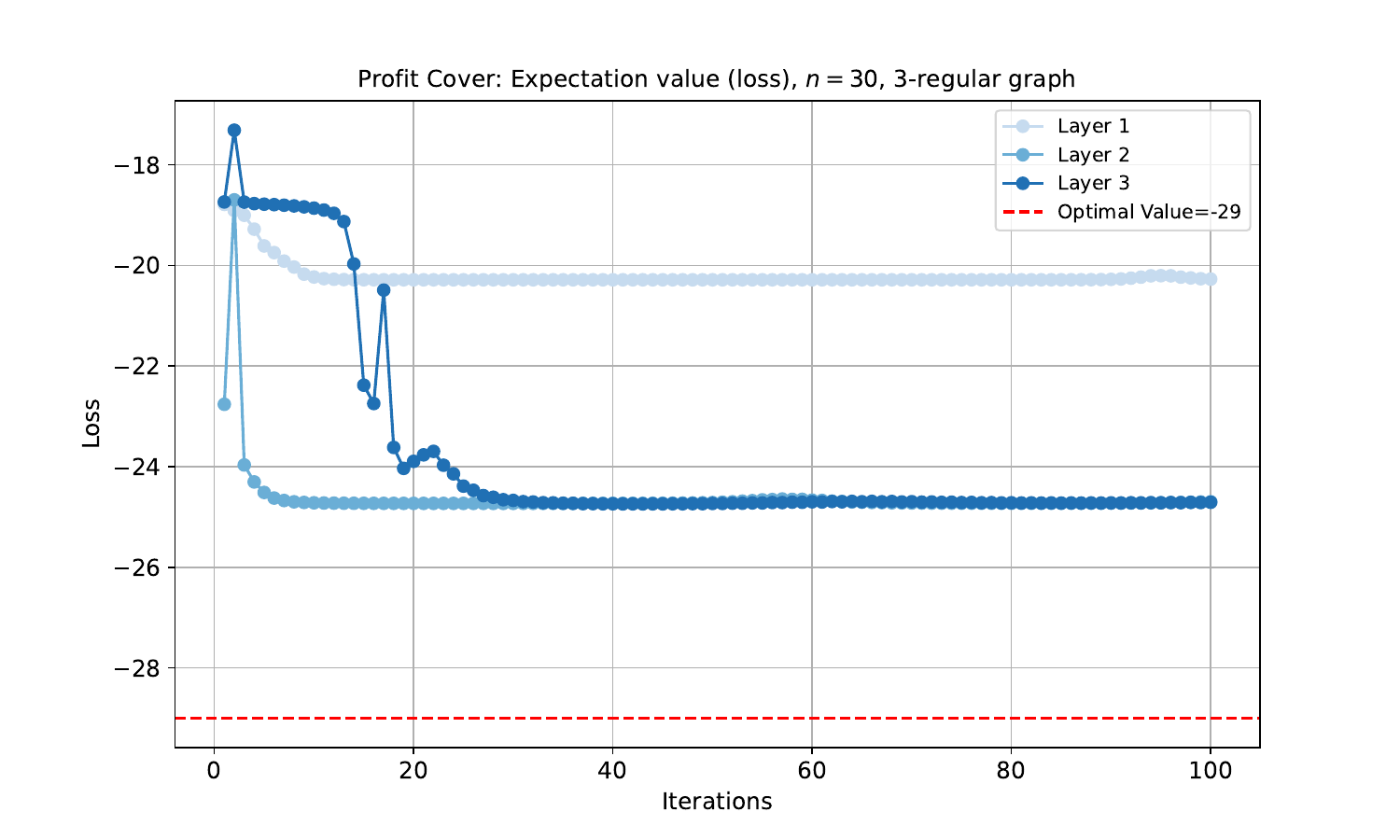} &
    \includegraphics[width=.5\textwidth]{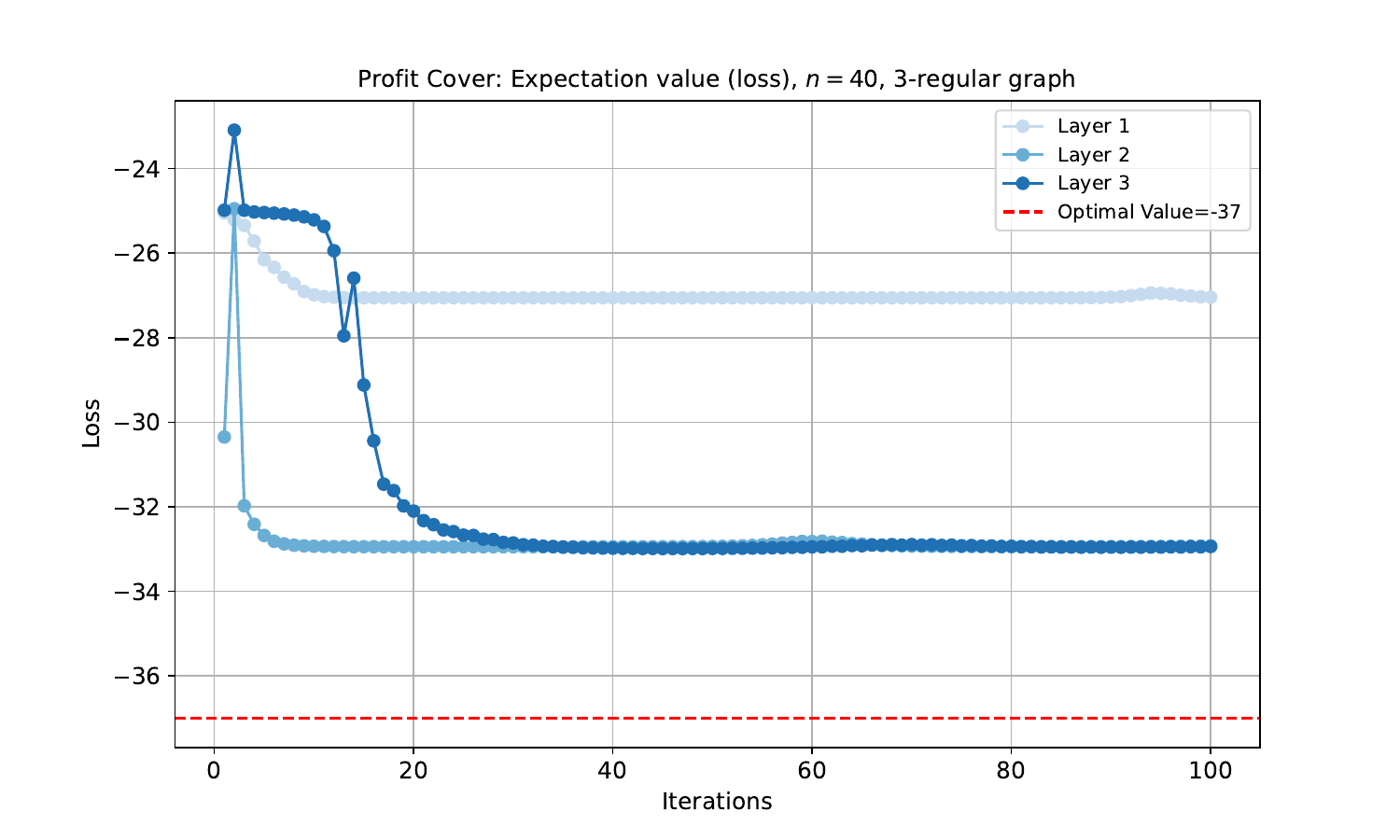} \\
    \includegraphics[width=.5\textwidth]{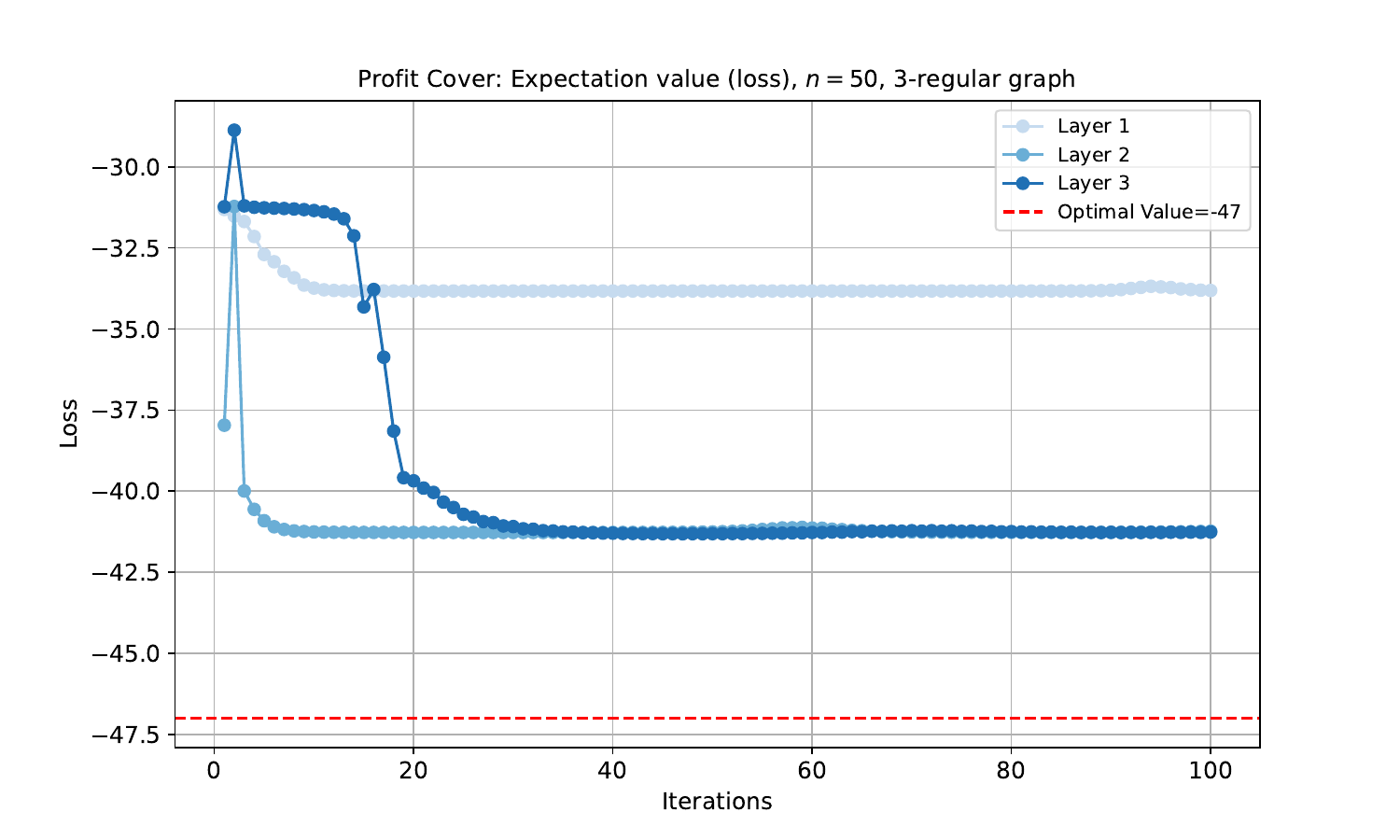} &
    \includegraphics[width=.5\textwidth]{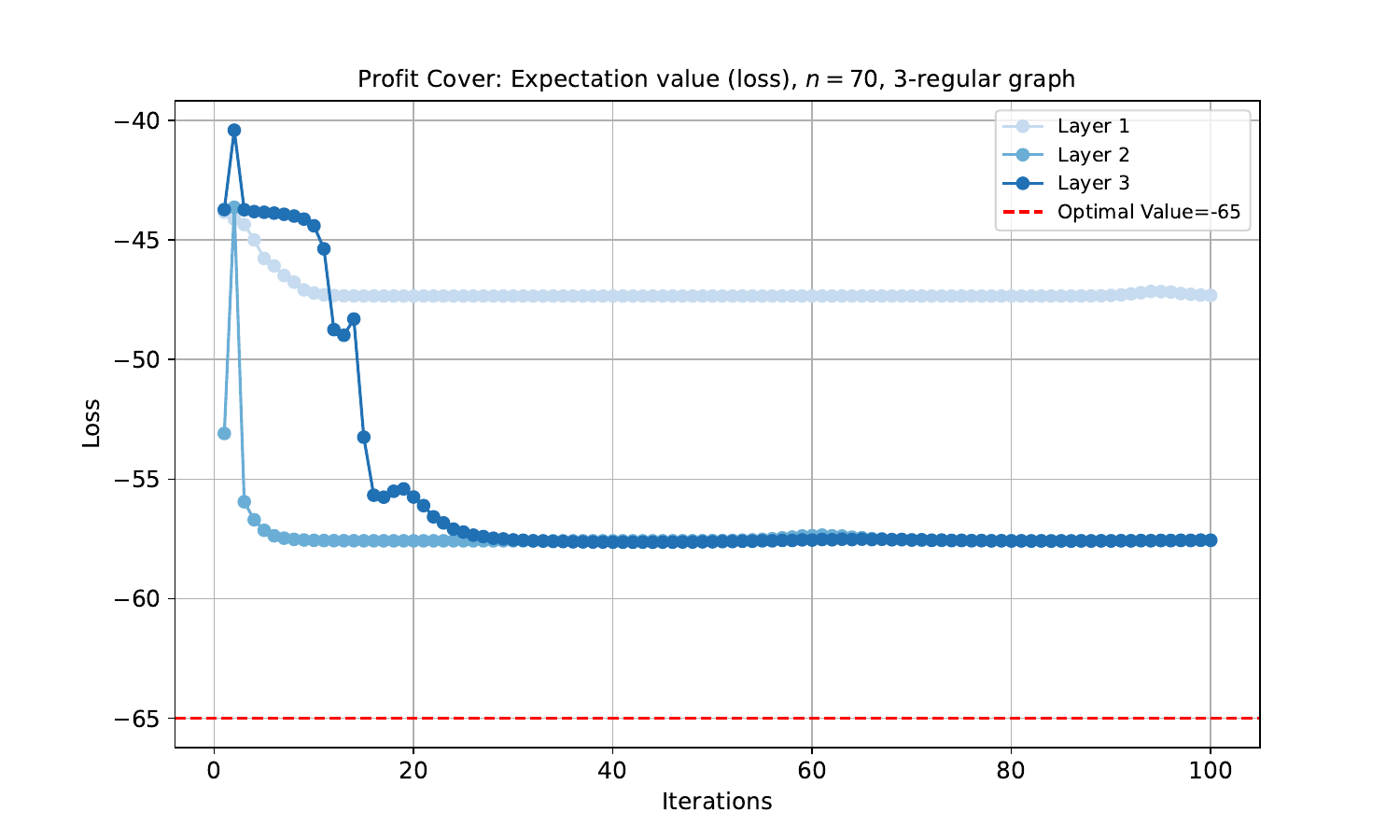} \\
    \includegraphics[width=.5\textwidth]{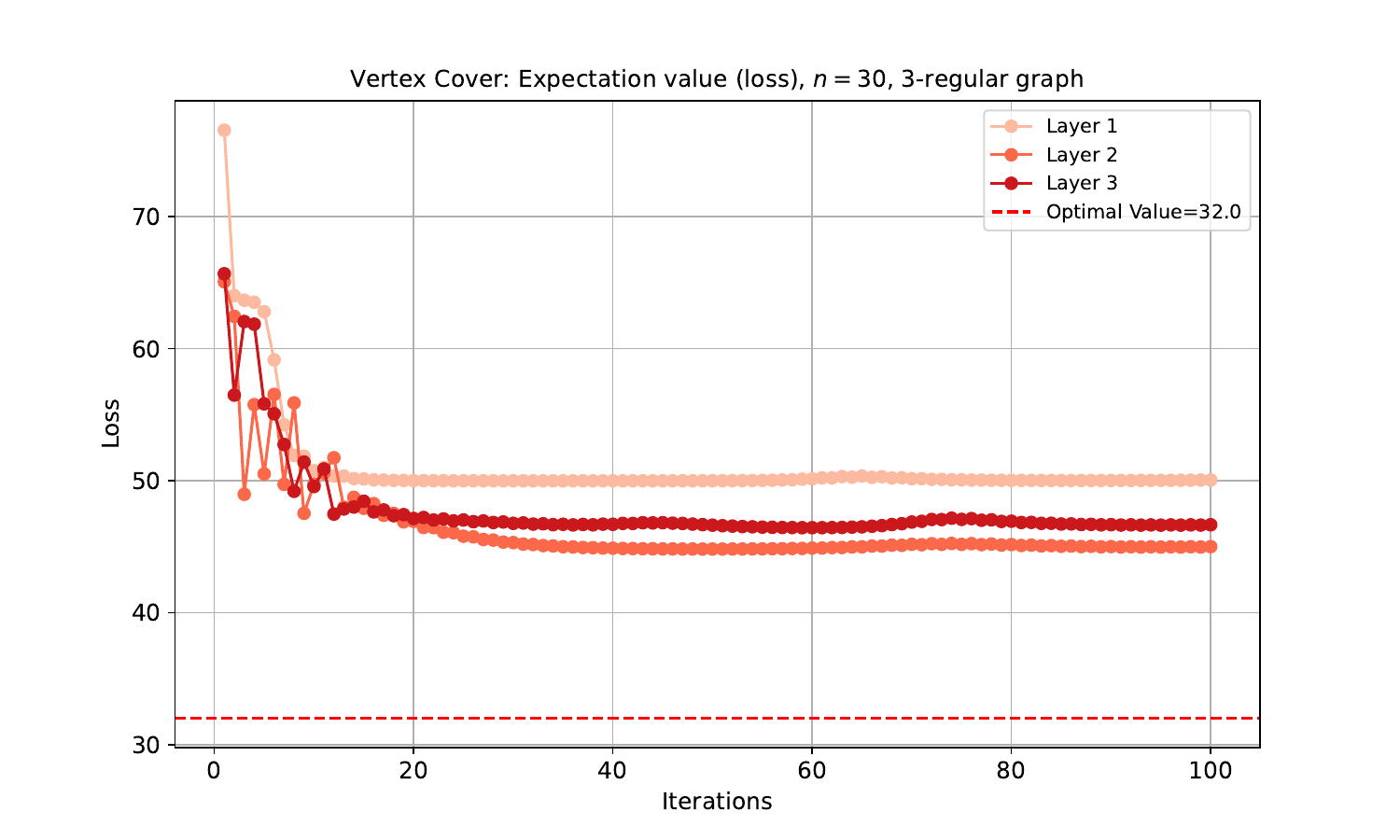} &
    \includegraphics[width=.5\textwidth]{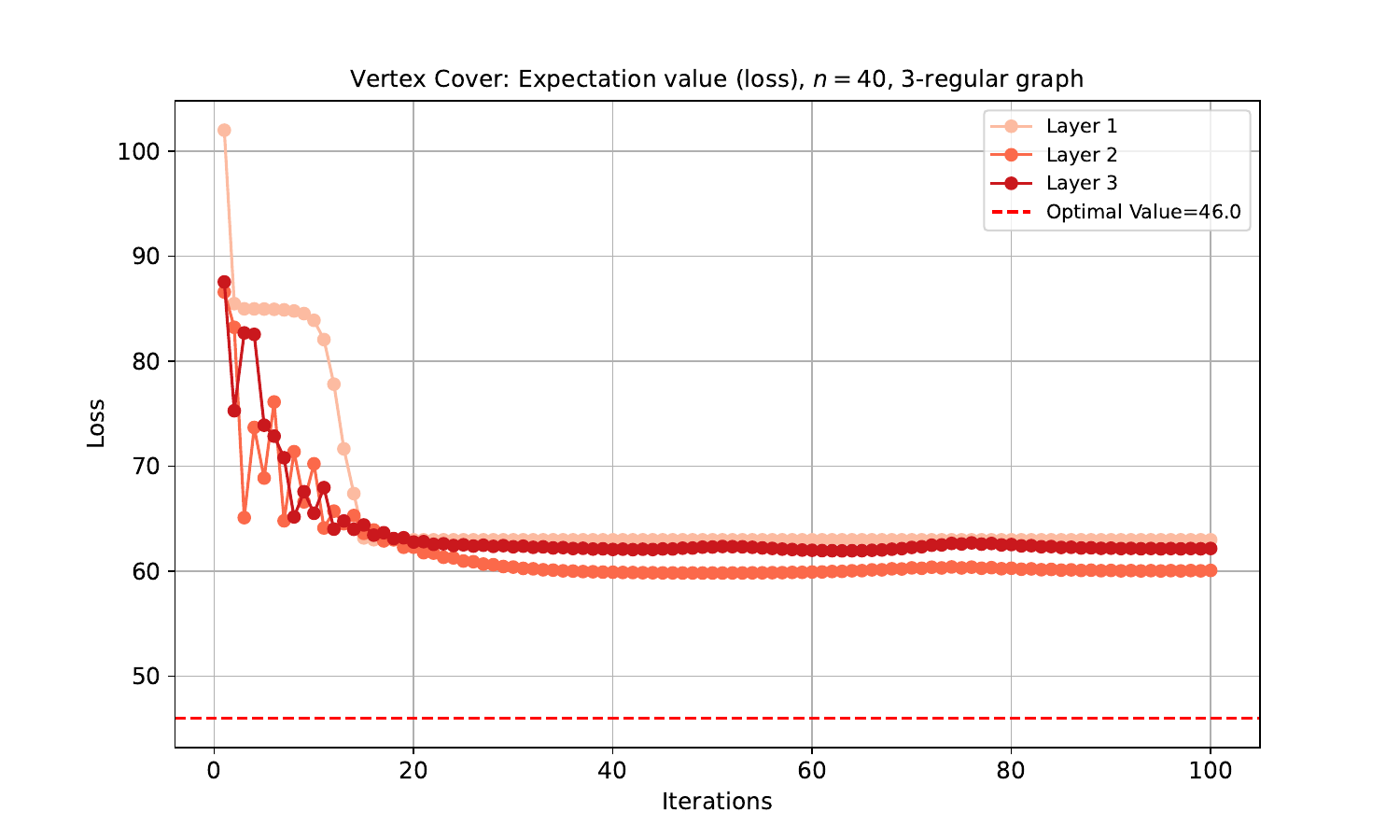} \\
    \includegraphics[width=.5\textwidth]{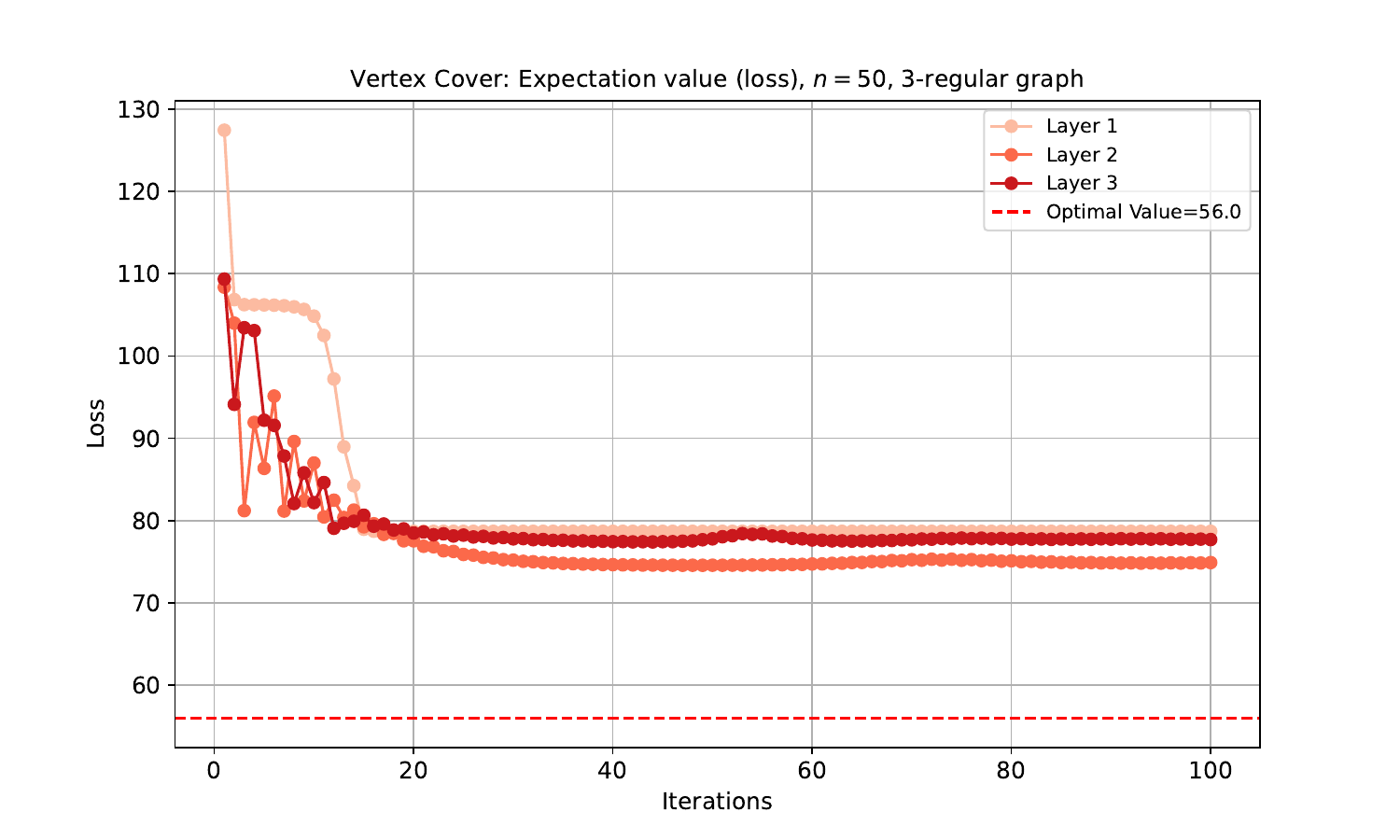} &
    \includegraphics[width=.5\textwidth]{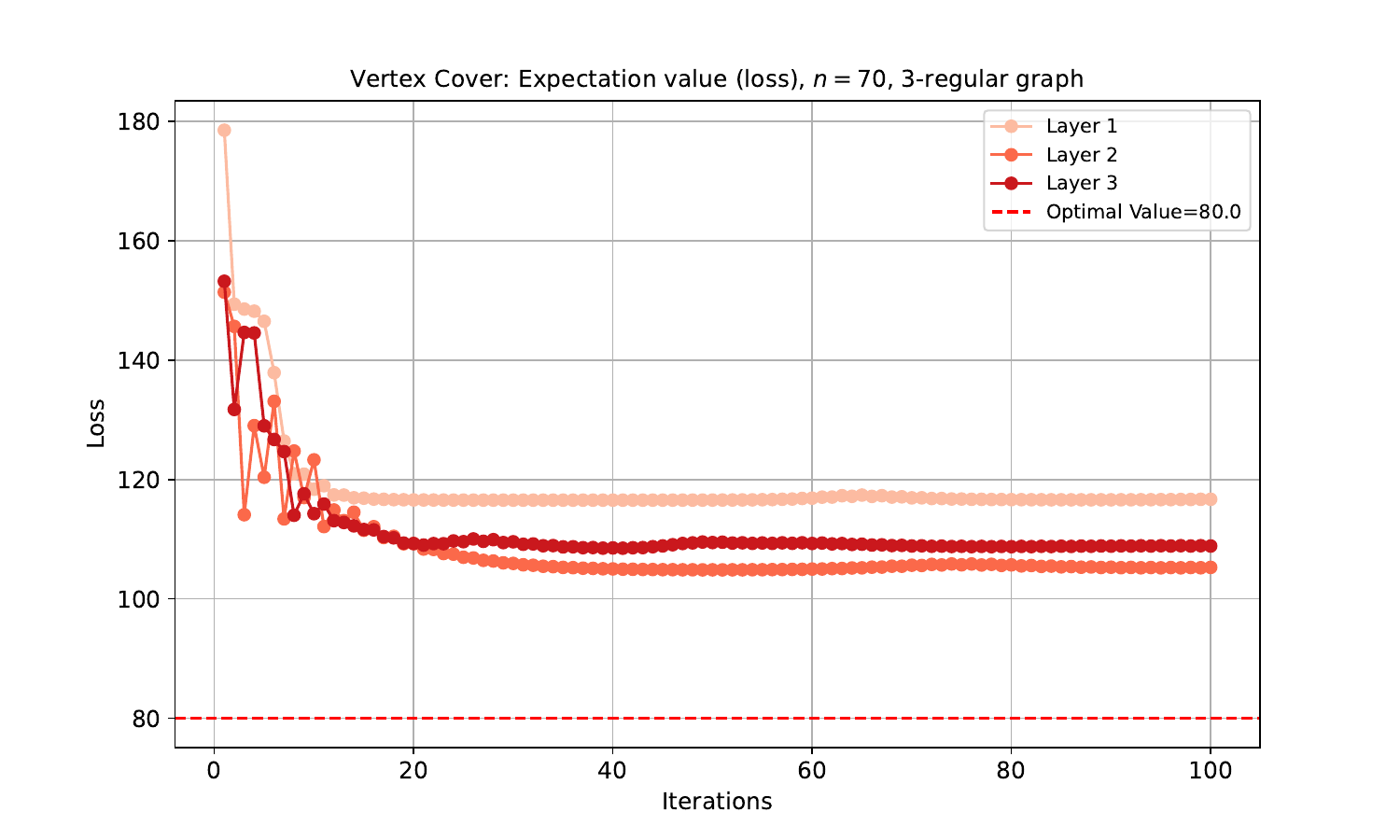} 
  \end{tabular}
  \caption{Cost optimization over $n \in \{30, 40, 50, 70\}$, for $3$-regular graphs on QTensor for \textit{\textsc{MaxPC}} and \textit{\textsc{MinVC}}.}
  \label{fig:cost-qtensor}
\end{figure*}

\begin{figure*}[!th]
    \centering
    \begin{tabular}{@{}cccc@{}}
    \includegraphics[width=0.5\textwidth]{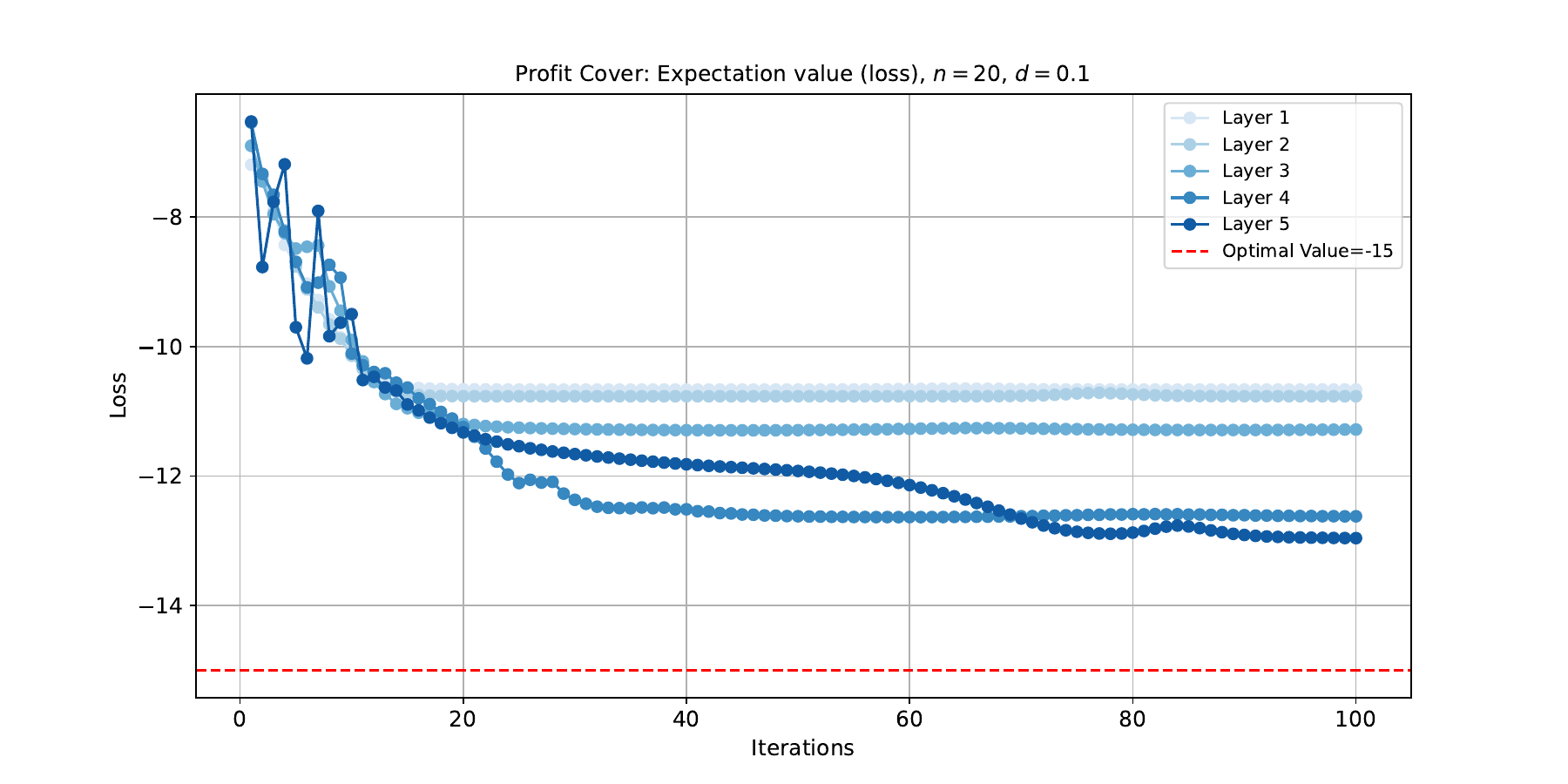} &
    \includegraphics[width=0.5\textwidth]{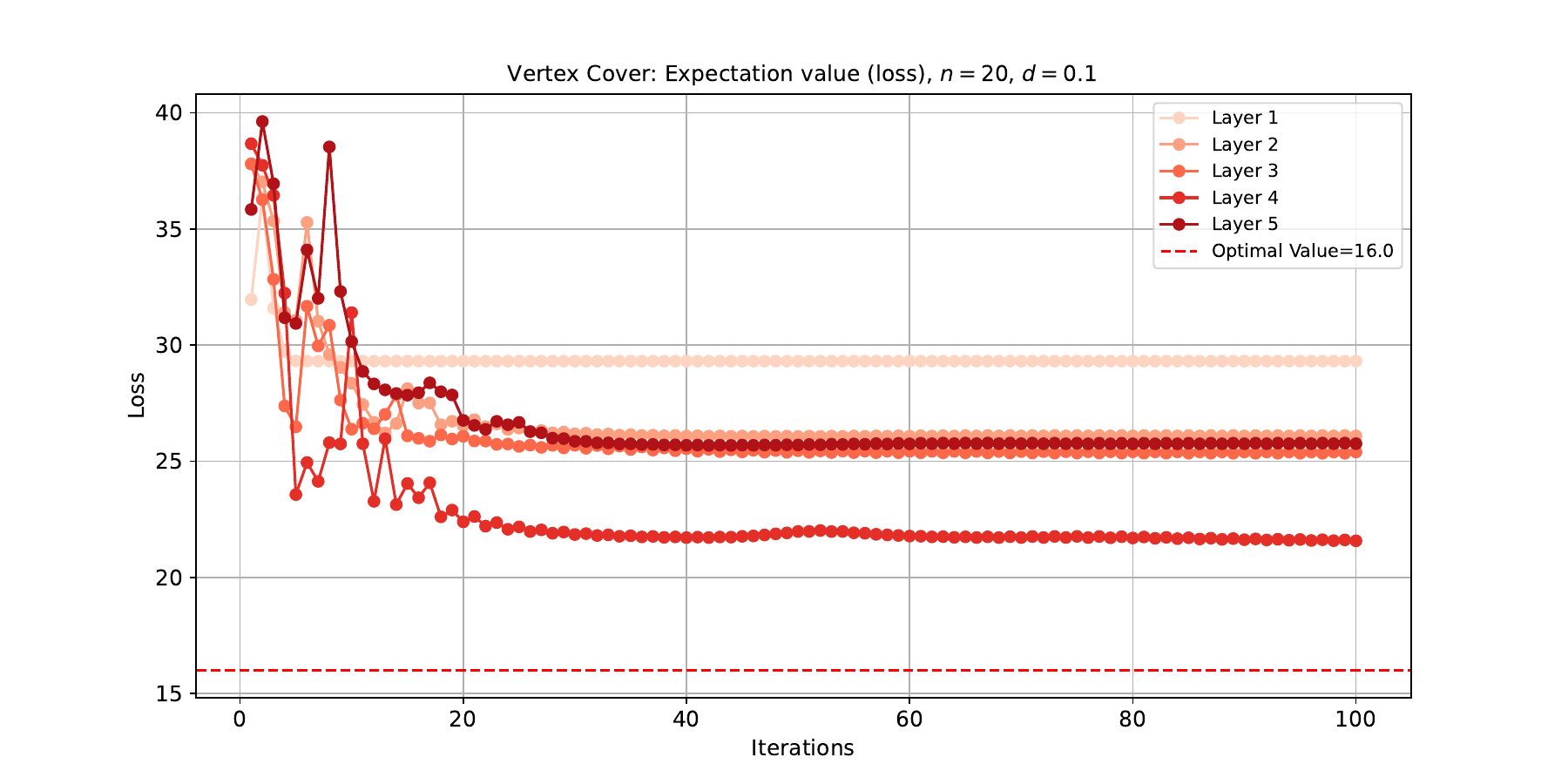} \\
    \includegraphics[width=0.5\textwidth]{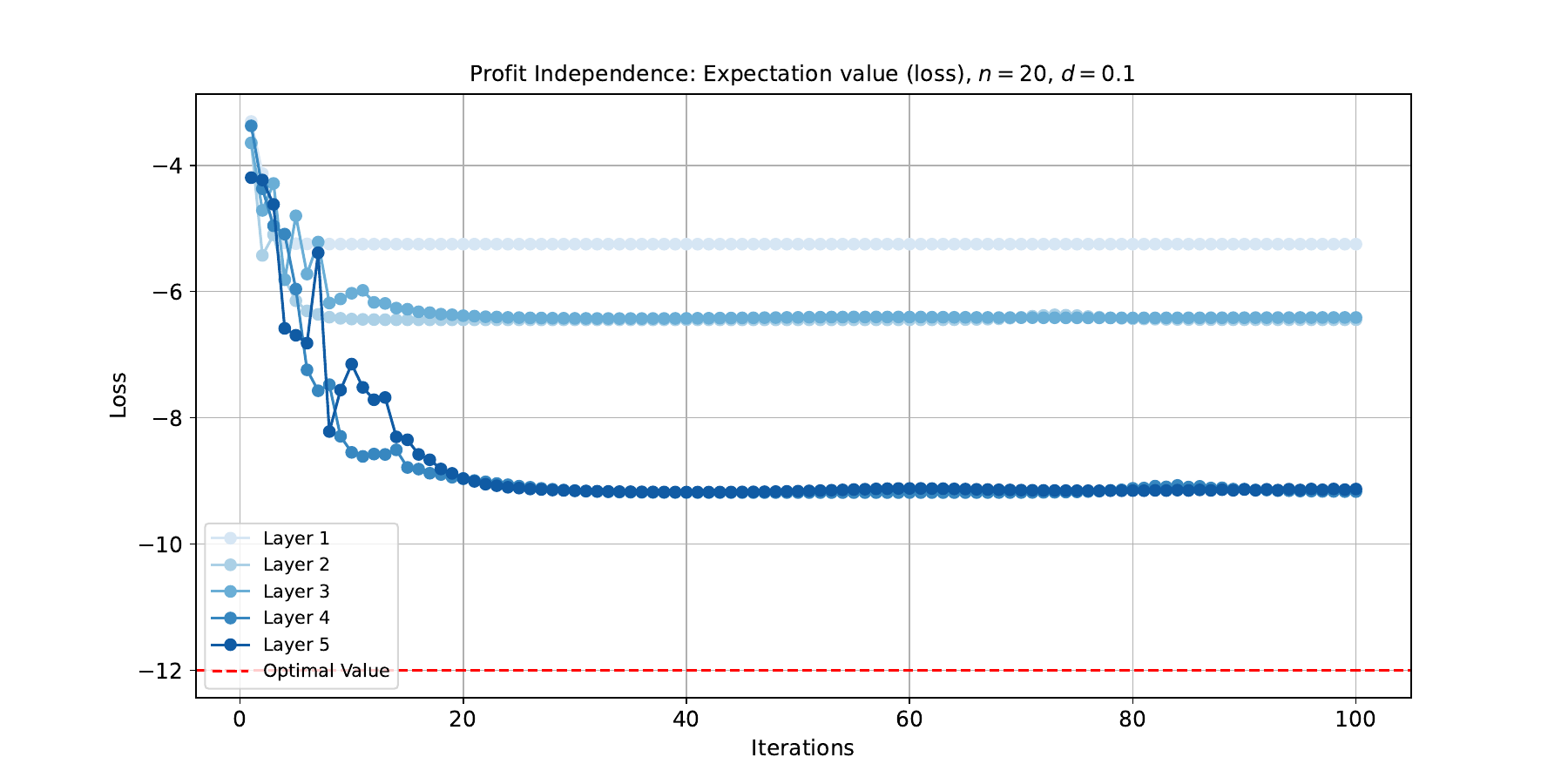} &
    \includegraphics[width=0.5\textwidth]{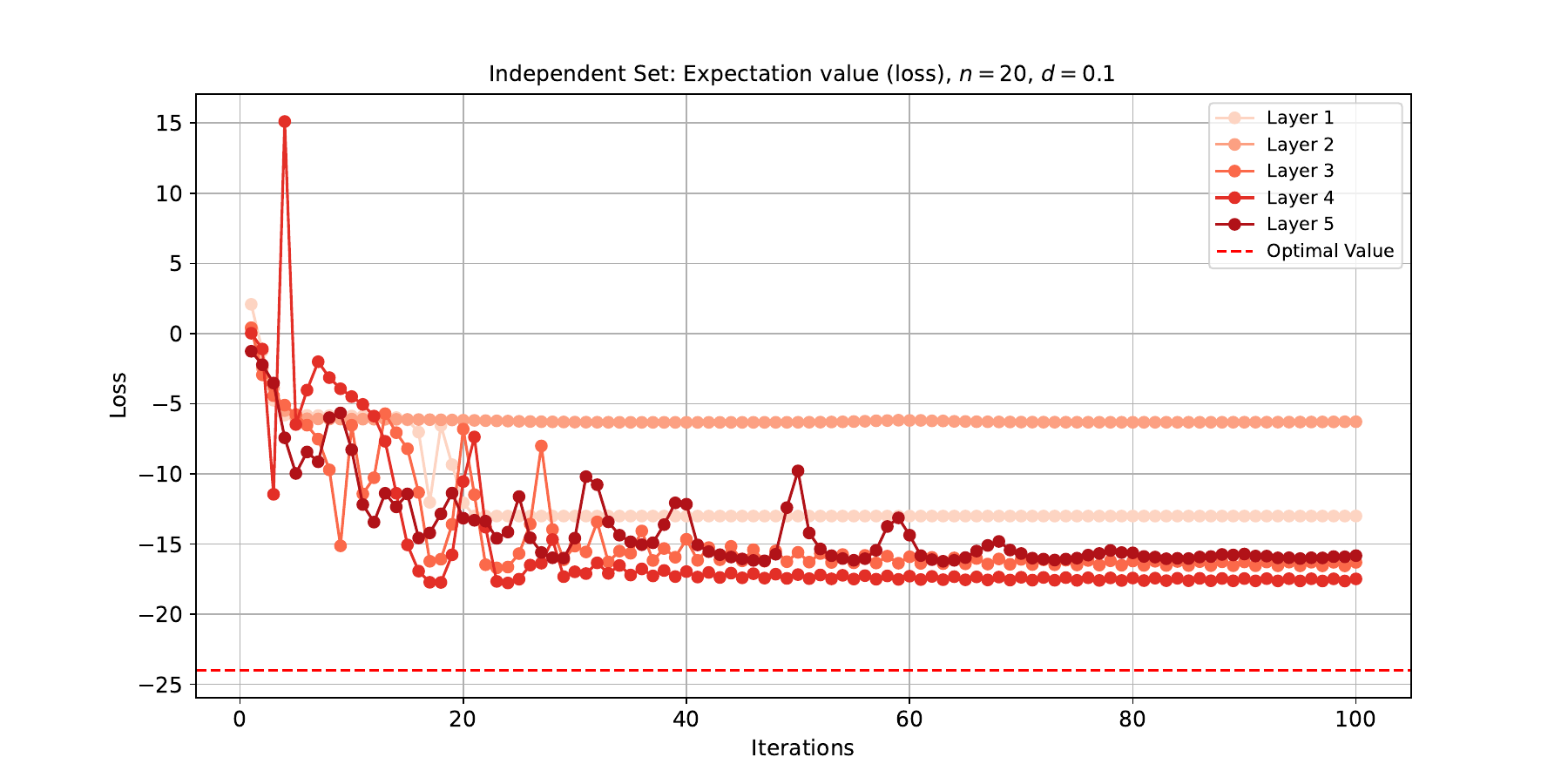} \\

  \end{tabular}
    \caption{Cost optimization for $n=20$ for a sparse (edge probability of 0.1) graph, on a graph over five layers. The figures or graphs in the top row present the expectation value minimized over 100 iterations for \textit{\textsc{MaxPC}} and \textit{\textsc{MinVC}}, while the figures or graphs in the bottom row show the results for \textit{\textsc{MaxPI}} and \textit{\textsc{MaxIS}}. Note that although \textit{\textsc{MaxPC}}, \textit{\textsc{MaxPI}}, and \textit{\textsc{MaxIS}} are maximization problems, the results presented correspond to the minimization versions of these problems.}
    \label{fig:vc-is-20-sparse-5layers}
\end{figure*}
\section{Discussion}\label{sec:discussion}
Our findings provide valuable insights into formulating and solving constrained optimization problems with vanilla QAOA. 
%
 Solving the cost Hamiltonian for constrained problems with QAOA requires fine-tuning of penalty parameters in addition to optimizing the variational parameters, to be able to maneuver the cost function landscape successfully. For profit formulations, since every state vector represents a feasible solution, one only needs to optimize the circuit's variational parameters without the need to set appropriate penalties.


Fig.~\ref{fig:vc-both} shows one  instance of probability distribution with $p\in \{1, 2, 3\}$ layers for the five node graph shown in the inset of the graph, run on PennyLane's standard simulator (\texttt{default.qubit}). Since there are three maximum profit covers, we see a higher probability for the states $01000$, $01100$, and $11000$ than all other states. While $01100$ and $11000$ are minimum vertex covers, the state $01000$ (i.e., $\textit{VC} = \{1\}$) is not a vertex cover but can be converted into a minimum vertex cover (i.e., $\textit{VC} = \{1, 0\}$ or $\textit{VC}=\{1, 2\}$) of size $2$ (Alg.~\ref{alg:add_vertices}). At $p=3$, the vertex cover formulation demonstrates improved outcomes.
Note that the summed probabilities for the profit covers ($43.1\%$) of the largest size are higher than the summed probabilities for the minimum vertex covers ($35.5\%$). 


For the same problem instance, the data points in burgundy in Fig.~\ref{fig:vc-both} shows the results of post-processing profit cover using Alg.~\ref{alg:add_vertices}. Recall that every bit string for profit cover is feasible, whereas only a subset of the bit strings are feasible for vertex cover. Therefore, we run Alg.~\ref{alg:add_vertices} on every bit string to convert each of those into feasible vertex covers and add up the probabilities. 
Fig.~\ref{fig:vc-is-cl-summed-probs} shows the summed probabilities over eight layers for constrained versions and profit formulations for varying edge probabilities averaged over ten graphs. 
We observe a trend of higher summed probabilities of the optimal solution with for the profit formulations than the constrained formulations. From Fig.~\ref{fig:vc-is-cl-summed-probs}, one can see that for optimal constrained problem solutions to catch up to the same summed probability value as profit cover, on average, it takes additional QAOA layers. Fig.~\ref{fig:vc-opt-summed-probs} shows the summed probabilities for near-optimal solutions. The near-optimal analysis shows a significant increase in the probability values for the profit formulations as compared to the constrained formulation. Due to the unconstrained nature of the profit formulations, many infeasible solutions are also candidate near-optimal solutions as compared to the constrained formulation. Note that while solutions may be infeasible, for profit versions, they can be easily converted to a feasible solution with no change in profit using classical post-processing. While the same classical post-processing routine can be used for the constrained formulations, it does not hold the same value as the constrained formulations promote feasible solutions over infeasible solutions. 

We observe a dip in the summed probabilities for both profit and constrained versions of vertex cover and independent set for dense graphs. A few dips in the summed probabilities over layers can be attributed to the classical optimizer becoming trapped at local minima. Further investigation into the cost function landscape of profit cover problems is warranted here. 

\begin{figure}[!th]
    \centering
    \begin{subfigure}[b]{0.50\textwidth} 
        \centering
        \includegraphics[width=\textwidth]{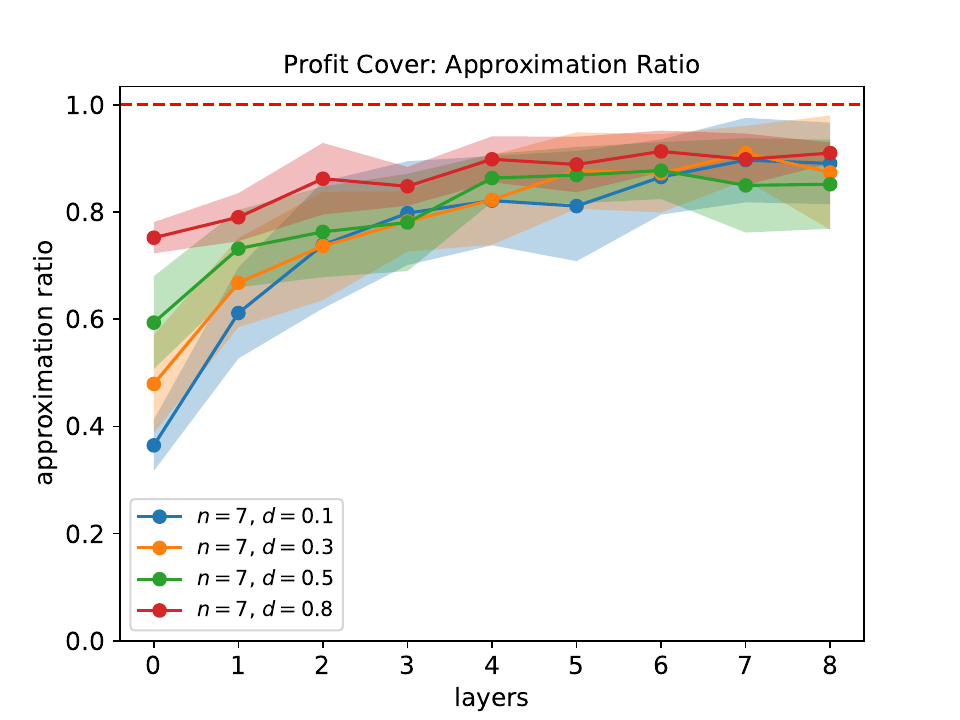}
        \caption{Approximation ratios for graphs with $n=7$}
        \label{fig:subfiga}
    \end{subfigure}
    \hfill
    \begin{subfigure}[b]{0.50\textwidth} 
        \centering
        \includegraphics[width=\textwidth]{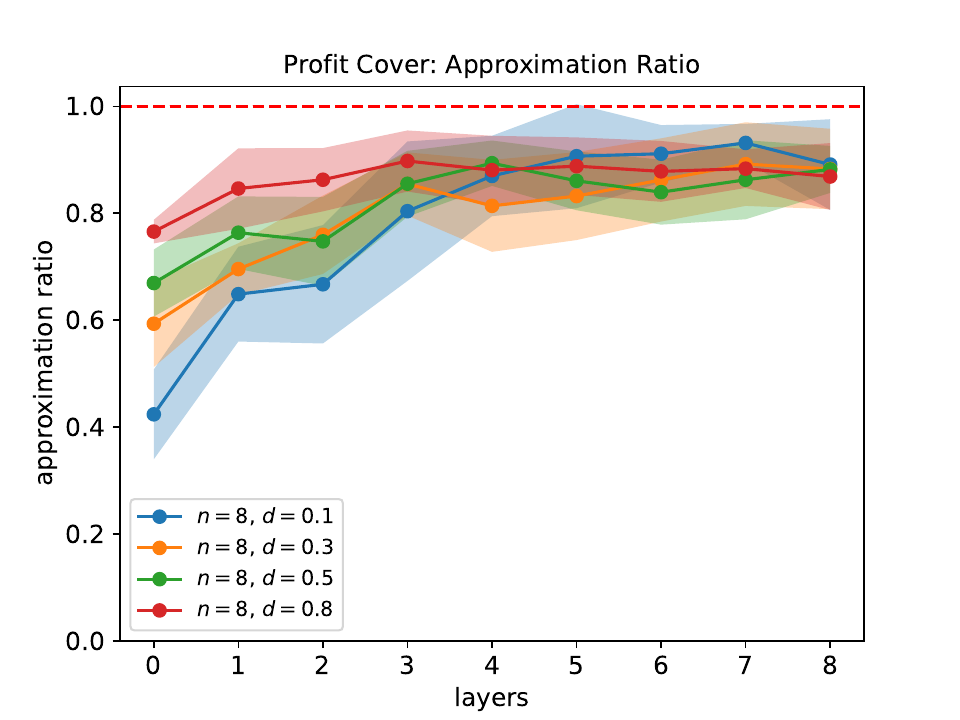}
        \caption{Approximation ratios for graphs with $n=8$}
        \label{fig:subfigb}
    \end{subfigure}
    \caption{Approximation ratios computed for \textit{\textsc{MaxPC}} on PennyLane. The rate of change of approximation ratio is higher for sparse graphs compared to denser graphs. Denser graphs have more near-optimal solutions compared to sparse graphs, therefore, start off with a high approximation ratios but have a smaller rate of change of the approximation ratio.}
    \label{fig:approx-ratios}
\end{figure}
\begin{figure}[!htb]

        \centering
        \includegraphics[width=\columnwidth]{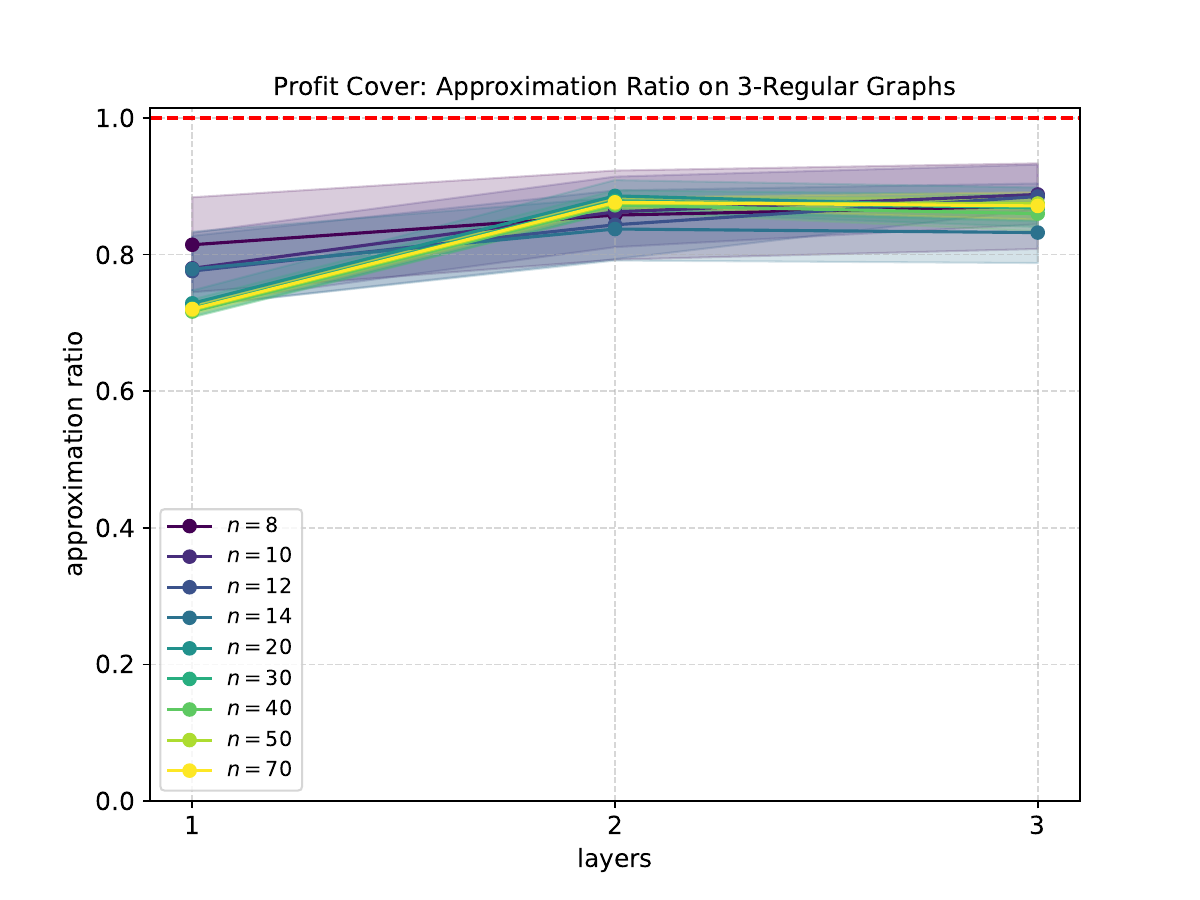}
        
    \caption{Approximation ratios computed for \textit{\textsc{MaxPC}} on PennyLane for $n=\{8, 10, 12, 14\}$ and QTensor for $n=\{20, 30, 40, 50, 70\}$ showing scalability of our approach.}
    \label{fig:approx-ratios-2}
\end{figure}

We can see the same trend between  results obtained in the summed probabilities, as depicted in Fig.\ref{fig:vc-is-cl-summed-probs}, and the approximation ratios obtained (i.e., the increase up to layer 2 and the dip in solution quality between layers 3 and 5). 

Fig.~\ref{fig:penalty-plot} and Fig.~\ref{fig:vc-both} show the effect of using different penalties for the constrained formulation of the \textit{\textsc{MinVC}} problem. In both figures, it is clear that not one penalty value is better than another for finding optimal solutions using QAOA. The performance of penalties is also not consistent over layers; for example, in the top left corner sub-figure of Fig.~\ref{fig:penalty-plot}, the summed optimal probability of QAOA with penalty parameters $A=6, B=2$ does not perform well on average in the initial layers, but improves after layer 5. Additionally, the sub-figure on the bottom left corner shows that QAOA with penalty parameters of $A=6, B=2$ does not perform well for graphs with edge density of $d=0.5$.

Fig.~\ref{fig:cost-qtensor} shows the cost function evaluations of \textit{\textsc{MaxPC}} and \textit{\textsc{MinVC}} for a single random  3-regular graph with $n \in \{30, 40, 50, 70\}$. Fig~\ref{fig:vc-is-20-sparse-5layers} shows \textit{\textsc{MaxPC}}, \textit{\textsc{MinVC}}, \textit{\textsc{MaxPI}} and \textit{\textsc{MaxIS}} results for a single sparse graph (edge density, $d=0.1$) We observe fluctuations in the early parts of the optimization for all four problems due to the adaptive learning rate of RMSProp. The red dotted line represents the exact solution to the corresponding problems. While QAOA tends to get stuck in local minima, it approaches the exact solution more closely in \textit{\textsc{MaxPC}} compared to \textit{\textsc{MinVC}}, due to the presence of more optimal and near-optimal solutions in \textit{\textsc{MaxPC}} than \textit{\textsc{MinVC}}. We further remark that in the case of constrained optimization (\textit{\textsc{MinVC}}) while the expectation value indicates how close the current solution is to the optimal one in terms of the objective, it does not give information about whether the constraints are satisfied. However, this is not the case for the unconstrained \textit{\textsc{MaxPC}}. For example, the top left graph of Fig.~\ref{fig:cost-qtensor} shows that the expectation value achieved with \textit{\textsc{MaxPC}} is in between $-24$ and $-26$, indicating that the optimal profit is at least $24$ for this particular graph. 
Given a total of $|E|$ edges for the graph, we can conclude that a vertex cover of size at most $|VC| \leq |E| - 24$ exists for this graph.

We demonstrate the scalability of our approach by comparing the approximation ratio achieved on \textit{\textsc{MaxPC}} for $p = 1, 2, 3$ layers, which shows that our approach can achieve an approximation ratio greater than $0.8$ for graphs of varying sizes, ranging from $n = 8$ to $n = 70$.
\section{Conclusion and Future Work}\label{sec:future}
Solution landscapes of constrained optimization problems include infeasible solutions, which can make finding optimum or good feasible solutions challenging. To formulate QUBOs of constrained optimization problems, one introduces penalty-parameters that---by themselves---can be challenging to optimize and add an extra hyper-parameter to QAOA. 

We demonstrate how the relaxation of constraints can eliminate the infeasible solutions from the cost function landscape. In Alg.~\ref{alg:add_vertices}, we describe how solutions to the relaxation of  \textit{\textsc{MinVC}}, i.e.~to {\sc MaxPC}, can be translated into feasible vertex cover solutions \textit{of the same quality}.  

 \textit{\textsc{MinVC}} comes with applications including conflict resolution and scheduling~\cite{Stege2000}; \textit{\textsc{MaxPC}} adds in applicability~\cite{van2008tractable} making the investigation of scalable solutions of this particular NP-complete problem important. 
With \textit{\textsc{MaxPC}}, we add an interesting problem to the pool of unconstrained optimization problems. For the particular case of solving \textit{\textsc{MinVC}}, \textit{\textsc{MaxPC}} opens the door to a penalty-free QUBO or Ising model of the problem for using the vanilla QAOA, reducing the complexity of parameter optimization. Note that for the experiments reported here, we did not employ a problem-specific mixer or a special state preparation (such as those used in the Quantum Alternating Operator Ansatz~\cite{hadfield2019quantum}). These approaches make use of additional gates and qubit interactions that come at a high cost for near-term quantum computing hardware.  

In this article, we evaluated \textit{\textsc{MaxPC}} on PennyLane and Argonne QTensor for small graphs up to eight layers as well as large sparse graphs with less than four QAOA layers. We used probabilities, summed optimal and near-optimal probabilities, expectation values, and approximation ratios as metrics for evaluation. While the summed probability metric has been used sparingly in the literature~\cite{Saleem2020}, we believe it provides valuable insights into the quality of the optimal and near-optimal solutions generated by our approach (cf. Section~\ref{subsec:near-opt}). We think it deserves more attention and should be considered a key measure of performance.
We plan to apply our constraint relaxation methodology to a broader family of constrained combinatorial optimization problems.
We will also investigate the cost function landscape of \textit{\textsc{MaxPC}} and contrast different classical optimizers, warm-starting, additional layers, mixer variants, as well as other QAOA variants~\cite{chalupnik2022augmenting, egger2021warm, chandarana2022digitized, herrman2022multi} to achieve better solution quality on different kinds of graphs. For Argonne QTensor, our objective is to incorporate a collection of constrained optimization problems suitable for lightcone simulation. By integrating these problems into the QTensor ecosystem, researchers can test and validate their methods for constrained optimization problems at scale. We would also like to run profit cover experiments on actual quantum hardware to evaluate the viability of the profit formulation through the use of OpenQAOA~\cite{sharma2022openqaoa}, a multi-backend SDK for quantum optimization that allows for seamless circuit execution for variational quantum algorithms (specifically for QAOA and its variants). Lastly, we would like to investigate more ways of ``unconstraining" constrained problems. The relationship between \textit{\textsc{MinVC}} and its relaxation \textit{\textsc{MaxPC}} in particular, and the relaxation of constraints of constrained combinatorial optimization problems in general, constitutes significant progress within the quantum community as it shows new ways of thinking about constrained problems through a different lens. We posit  that our approach offers  more effective ways for solving many constrained optimization problems and as such opens an avenue toward quantum utility.

\bibliographystyle{IEEEtran}
\bibliography{IEEEabrv, qaoafriendly}

\section{Appendix}
\label{sec:appendix}
\subsection{Classical results}
\label{sec:classical-results}
The challenges of coping with intractability have been extensively studied in theoretical computer science~\cite{karp1972reducibility, downey1995fixed, vazirani2001approximation, hromkovivc2013algorithmics}.

While there is a $2$-approximation algorithm for \textit{\textsc{MinVC}} that is based on the idea of  a maximal matching, when assuming that the Unique Games Conjecture holds, \textit{\textsc{MinVC}} 
is hard to approximate within a factor of smaller than $2$~\cite{garey1979computers}. Generally, \textit{\textsc{MinVC}} is known to be APX-complete~\cite{papadimitriou1988optimization}.  
The best known approximation factor for \textit{\textsc{MinVC}} is $2-\Theta\left(1/\sqrt{\log{|V|}}\right)$~\cite{Karakostas2009}. 
%
 Moreover, the decision version\footnote{The decision version of \textit{\textsc{Minimum Vertex Cover}} of \textit{\textsc{Minimum Vertex Cover}} asks, when given $G=(V,E)$ and a positive integer $k$, whether or not there exists a subset $\textit{VC}\subseteq V$ of size $|VC|\leq k$.} is fixed-parameter tractable when parameterized by the size of the vertex cover to be determined~\cite{DFS99}; its fastest known fixed-parameter algorithm has a time complexity of $O^*\!\left(1.25284^k\right)$ \cite{harris2024faster} where $k$ is the size of the desired solution. 

\textit{\textsc{MaxIS}} is APX-hard and Poly-APX-complete for general graphs~\cite{bazgan2005completeness}.
Approximation algorithms with constant approximation factor exist for bounded-degree graphs~\cite{berman1994approximating, de2023stable}. 
In contrast to the fixed-parameter tractability result for the decision version of \textit{\textsc{Minimum Vertex Cover}}, the decision version of \textit{\textsc{Maximum Independent Set}} for general graphs is $W[1]$-complete when parameterized by the size of the independent set to be determined, but fixed-parameter tractable for planar graphs~\cite{downey2013parameterized}. Inapproximability results state that \textit{\textsc{MaxIS}} cannot be approximated within a factor of $n^{1-\epsilon}$ for any $\epsilon > 0$, unless P = NP~\cite{haastad1999clique, zuckerman2006linear}. Exact algorithms for \textit{\textsc{MaxIS}} include Robson’s $O(2^{n/4})$ algorithm~\cite{robson2001finding} and more recently a $O(1.1664^n)$-algorithm  by Xiao and Nagamochi~\cite{xiao2017exact}.  

Due to their close relationship, the above findings for \textit{\textsc{Maximum Independent Set}}  transfer to  \textit{\textsc{Maximum Clique}}.

\subsection{Related Work on QAOA}
\label{sec:related-work}
%
Several QAOA components are tunable, including initialization, ansatz construction, parameter choice, and the classical optimization method. These components significantly influence the algorithm's efficacy and can be tailored to specific problem instances for improved performance.  Hadfield et al.~\cite{hadfield2019quantum} extend QAOA by introducing the Quantum Alternating Operator Ansatz (QAO-Ansatz), which alternates between more general families of unitary operators. This can potentially narrow the algorithm's focus to a more useful set of states. The QAO-Ansatz can be used to guarantee that the state of the circuit never leaves the set of feasible states. However, the circuit is composed of complicated circuitry. For example, multi-controlled Toffoli gates are often used in the ansatz, which are challenging to execute depending on the connectivity of qubits on a quantum computer~\cite{he2017decompositions}. 
Golden et al.~\cite{golden2023numerical} compare the performance of different variations of QAOA on three problems, Max Bisection, Max $k$-Vertex Cover, and $k$-Densest Subgraph, using different kinds of mixers and show a possibility of achieving a super-polynomial advantage over Grover's unstructured search. The problem Max $k$-Vertex Cover has also been studied by~\cite{Cook_2020, Bartschi_2020}, and while it is related to \textit{\textsc{MinVC}}, is not as complex, due to a deterministic quantum algorithm to prepare Dicke states~\cite{bartschi2019deterministic} that allows for exploration only in the feasible solution subspace of Max $k$-Vertex Cover. These problems are either unconstrained or constrained by the Hamming weight and, therefore, do not need penalties in the cost function.  Pelofske et al.~\cite{pelofske2019solving, pelofske2023solving} propose a recursive classical decomposition of large problems (as a pre-processing and pruning step) such that they can be solved on quantum annealers. Saleem et al.~\cite{SaleemTTS23, tomesh2023divide} consider penalty-term approaches, QAO-Ansatz, and introduce a new ansatz variant that adapts to the quantum resources available for the \textit{\textsc{MaxIS}}. 

Other notable variants of QAOA includes QAOA+~\cite{chalupnik2022augmenting}, WS-QAOA~\cite{egger2021warm}, Digitized Counterdiabatic QAOA~\cite{chandarana2022digitized}, and multi-angle QAOA~\cite{herrman2022multi}. For QAOA+, the authors propose adding a problem-independent layer of parameterized $Z\!Z$-gates and the parameterized $X$-gates (as a mixer) to improve the approximation ratio of \textit{\textsc{MaxCut}}, which is an unconstrained optimization problem. In WS-QAOA, warm starting methods are used to prepare an initial state that corresponds to the solution of a relaxation of the portfolio optimization problem. In multi-angle QAOA, a parameter is assigned to every element of the ansatz to improve the approximation ratio achieved for \textit{\textsc{MaxCut}}. Digitized Counterdiabatic QAOA appends a layer (and therefore, an additional variational parameter) to perform counterdiabatic driving to converge to the optimal solution faster (thereby requiring a shorter circuit depth). 

In literature, the performance of QAOA has been typically studied for Erd\H{o}s-R\'enyi Random Graphs with different probabilities of edge connections~\cite{golden2023numerical, SaleemTTS23} as well as bounded-degree graphs (specifically random 3-regular graphs~\cite{shaydulin2023qaoawith, lykov2021performance}). Herrman et al.~\cite{herrman2021impact} perform a detailed analysis  across different graph structures (up to eight nodes)  to understand QAOA \textit{\textsc{MaxCut}} performance of up to three layers. 

\subsection{Near-optimal analysis on Independent Sets}
\begin{figure*}[!bth]
\centering
  \begin{tabular}{@{}cccc@{}}
    \includegraphics[width=0.33\textwidth]{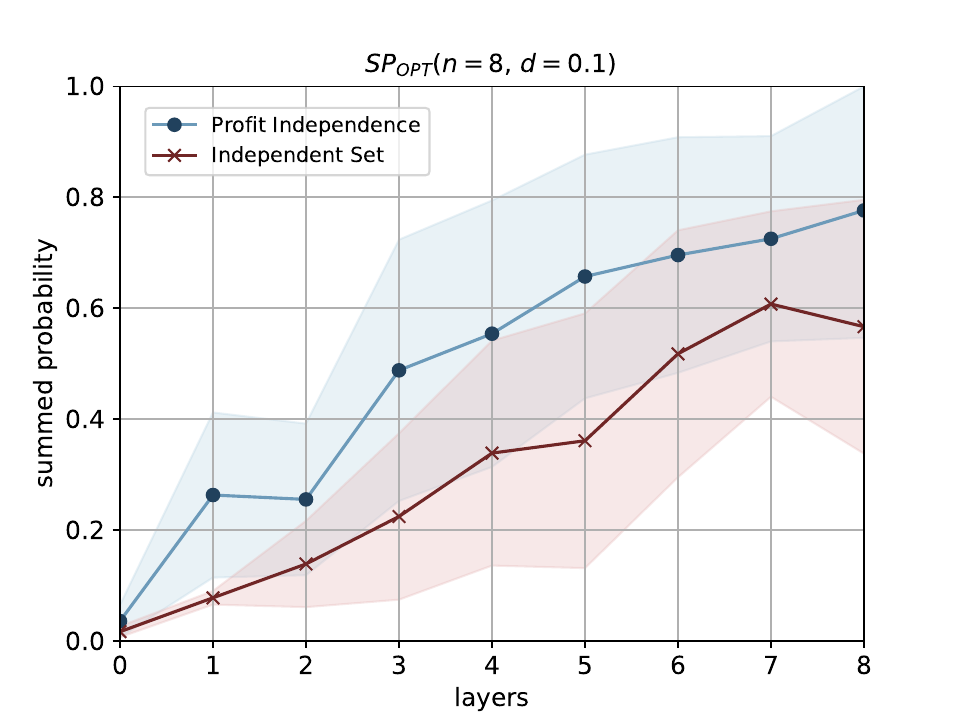} &
    \includegraphics[width=.33\textwidth]{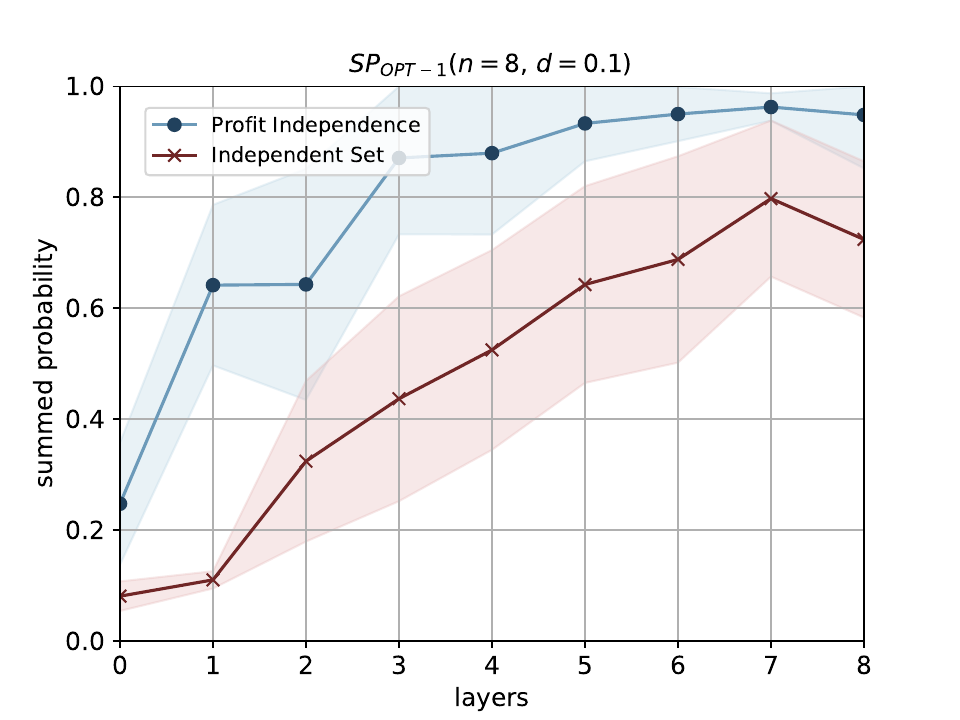} &
    \includegraphics[width=.33\textwidth]{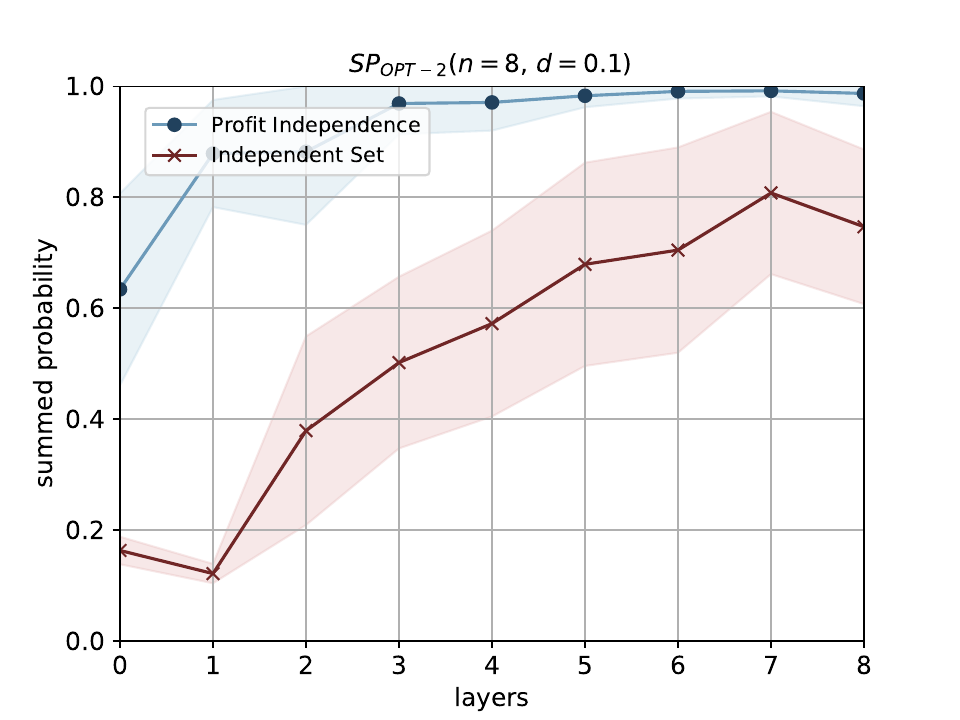} &
    \\
    \includegraphics[width=0.33\textwidth]{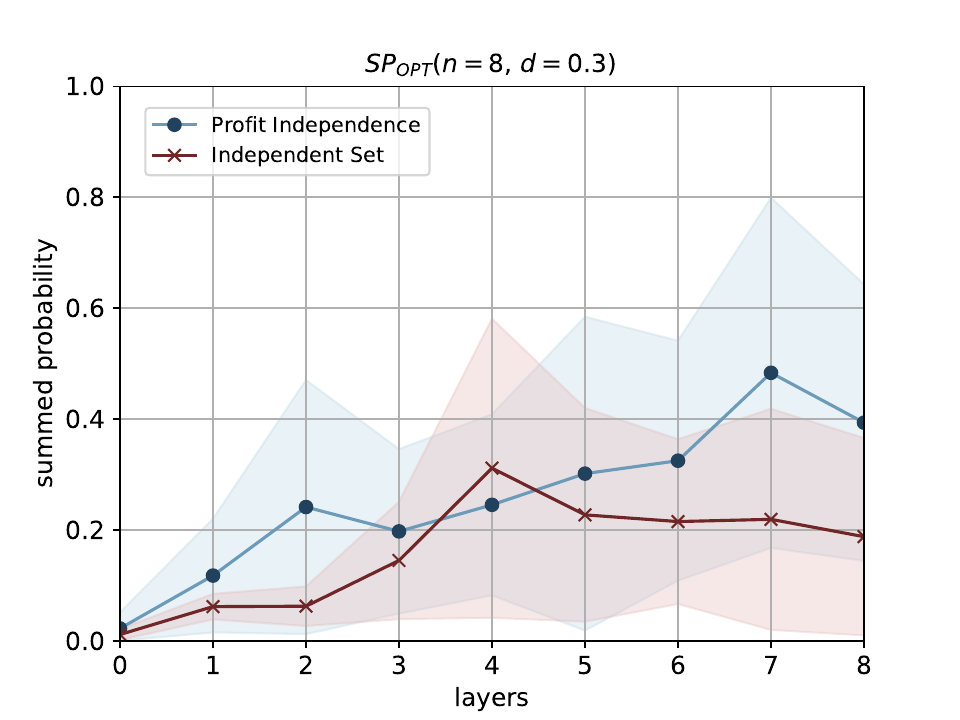} &
    \includegraphics[width=.33\textwidth]{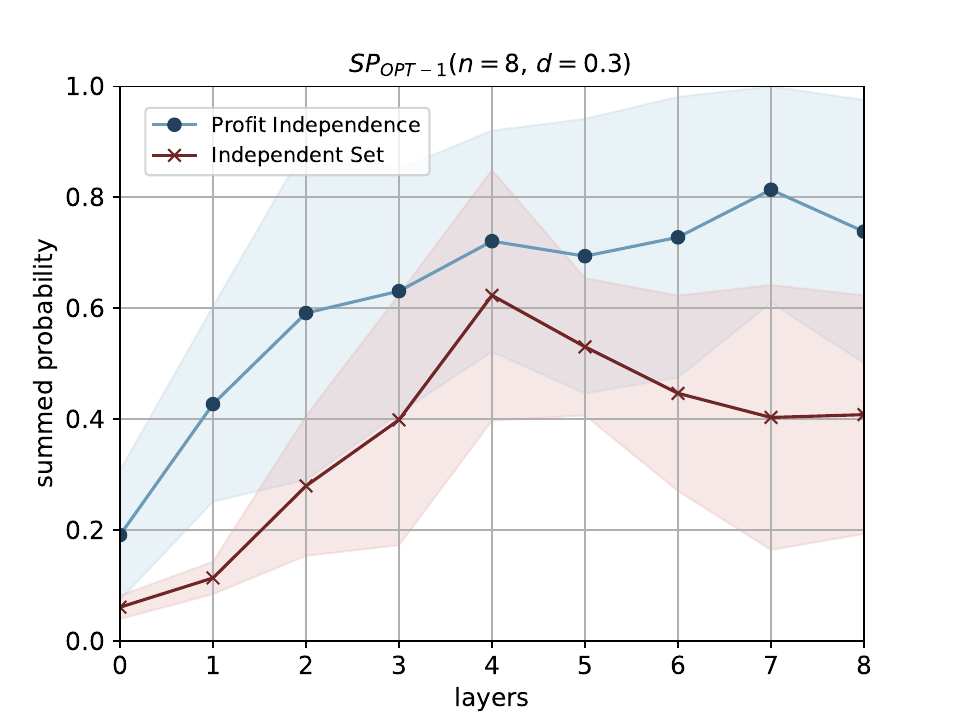} &
    \includegraphics[width=.33\textwidth]{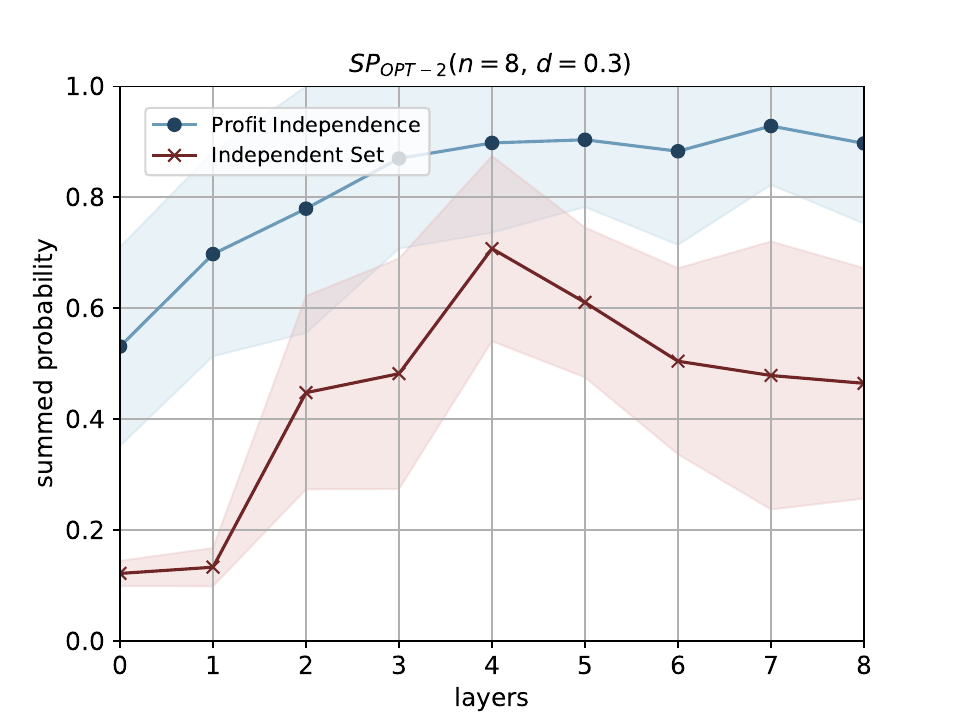} &
    \\
    \includegraphics[width=0.33\textwidth]{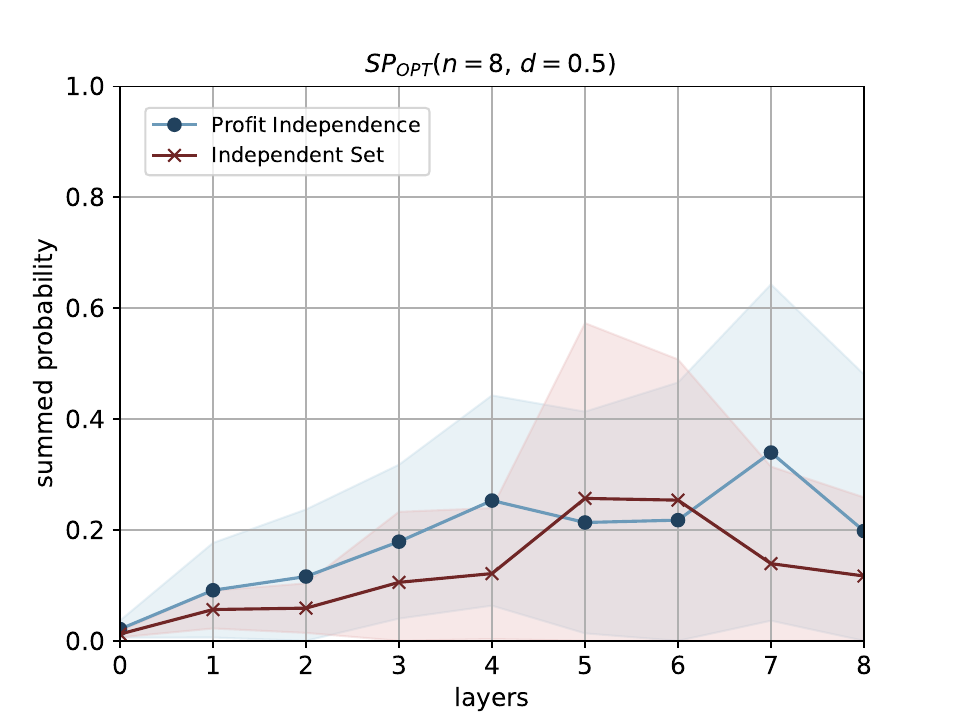} &
    \includegraphics[width=.33\textwidth]{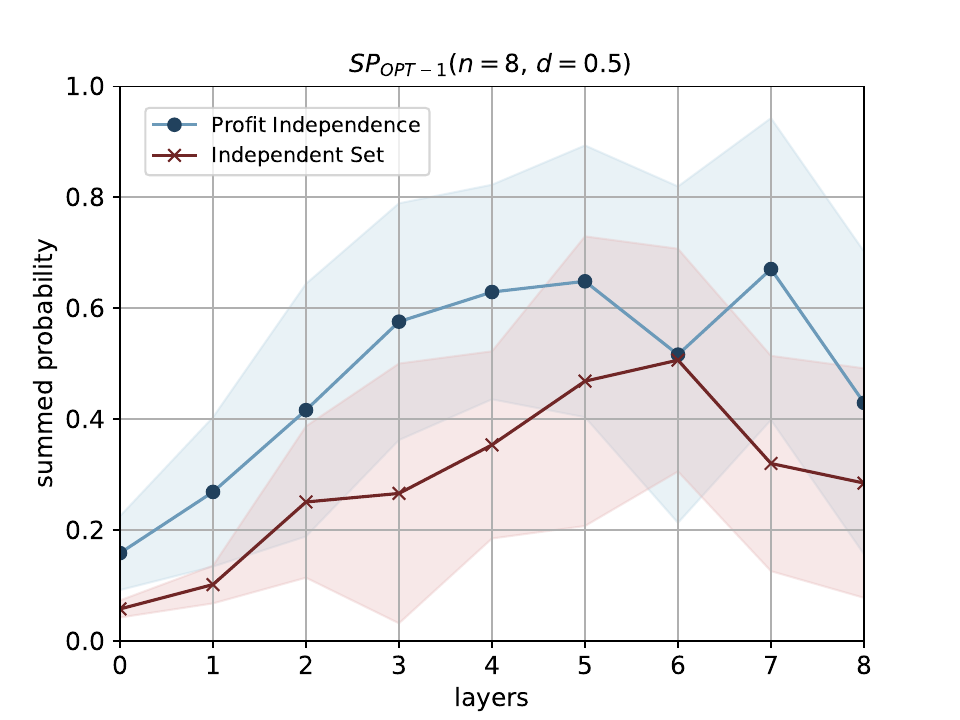} &
    \includegraphics[width=.33\textwidth]{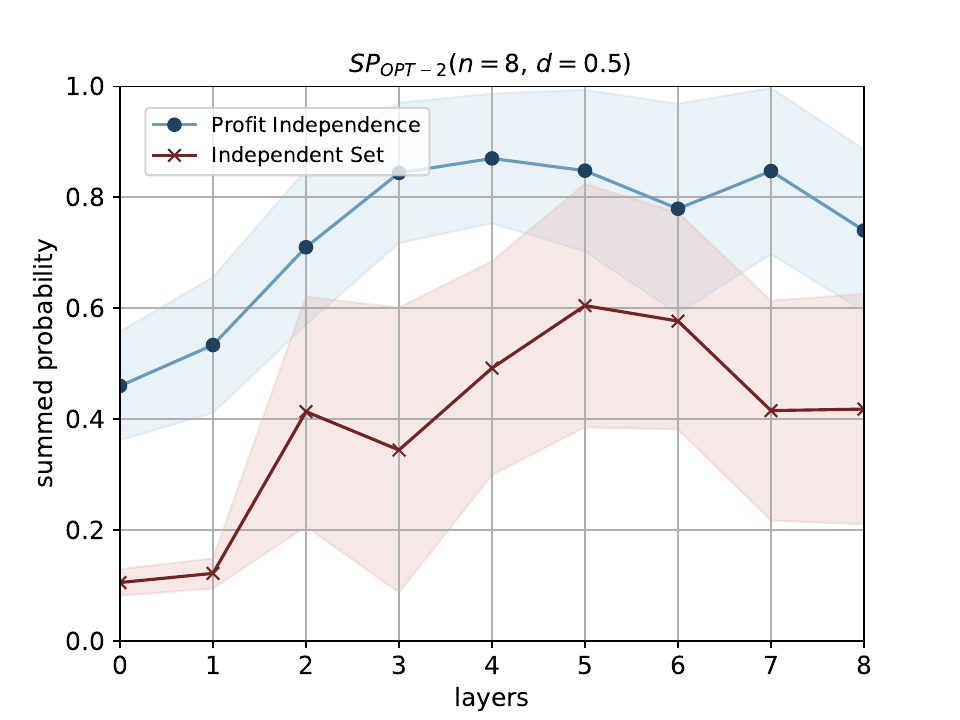} &
    \\
    \includegraphics[width=0.33\textwidth]{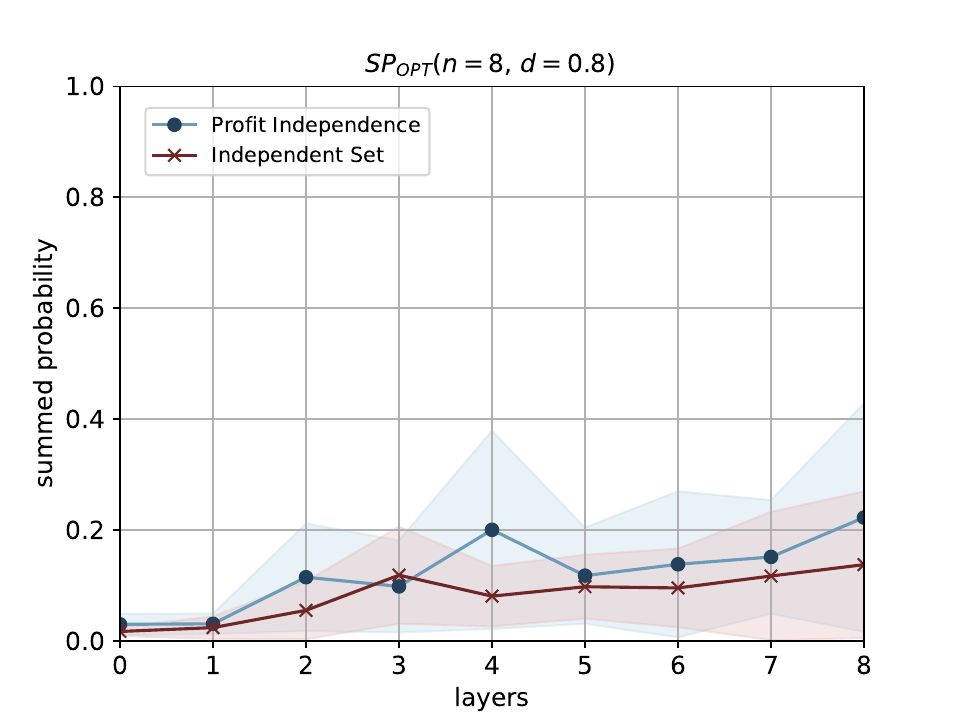} &
    \includegraphics[width=.33\textwidth]{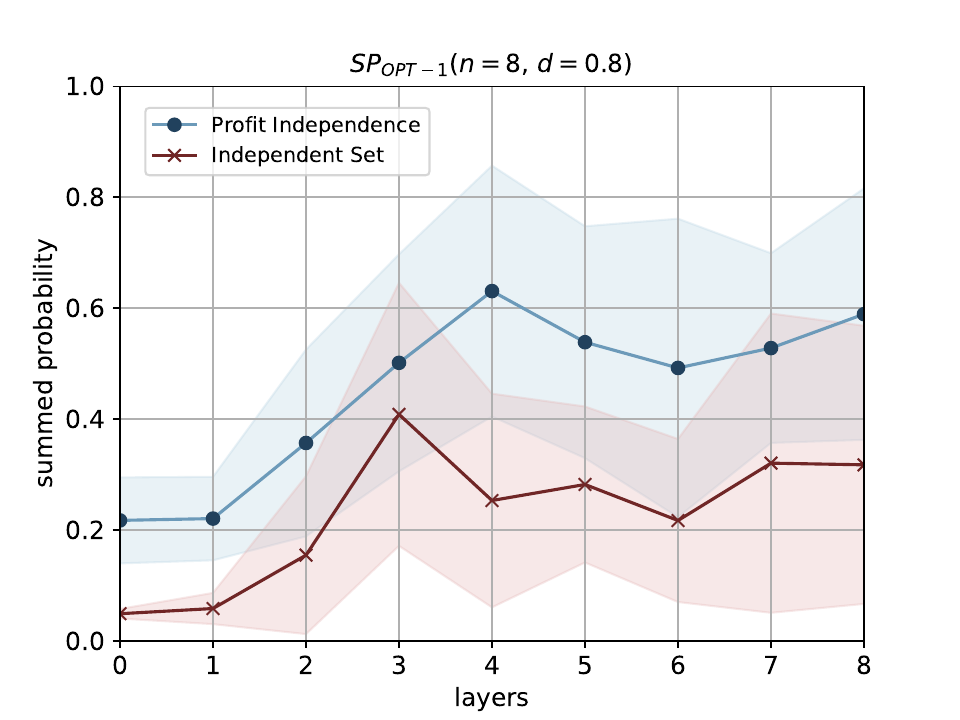} &
    \includegraphics[width=.33\textwidth]{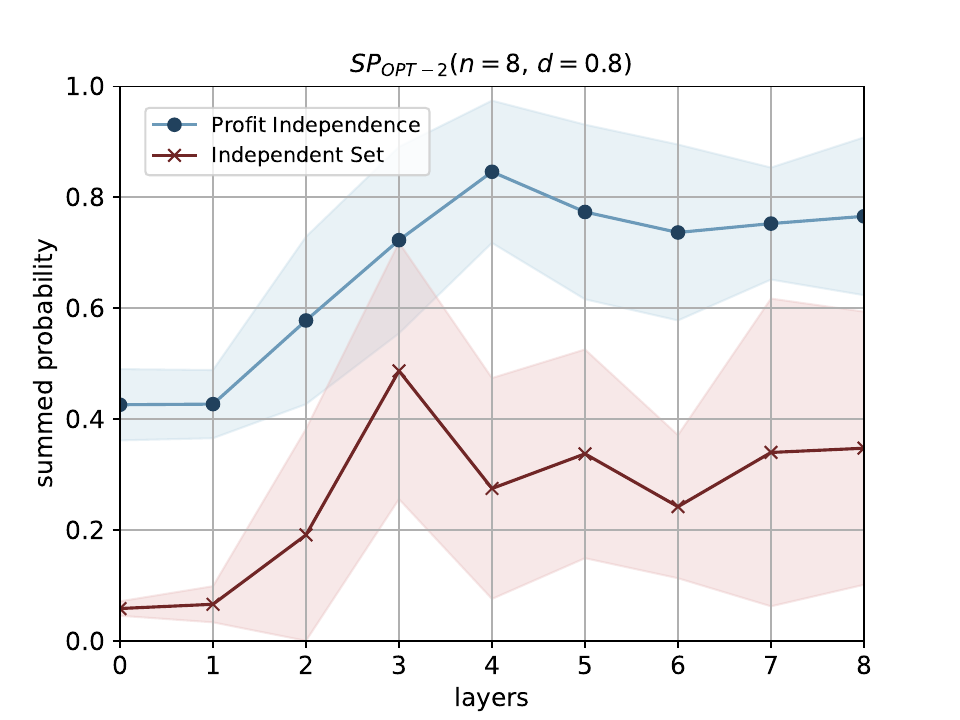} &
    \\
  \end{tabular}
  \caption{Probabilities of optimal and near-optimal solutions obtained over eight layers for \textit{\textsc{MaxIS} and \textit{\textsc{MaxPI}}}. The first column shows the summed probability of all the optimal solutions averaged over ten graphs with $n = 8$. The second column depicts the summed probability of obtaining the optimal solution and the second best solution. The third column indicates the summed probability of the optimal, second best and the third best solutions. Each row indicates a different edge density }
  \label{fig:sp-pennylane-ispi}
\end{figure*}
Fig. \ref{fig:sp-pennylane-ispi} is analogous to Fig. \ref{fig:vc-opt-summed-probs} (which shows summed optimal and near-optimal probability results for vertex cover), but applied to the \textit{\textsc{MaxIS}} and \textit{\textsc{MaxPI}} problems.

\subsection{Comparison to Penalty-term formulation}

\begin{figure*}[!htbp]
\centering
  \begin{tabular}{@{}cccc@{}}
    \includegraphics[width=0.5\textwidth]{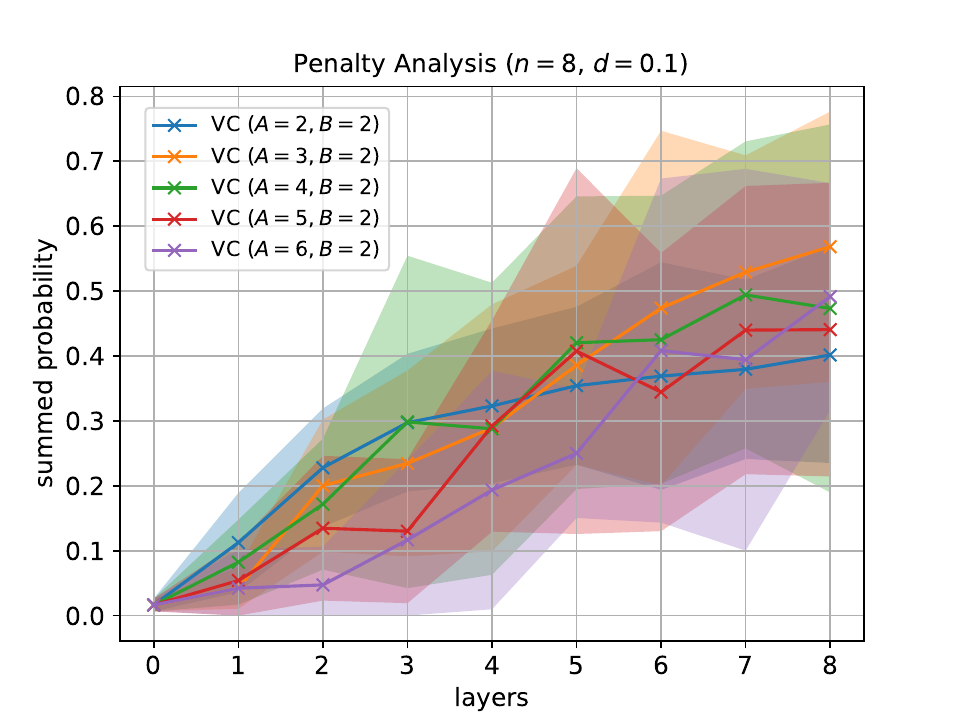} &
    \includegraphics[width=.5\textwidth]{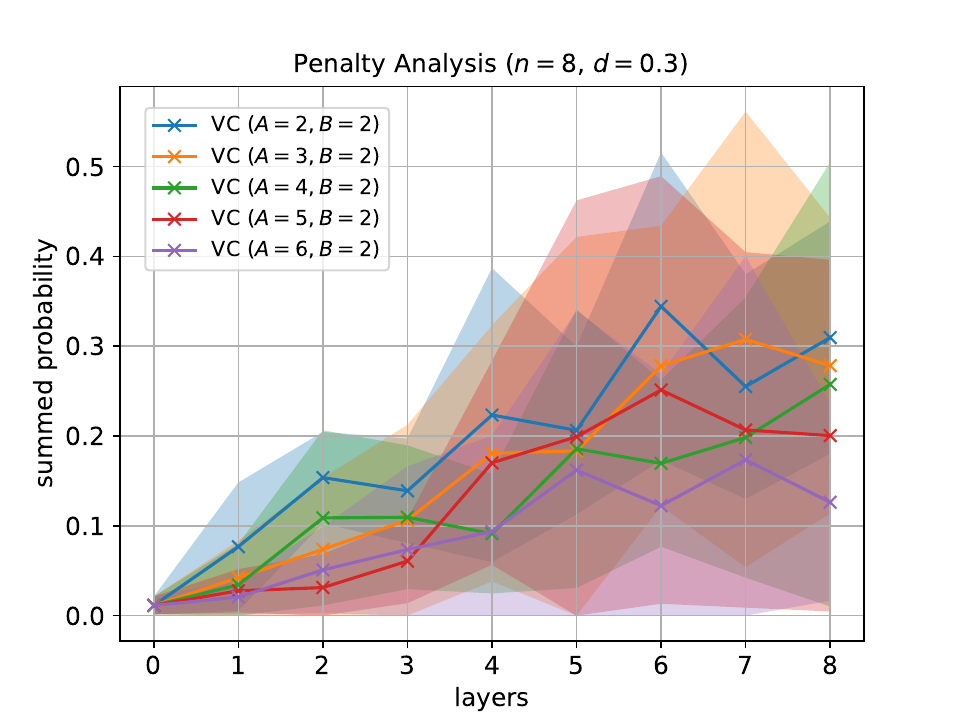} \\
    \includegraphics[width=0.5\textwidth]{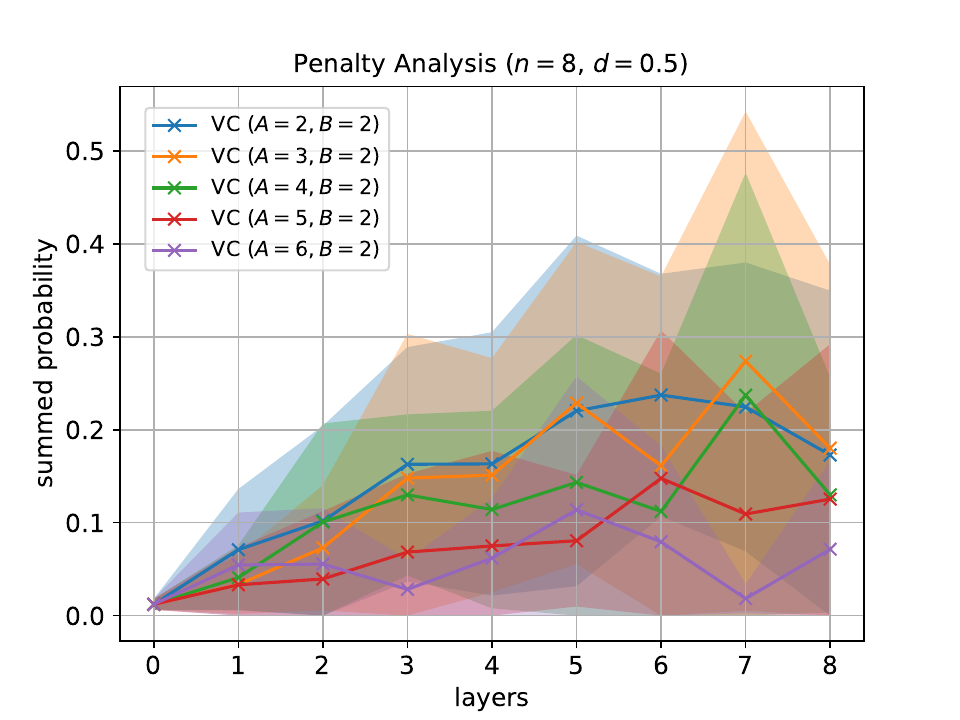} &
    \includegraphics[width=.5\textwidth]{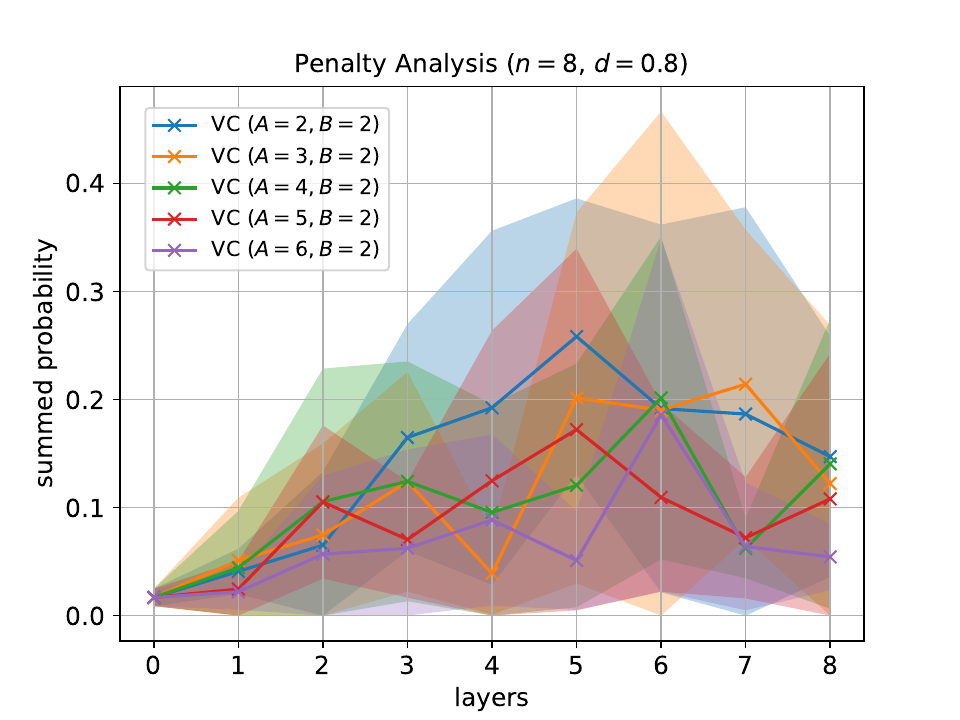} \\
  \end{tabular}
  \caption{The plots above show summed probabilities for constrained optimization problems using penalty term formulations for $d\in \{0.1, 0.3, 0.5, 0.8\}$ for graphs with eight nodes averaged over ten graphs. It is evident that no single penalty value consistently outperforms the others across all graph types. The optimal penalty varies depending on the specific characteristics of each graph. Even within a single class of graphs, certain penalties may perform better for some graphs while being less effective for others.}
  \label{fig:penalty-plot}
\end{figure*}
Fig.~\ref{fig:penalty-plot} shows four summed probability plots for eight nodes with varying edge probabilities ($d\in \{0.1, 0.3, 0.5, 0.8\}$) and varying penalty parameters. Penalties are varied with different values for the term $A$, keeping $B = 2$ constant.

\section*{Acknowledgment}

We thank Faisal Abu-Khzam for the support in implementing classical vertex cover algorithms. 

\end{document}